\documentclass[showversion,USenglish,authorcolumns,numberwithinsect,footnotebackref,a4paper]{no-lipics} 

\usepackage{bm}
\usepackage{stmaryrd}
\usepackage{mathtools}
\usetikzlibrary{arrows,arrows.meta,shapes.misc,decorations.pathreplacing,decorations.pathmorphing,calligraphy, patterns.meta}
\usetikzlibrary{calc}

\usepackage{accents}

\newcommand{\A}{\mathcal{A}}

\SetKwInput{KwInput}{Input}
\SetKwInput{KwOutput}{Output}

\def\twoheadleadsto{\tikz[baseline=(a.base)]{\draw[%
    decorate,decoration={zigzag,segment length=4, amplitude=.9},%
    ] (0,0) -- (.25, 0);%
    \draw[%
    -{Classical TikZ Rightarrow}.{Classical TikZ Rightarrow},%
    ] (.25, 0) -- (.4, 0);%
    \node (a) at (.4/2,-.03) {\phantom{\(\leadsto\)}};%
}}

\newcommand{\onto}{\twoheadleadsto}

\makeatletter
\newcommand{\douwidehat}[2]{%
  \sbox0{$\m@th#1\widehat{\hphantom{#2}}\vphantom{t}$}%
  \sbox2{$t$}%
  \dimen2=\ht0
  \advance\dimen2 -\ht2
  \sbox2{$#2$}%
  \dimen0=\ht0
  \rlap{%
    \raisebox{\dimexpr-\dimen0-\dp2-1pt}[0pt][\dimexpr\dimen2+\dp2]{\box0}%
  }
  {#2}%
}
\makeatother

\newcommand{\Oh}{\mathcal{O}}
\newcommand{\Ohtilde}{\tilde{\Oh}}

\DeclareMathOperator*{\argmin}{arg\,min}
\DeclareMathOperator*{\argmax}{arg\,max}

\def\fragmentco#1#2{\bm{[}\,#1\,\bm{.\,.}\,#2\,\bm{)}}
\def\fragmentoc#1#2{\bm{(}\,#1\,\bm{.\,.}\,#2\,\bm{]}}
\def\fragmentoo#1#2{\bm{(}\,#1\,\bm{.\,.}\,#2\,\bm{)}}
\def\fragment#1#2{\bm{[}\,#1\,\bm{.\,.}\,#2\,\bm{]}}

\def\position#1{\bm{[}\,#1\,\bm{]}}

\newcommand{\ceil}[1]{\lceil #1 \rceil}
\newcommand{\floor}[1]{\lfloor #1 \rfloor}

\newcommand{\per}{\operatorname{per}}

\newcommand{\ed}{\delta_E}

\newcommand{\edp}[2]{{\delta_E}(#1,#2^*)}
\newcommand{\edl}[2]{{\delta_E}(#1,{}^*\!#2^*)}
\newcommand{\eds}[2]{{\delta_E}(#1,{}^*\!#2)}
\newcommand{\edal}[1]{\delta_E^{#1}}
\newcommand\selfed{\mathsf{self}\text{-}\ed}
\newcommand{\OccE}{\mathrm{Occ}^E}

\newcommand{\Occ}{\mathrm{Occ}}

\newcommand{\last}{\mathrm{last}}

\newcommand{\mA}{\mathcal{A}}
\newcommand{\mB}{\mathcal{B}}
\newcommand{\mX}{\mathcal{X}}
\newcommand{\mY}{\mathcal{Y}}
\newcommand{\mZ}{\mathcal{Z}}

\newcommand{\cost}{\operatorname{cost}}
\newcommand{\w}{\operatorname{w}}

\newcommand{\pref}{\mathsf{pref}}
\newcommand{\suf}{\mathsf{suf}}
\newcommand{\mXpref}{\mX_{\mathsf{pref}}}
\newcommand{\mXsuf}{\mX_{\mathsf{suf}}}
\newcommand{\mXmid}{\mX_{\mathsf{mid}}}

\newcommand{\Z}{\mathbb{Z}}
\newcommand{\Zz}{\mathbb{Z}_{\ge 0}}

\newcommand{\hi}{\hat{\imath}}
\newcommand{\hj}{\hat{\jmath}}
\newcommand{\hx}{\hat{x}}
\newcommand{\hy}{\hat{y}}

\newcommand{\sE}{\mathsf{E}}

\newcommand{\bG}{\mathbf{G}}

\newcommand{\bc}{\operatorname{bc}}
\newcommand{\cc}{\operatorname{cc}}
\newcommand\Tau{\mathcal{T}}

\makeatletter
\renewenvironment{cases}{%
    \matrix@check\cases\env@cases
}{%
    \endarray\right.%
}
\def\env@cases{%
    \let\@ifnextchar\new@ifnextchar
    \left\lbrace
        \def\arraystretch{1.1}%
        \array{@{\;}c@{\quad}l@{}}%
    }
    \makeatother

\def\mid{\ensuremath :}

\def\emptyset{\varnothing}

\newcommand{\zero}{\mathtt{0}}

\def\pn#1{\textsc{#1}}
\def\PMwE{\pn{Pattern Matching with Edits}\xspace}
\def\PMwM{\pn{Pattern Matching with Mismatches}\xspace}

\title{The Communication Complexity of~Pattern~Matching~with~Edits Revisited}

\author{Tomasz Kociumaka}{Max Planck Institute for Informatics,\\Saarland Informatics
Campus,\\Saarbrücken, Germany}{tomasz.kociumaka@mpi-inf.mpg.de}{https://orcid.org/0000-0002-2477-1702}{}
\author{Jakob Nogler}{Massachusetts Institute of Technology,\\Cambridge, United States}{jnogler@mit.edu}{https://orcid.org/0009-0002-7028-2595}{}
\author{Philip Wellnitz}{National Institute of Informatics,\\The Graduate University for Advanced Studies, SOKENDAI\\Tokyo, Japan}{wellnitz@nii.ac.jp}{https://orcid.org/0000-0002-6482-8478}{}

\authorrunning{T. Kociumaka, J. Nogler, and P. Wellnitz}

\begin{document}
\pagenumbering{roman}
\maketitle
\begin{abstract}
    In the decades-old \PMwE problem,
    given a length-$n$ string \(T\) (the text), a~length-$m$ string \(P\) (the pattern),
    and a positive integer \(k\) (the threshold),
    the task is to list the $k$-error occurrences of $P$ in $T$, that is, all
    fragments of \(T\) whose edit distance to~\(P\) is at most \(k\).
    The one-way communication complexity of \PMwE is the minimum number of bits
    that Alice, given an instance $(P,T,k)$ of the problem, must send to Bob so that Bob
    can reconstruct the answer solely from that message.

    For the natural parameter regime of $0<k<m<n/2$, our recent work~[STOC'24] yields that
    $\Omega(n/m \cdot k \log(m/k))$ bits are necessary and $\Oh(n/m \cdot k \log^2 m)$
    bits are sufficient for \PMwE.
    More generally, for strings over an alphabet~$\Sigma$, our recent work [STOC’24] gives
    an $\Oh(n/m \cdot k \log m \log (m|\Sigma|))$-bit encoding that allows one to recover
    a shortest sequence of edits for every $k$-error occurrence of $P$ in~$T$.

    In this work, we revisit the original proof and improve the encoding size to $\Oh(n/m
    \cdot k \log (m|\Sigma|/k))$, which matches the lower bound for constant-sized
    alphabets.
    We further establish a new tight lower bound of $\Omega(n/m \cdot k
    \log(m|\Sigma|/k))$ for the edit sequence reporting variant that we solve.
    Our encoding size also matches the communication complexity established for the
    simpler \PMwM problem in the context of streaming
    algorithms [Clifford, Kociumaka, Porat; SODA'19].
\end{abstract}

\tableofcontents

\clearpage
\pagenumbering{arabic}
\section{Introduction}

In the more-than-classic \PMwE problem (PMWE)~\cite{S80},
we are given a text $T$ of length $n$, a pattern $P$ of length $m$, and a threshold $k$,
and the task is to compute the starting positions of all substrings of $T$ whose edit
distance to $P$ is at most $k$.
Formally, we aim to compute the set
\[\OccE_k (P,T) \coloneqq \{i \in \fragment{0}{n} \mid \exists_{j \in \fragment{i}{n}} \:
\ed(P, T\fragmentco{i}{j}) \leq k\},\]
where $\ed(\star, \star)$ denotes the edit (Levenshtein) distance~\cite{Levenshtein66}.
This metric quantifies the (dis)similarity between strings by counting the minimum number
of single-character edits (insertions, deletions, and substitutions) required to transform
one string into the other.

Research into this problem has expanded far beyond the classical
setting~\cite{S80,LV88,LandauV89,ColeH98,CKW20,CKW22}, with modern variants exploring
compressed~\cite{GS13,Tis14,BLRSSW15,CKW20}, dynamic~\cite{CKW20},
streaming~\cite{Sta17,KPS21,BK23}, weighted~\cite{CKW25}, and quantum~\cite{KNW24,KNW25}
settings.
We recently investigated the problem's one-way communication complexity~\cite{KNW24},
which involves determining the minimum space required to encode an instance so that the
set of $k$-error occurrences can be reconstructed without further access to the original
input.

The \emph{communication complexity} framework models this as a two-party game: Alice, who
holds the problem instance, transmits a single message to Bob.
Bob's task is to derive the full output solely from this message. Because Bob lacks any
initial input of his own, the problem reduces to a one-way single-round protocol, and the
primary objective is to minimize the number of bits Alice must transmit.

For the natural parameter regime of $0<k<m<n/2$,
our work \cite{KNW24} establishes that
$\Omega(n/m \cdot k \log(m/k))$ bits are necessary for this problem, while $\Oh(n/m \cdot
k \log^2 m)$ bits are sufficient.
More generally, we showed that if the shortest sequence of edits must be recovered for
each
$k$-error occurrence, the upper bound increases to $\Oh(n/m \cdot k \log m \log
(m|\Sigma|))$, where $\Sigma$ is the input alphabet.
The former result is implied by the latter by performing a simple alphabet reduction: all
characters of $T$ not present in $P$ are mapped to a single character.
The result of \cite{KNW24} is significant for three reasons.
\begin{enumerate}
    \item A central open question~\cite{ColeH98,CKW22} for \PMwE
        is whether a static algorithm can match the $\Oh(n + n/m \cdot
        k^{2-\Omega(1)})$ conditional lower bound \cite{bi18}. The current
        state-of-the-art algorithm by Charalampopoulos, Kociumaka, and Wellnitz~\cite{CKW22} runs in
        $\Ohtilde(n + n/m\cdot k^{3.5})$ time, and relies on structural
        results~\cite{CKW20} that allow it to return the required output $\OccE_k (P,T)$
        using $\Ohtilde(n/m \cdot k^3)$ bits. This encoding is fundamentally incompatible
        with reaching the lower bound;
        however, the results of~\cite{KNW24} demonstrate that this barrier can indeed be
        overcome.

    \item A similar question was previously settled for the simpler \PMwM (PMWM)~\cite{CKP19};
        here $\Omega(n/m \cdot k \log(m/k))$ bits are necessary and $\Oh(n/m \cdot
        k \log (m|\Sigma|/k))$ bits suffice to also report the mismatches for all $k$-mismatch
        occurrences.
        Thus, by \cite{KNW24} and up to essentially an $\Oh(\log
        m)$ factor,
        the edit distance setting behaves the same as the mismatch setting in
        terms of \mbox{information density}.

    \item Said result of~\cite{KNW24} also paves the way for better algorithms in other
        computational models (for instance, the quantum algorithms for PMWE of~\cite{KNW24,KNW25},
        which crucially rely on this result). The prior work
        on PMWM~\cite{CKP19} suggests that one-way
        communication complexity is a stepping stone toward $\Ohtilde(k)$-space
        algorithms in the streaming and semi-streaming models.
        This is particularly relevant given the concurrent development of $\Oh(k\cdot
        n^{o(1)})$-size edit distance sketches~\cite{KS24}. Although the state-of-the-art
        sketches do not yet offer the space-efficient sketch construction needed for a
        full streaming implementation, the advances in~\cite{KNW24} represent significant
        steps toward~it.
\end{enumerate}

\subparagraph*{Our Results.}
We re-examine the encoding provided in~\cite{KNW24} and, through a series
of refinements, eliminate the $\Theta(\log m)$ term separating the upper and lower bounds.
Thereby, we place PMWE on par with PMWM. Formally, we show the following result.

\begin{restatable*}{mtheorem}{ccompl}\label{thm:ccompl}
    Fix integers $n,m,k$ with $n \ge m \ge k > 0$ and an input alphabet $\Sigma$.
    The \PMwE problem admits a one-way deterministic communication
    protocol that sends $\Oh(n/m\cdot k \log(m|\Sigma|/k))$ bits.

    Within the same communication complexity, one can also reconstruct the set of all
    fragments $T\fragmentco{i}{j}$ that satisfy $\ed(P,T\fragmentco{i}{j})\le k$ and, for
    each such fragment, the edit information (the positions and values of edited
    characters) of all optimal edit sequences transforming $P$ into the fragment
    $T\fragmentco{i}{j}$.
\end{restatable*}

\begin{remark}
   Our result extends routinely beyond $n \ge m \ge k>0$.
   If $n<m$, it suffices to apply \cref{thm:ccompl} with the text $T$ replaced by
   $\zero^{m-n} \cdot T$, where $\zero\in \Sigma$ is an arbitrary character.

   Applying \cref{thm:ccompl} with threshold $k=m$ allows one to reconstruct the entire
   pattern from the edit sequence for an empty fragment of the text and the entire text
   from the edit sequences for single-character fragments; this covers $k>m$.
   Finally, the case $k=0$ is no harder than $k=1$.
   Thus, in full generality, the communication complexity becomes
   \[\Oh\left(\max\!\left(1,{n}/{m}\right)\cdot \min(m,k+1) \cdot
   \log({m|\Sigma|}/{\min(m,k+1)})\right).
    \qedhere\]
\end{remark}

Through the aforementioned alphabet reduction, we also achieve $\Oh(n/m \cdot k \log m)$
bits for the variant where only the positions (but not the values) of edited characters
are reported.
This suffices for the baseline problem asking for the starting positions of $k$-error
occurrences, which means that we match the lower bound of~\cite{KNW24} unless both $k\ge
m^{1-o(1)}$ and $|\Sigma|\ge \omega(1)$.

It is not difficult to show that a single optimal edit sequence for a single $k$-error
occurrence can be encoded using $\Oh(k\log (m|\Sigma|/k))$ bits. While the previous
approach in~\cite{KNW24} required storing information proportional to $\Oh(\log m)$ edit
sequences for each length-$\Theta(m)$ block of the text, our improvement achieves the
information content equivalent to $\Oh(1)$ edit sequences.
Hence, in a $\Theta(m)$-length block of~$T$, the cost of encoding all $k$-error
occurrences and their respective edit sequences is asymptotically no greater than the cost
of encoding one!

We also establish a new lower bound proving that our protocol is optimal for edit
retrieval.

\begin{restatable*}{mtheorem}{cclb}\label{thm:cclb}
    Fix integers $n,m,k$ with $n \ge m \ge k > 0$ and an input alphabet $\Sigma$
    with $\zero \in \Sigma$ and \(|\Sigma| \ge 2\).
    Consider an encoding that, for each fragment $T\fragmentco{i}{j}$
    with $\ed(P,T\fragmentco{i}{j})\le k$, allows to reconstruct an edit sequence $P \onto
    T\fragmentco{i}{j}$ and the corresponding edit information.

    Such an encoding must use $\Omega(n/m\cdot k\log(|\Sigma| m/k))$ bits for $P=\zero^m$ and
    some $T\in \Sigma^n$.
\end{restatable*}

\subparagraph*{Our Techniques.}
On a very high level, in \cite{KNW24}, the text $T$ is first partitioned into $\Oh(n/m)$
partially overlapping
blocks, each of length $\Oh(m)$ and such that each $k$-error occurrence of $P$ in $T$ is
fully contained in some block. For each such block, the encoding contains
\begin{enumerate}
    \item a set $S$ of $\Oh(\log m)$ many $k$-error occurrences of $P$ in $T$, each stored
        together with an optimal edit sequence using $\Oh(k \log (m|\Sigma|/k))$ bits;
        \label{it:stored:1}
    \item\label{it:stored:2} the Lempel--Ziv LZ77~\cite{DBLP:journals/tit/ZivL77}
        compressed representation of a collection of fragments of $T$ selected based on
        $S$. The total number of LZ77 phrases is proportional to the total cost of
        $k$-error occurrences in $S$, so they can be encoded using $\Oh(k\log m \log
        (m|\Sigma|/k))$ bits.
\end{enumerate}
Overall, this yields an encoding cost of $\Oh(k \log m \log (m |\Sigma|/k))$ bits per
block and $\Oh(n/m \cdot k \log m \cdot \log (m |\Sigma|/k))$ bits for the entire text.
The underlying construction is relatively complex, relying on recent insights relating
edit distance to compressibility~\cite{CKW23,GJKT24} via the so-called \emph{self-edit
distance}.
Nevertheless, our key improvement is simple to describe at a high level.
In~\eqref{it:stored:1}, we partition the $i$-th $k$-error occurrence into $\Theta(2^i)$
pieces and store the edit information only for the ``cheapest'' piece.
This incurs only $k_i = \Oh(k/2^i)$ edits, for a total of $\Oh(k)$ instead of $\Oh(k\log m)$.
The fragments in~\eqref{it:stored:2} are defined essentially the same way and, thanks to
the decreased total cost in \eqref{it:stored:1}, they can be encoded in $\Oh(k\log
(m|\Sigma|/k))$ bits.

Generalizing the relevant concepts to capture pieces of $P$ requires several delicate
steps, which we discuss in \cref{sec:ub}.
Moreover, our approach breaks for $k = o(\log m)$ since $0$-error pieces still need a
$\Theta(\log m)$-bit representation, which becomes a bottleneck when $|S|\ge \omega(k)$.
In that case, through the structural characterization of Charalampopoulos, Kociumaka, and
Wellnitz~\cite{CKW20}, the existence of $\omega(k)$ occurrences implies that $P$ and the
relevant part of $T$ are at edit distance $\Oh(k)$ from highly periodic strings.
We then use this rigid structure to show that just three $\Oh(k)$-error occurrences of $P$
in $T$ can play the role of $S$ in \eqref{it:stored:1}.

\section{Preliminaries}\label{sec:prelims}

\subsubsection*{Sets}
For integers $i,j \in \mathbb{Z}$, we write $\fragment{i}{j}$ to denote the set $\{i,
\dots, j\}$ and $\fragmentco{i}{j}$ to denote the set $\{i,\dots, j-1\}$; we define the
sets $\fragmentoc{i}{j}$ and $\fragmentoo{i}{j}$ similarly.

For a set $S$ of integers and a parameter $k > 0$, we also define $\floor{S/k} \coloneqq
\{ \floor{s/k} \mid s \in S \}$.

\subsubsection*{Strings}

An \emph{alphabet} \(\Sigma\) is a set of characters.
We write $X=X\position{0}\, X\position{1}\cdots X\position{n-1} \in \Sigma^{n}$ to denote
a \textit{string} of length $|X|=n$ over $\Sigma$.
For a \emph{position} $i \in \fragmentco{0}{n}$, we say that $X\position{i}$ is the $i$-th
character of $X$.
A string $Y$ is a \emph{substring} of another string $X$ if $Y=X\position{i} \cdots
X\position{j-1}$ holds for some integers $0 \le i \le j \le |X|$.
In this case, we say that there is an (\emph{exact}) \emph{occurrence} of~$Y$ starting at
position $i$ in~$X$.
The occurrence is a \emph{fragment} of $X$ denoted $X\fragmentco{i}{j}$; formally, the
fragment can be interpreted as a tuple $(X,i,j)$ consisting of a (reference to) $X$ and
the two positions $i$ and $j$.
We may also write $X\fragment{i}{j-1}$, $X\fragmentoc{i-1}{j-1}$, or
$X\fragmentoo{i-1}{j}$ for the fragment $X\fragmentco{i}{j}$.
A \emph{prefix} of~a string $X$ is a fragment of the form $X\fragmentco{0}{j}$, and a
\emph{suffix} of~a string $X$ is a fragment of the form $X\fragmentco{i}{|X|}$.

For two strings $A$ and $B$, we write $AB$ for their concatenation.
We write $A^k$ for the concatenation of $k\in \Zz$ copies of the string $A$.
Moreover, $A^\infty$ is an infinite string (indexed with non-negative integers) formed as
the concatenation of an infinite number of copies of \(A\).

An integer $p\in \fragment{1}{n}$ is \emph{a period} of a string $X\in\Sigma^n$ if we have
$X\position{i} = X\position{i + p}$ for all $i \in \fragmentco{0}{n-p}$.
In this case, we also say that the string $X\fragmentco{0}{p}$ is a \emph{string period}
of~$X$.
In other words, a string $P$ is a string period of a string $X$ if $X$ is a prefix of
$P^\infty$ and $|P|\le |X|$.
\emph{The period} of a non-empty string \(X\), denoted $\per(X)$, is the smallest period
of \(X\).
A non-empty string \(X\) is \textit{periodic} if \(\per(X) \le |X| / 2\).

\subsubsection*{Edit Distance and Alignments}
The \emph{edit distance} (\emph{Levenshtein distance}~\cite{Levenshtein66}) between two
strings $X$ and $Y$, denoted by $\ed(X,Y)$, is the minimum number of character insertions,
deletions, and substitutions required to transform $X$ into~$Y$.
Formally, we first define an \emph{alignment} between string fragments.
\begin{definition}[{\cite[Definition 2.1]{CKW22}}]\label{def:alignment}
    A sequence $\mA=(x_i,y_i)_{i=0}^{\ell}$ is an \emph{alignment} of
    $X\fragmentco{x}{x'}$ onto $Y\fragmentco{y}{y'}$, denoted by \(\mA:
    X\fragmentco{x}{x'} \onto Y\fragmentco{y}{y'}\), if it satisfies $(x_0,y_0)=(x,y)$,
    $(x_{i+1},y_{i+1})\in \{(x_{i}+1,\allowbreak y_{i}+1),(x_{i}+1,\allowbreak
    y_{i}),(x_{i},y_{i}+1)\}$ for $i\in \fragmentco{0}{\ell}$, and $(x_\ell,y_\ell)
    =(x',y')$. Moreover, for $i\in \fragmentco{0}{\ell}$
    \begin{itemize}
        \item if $(x_{i+1},y_{i+1})=(x_{i}+1,y_{i})$, we say that $\mA$ \emph{deletes}
            $X\position{x_i}$;
        \item if $(x_{i+1},y_{i+1})=(x_{i},y_{i}+1)$, we say that $\mA$ \emph{inserts}
            $Y\position{y_i}$;
        \item if $(x_{i+1},y_{i+1})=(x_{i}+1,y_{i}+1)$, we say that $\mA$ \emph{aligns}
            $X\position{x_i}$ to $Y\position{y_i}$.
        If also $X\position{x_i}=Y\position{y_i}$, then $\mA$ \emph{matches}
        $X\position{x_i}$ and $Y\position{y_i}$; otherwise, $\mA$ \emph{substitutes}
        $X\position{x_i}$ with~$Y\position{y_i}$.
        \qedhere
    \end{itemize}
\end{definition}

The \emph{cost} of an alignment $\mA$ of $X\fragmentco{x}{x'}$ onto $Y\fragmentco{y}{y'}$
is the total number of characters that $\mA$ inserts, deletes, or substitutes;
we denote the cost of an alignment \(\mA\) by $\cost(\mA)$ or
$\edal{\mA}(X\fragmentco{x}{x'},Y\fragmentco{y}{y'})$.
The edit distance $\ed(X,Y)$ is the minimum cost of an alignment of $X\fragmentco{0}{|X|}$
onto~$Y\fragmentco{0}{|Y|}$.
An alignment of $X$ onto $Y$ is \emph{optimal} if its cost is equal to $\ed(X, Y)$.

Given an alignment $\A:X\fragmentco{x}{x'}\onto Y\fragmentco{y}{y'}$ and a fragment
$X\fragmentco{\bar{x}}{\bar{x}'}$ contained in $X\fragmentco{x}{x'}$, we write
$\A(X\fragmentco{\bar{x}}{\bar{x}'})$ for the fragment $Y\fragmentco{\bar{y}}{\bar{y}'}$
of $Y\fragmentco{y}{y'}$ that $\A$ aligns against $X\fragmentco{\bar{x}}{\bar{x}'}$.
As insertions and deletions may render this definition ambiguous, we set
\[\bar{y} \coloneqq \min\{\hat{y} : (\bar{x},\hat{y})\in \A\}\quad\text{and}\quad
    \bar{y}' \coloneqq \left\{\begin{array}{c l}
            y' & \text{if }\bar{x}' = x',\\
            \min\{\hat{y}' : (\bar{x}',\hat{y}')\in \A\} & \text{otherwise}.
    \end{array}\right.
\]
This particular choice satisfies the following decomposition property.
\begin{lemmaq}[{\cite[Fact 2.2]{CKW22}}]\label{fct:ali}
    For an alignment $\A: X \onto Y$ and a decomposition $X=X_1\cdots X_t$ into $t$
    fragments, $Y=\A(X_1)\cdots \A(X_t)$ is a decomposition into $t$ fragments with
    $\ed^\A(X,Y)= \sum_{i=1}^t \ed^{\A}(X_i,\A(X_i))$.
    If \(\A\) is optimal, then $\ed(X,Y) = \sum_{i=1}^t \ed(X_i,\A(X_i))$.
\end{lemmaq}

We use the following \emph{edit information} notion to encode alignments.

\begin{definition}[Edit information of an alignment]\label{def:edinfo}
    For an alignment
    \[(x_i,y_i)_{i=0}^\ell\!=\mA : X\fragmentco{x}{x'}\onto Y\fragmentco{y}{y'},\]
    the \emph{edit information} is the set of 4-tuples $\sE_{X,
        Y}(\mA)=\{(x_i,\mathsf{cx}_i \mid y_i,\mathsf{cy}_i) : i\in
    \fragmentco{0}{\ell}\text{ and }\mathsf{cx}_i\ne \mathsf{cy}_i\}$, where
    \[\mathsf{cx}_i = \begin{cases}
        X\position{x_i} & \text{if }x_{i+1}=x_i+1,\\
        \varepsilon & \text{otherwise};
    \end{cases}\qquad\text{and}\qquad
    \mathsf{cy}_i = \begin{cases}
        Y\position{y_i} & \text{if }y_{i+1}=y_i+1,\\
        \varepsilon & \text{otherwise}.
    \end{cases}
    \]
    By monotonicity of the alignment, the order of the tuples $(x_i,\mathsf{cx}_i \mid
    y_i,\mathsf{cy}_i)$ with respect to $i$
    coincides with their lexicographic order by $(x_i,y_i)$ and by $(y_i,x_i)$.
\end{definition}

Observe that, given two strings \(X\) and \(Y\), the endpoints \(x,x',y,y'\), and the edit
information \(\sE_{X, Y}(\mA)\) of an alignment \(\mA : X\fragmentco{x}{x'} \onto
Y\fragmentco{y}{y'}\), we are able to fully reconstruct \(\mA\). Indeed, each tuple of
\(\sE_{X, Y}(\mA)\) specifies a non-matching operation, whereas all pairs of \(\mA\)
before the first tuple, between consecutive tuples, and after the last tuple must be
matches.
Moreover, given the characters of \(X\fragmentco{x}{x'}\) and the edit information
\(\sE_{X,Y}(\mA)\), one can recover the characters of
\(Y\fragmentco{y}{y'}\). Indeed, the tuples of \(\sE_{X,Y}(\mA)\) reveal all
characters of \(Y\fragmentco{y}{y'}\) created by insertions and substitutions, while
the remaining characters of \(Y\fragmentco{y}{y'}\) correspond to matches and can
therefore be copied from \(X\fragmentco{x}{x'}\).

\subsubsection*{Pattern Matching with Edits}

In the context of two strings $P$ (referred to as the pattern) and $T$ (referred to as the
text), along with a positive integer $k$ (referred to as the threshold), we say that
$T\fragmentco{t}{t'}$ is a \emph{$k$-error occurrence} of $P$ in $T$ if $\ed(P,
T\fragmentco{t}{t'})\leq k$ holds.

\section{Upper Bound: Proof Overview}\label{sec:ub}

In this section, we provide an overview of the proof of \cref{thm:ccompl}.
To this end, we fix two strings $P \in \Sigma^m$ and $T \in \Sigma^n$, and a positive
threshold $k$. Throughout this section, we assume that $k = o(m)$, $n \leq 2m-2k$, and
that $P$ has $k$-error occurrences as a prefix and as a suffix of~$T$.
In \cite[{Claim 4.31\footnote{All references to numbers of theorems, lemmas, and
definitions of \cite{KNW24} refer to their full version.}}]{KNW24}, a standard
block-splitting argument is employed to demonstrate that a protocol using $\Oh(k \log m
\log(m|\Sigma|))$ bits for this specific case is sufficient to achieve $\Oh(n/m\cdot k
\log m \log(m|\Sigma|))$ bits in general.
We improve the former communication bound to $\Oh(k \log(m|\Sigma|/k))$ bits and adapt
\cite[Claim 4.31]{KNW24} as \cref{clm:split} to drop the assumptions.

The remainder of this section is structured as follows. In \cref{subsec:prev_overview}, we
provide a concise overview of the proof from \cite{KNW24}, focusing specifically on the
components we adapt to achieve our improvements. Subsequently, in \cref{subsec:improving},
we outline our improved construction.

\subsection{Recap: A Brief Overview of the Previous Encoding}
\label{subsec:prev_overview}

\subsubsection*{\boldmath The Graph $\bG_S$ and the Induced Periodic Structure}

The main ingredient of the encoding of~\cite{KNW24} is a set $S$ of $\Oh(\log m)$
alignments of $P$ onto fragments of $T$ with cost at most $k$ each.
The starting point of \cite{KNW24} is to analyze how much information such a set of
alignments carries, in order to better understand how one should choose which alignments
to include and what additional information is needed to fully encode $\OccE_k (P,T)$.
To enable this analysis, the set $S$ is associated with a graph $\bG_S$, which we call the
\emph{inference graph}.

\begin{restatable}[{\cite[Definition 4.1]{KNW24}}]{definition}{bg}\label{def:bg}
    Let $S$ be a set of alignments $\mX : P\fragmentco{p}{p'}\onto T\fragmentco{t}{t'}$.
    We define the undirected graph $\bG_{S} = (V, E)$ as follows.
    The vertex set $V$ contains
    \begin{enumerate}
        \item $|P|$ vertices representing characters of $P$;
        \item $|T|$ vertices representing characters of $T$; and
        \item one special vertex $\bot$.
    \end{enumerate}
    The edge set $E$ contains the following edges for each alignment $\mX\in S$.
    \begin{enumerate}
        \item $\{P\position{x},\bot\}$ for every character $P\position{x}$ that $\mX$ deletes;\label{it:bg:i}
        \item $\{\bot, T\position{y}\}$ for every character $T\position{y}$ that $\mX$ inserts;\label{it:bg:ii}
        \item $\{P\position{x},T\position{y}\}$ for every pair of characters
            $P\position{x}$ and $T\position{y}$ that $\mX$ aligns.\label{it:bg:iii}
    \end{enumerate}
    We say that an edge $\{P\position{x},T\position{y}\}$ is \emph{black} if $\mX$ matches $P\position{x}$ and $T\position{y}$. All other edges are \emph{red}.

    A connected component of $\bG_{S}$ is \emph{red} if it contains at least one red edge;
    otherwise, the connected component is \emph{black}.
    We denote the number of black components with $\bc(\bG_{S})$.
\end{restatable}

Note that all vertices contained in black components correspond to characters of $P$ and
$T$.
Moreover, all characters of a single black component are the same, because the presence of
a black edge indicates that some alignment in $S$ matches the two corresponding
characters.

The inference graph $\bG_S$ is represented implicitly via the edit information $\sE_{P,
T}(\mX)$ for each alignment $\mX: P \fragmentco{p}{p'}\onto T\fragmentco{t}{t'}$ in $S$.
In \cite{KNW24} we naively use the fact that $\Oh(k \log (m|\Sigma|))$ bits suffice to
encode this information; here, we provide a (slightly) more efficient encoding argument
that is required to obtain our tight results for small alphabets.

\begin{restatable}{lemma}{encbg}\label{prp:enc_gs}
    Let $S$ be a set of alignments $\mX : P\fragmentco{p}{p'}\onto T\fragmentco{t}{t'}$.
    The set $\{\sE_{P, T}(\mX) \mid \mX \in S\}$ together with the starting/ending points
    for each $\mX \in S$ can be encoded using
    \[
        \Oh\bigg(|S|\log m + \sum\nolimits_{\mX \in S \mid \cost(\mX) > 0} \cost(\mX)
        \cdot \log\big({m|\Sigma|}/{\cost (\mX)}\big)\bigg)
        \text{ bits.}
    \]
    This information, together with the bit encoding of $n = |T|$ and $m = |P|$, suffices
    to fully reconstruct the complete edge set of $\bG_{S}$ and the color of each edge.
    Moreover, this information suffices to identify the character $\sigma \in \Sigma$ for
    every node in a red component.
\end{restatable}
\begin{proof}
    For each $\mX \in S$ with $\cost(\mX) > 0$, the elements in $\sE_{P,T}(\mX)$ are
    monotone in their first and third components.
    Using a ``stars and bars'' encoding, these components take $\Oh(\log
    \binom{m+\cost(\mX)}{\cost(\mX)})$ bits. If $\cost(\mX) \leq k = o(m)$, this
    simplifies to $\Oh(\cost(\mX) \log (m/\cost(\mX)))$. The second and fourth components
    are encoded separately using $\Oh(\cost(\mX) \log |\Sigma|)$ bits.
    For each $\mX \in S$ with $\cost(\mX) = 0$ we only need to store the endpoints using
    $\Oh(\log m)$ bits each.

    Finally, for every character in $P$ and $T$ within a red component, there exists a
    path of black edges (possibly of length zero) connecting that character to one
    incident to a red edge. Since black edges connect identical characters, the character
    value is invariant along this path. Because we explicitly store all characters
    incident to red edges, we can deduce the character at the origin of any such path.
\end{proof}

\Cref{prp:enc_gs} provides an encoding that reveals all characters contained in red
components of in the inference graph $\bG_S$.
Thus, the case $\bc(\bG_S) = 0$ is easy as we can then fully retrieve $P$ and~$T$.

Consequently, we assume $\bc(\bG_S) > 0$ for the rest of the section, as this is the
remaining case.
Even though one cannot learn the characters in black components via the edit information
$\sE_{P, T}(\mX)$ for each alignment $\mX: P\fragmentco{p}{p'} \onto T\fragmentco{t}{t'}$
in $S$ (and storing all of them would be prohibitive), one can still infer which character
belongs to which component.

The key observation of \cite{KNW24} is that black components are extremely structured and
appear in a periodic fashion.
To make this structure more formal, we define two strings, $T_{|S}$ and $P_{|S}$, obtained
by retaining only the characters that belong to black connected components.

\begin{restatable}[{\cite[Definition 4.3]{KNW24}}]{definition}{tsps}\label{def:tsps}
    Let $T_{|S}$ and $P_{|S}$ denote the subsequences of $T$ and $P$, respectively,
    that consist of the characters contained in black components of~$\bG_{S}$.
    We denote the lengths of $T_{|S}$ and $P_{|S}$ by $n_S$ and $m_S$, respectively.
\end{restatable}

Consider the black edges of the inference graph \(\bG_S\)
that are induced by a single alignment \(\mX \in S\).
The first structural observation is that such edges
correspond to an \emph{exact occurrence} of \(P_{|S} \) in \( T_{|S} \).

\begin{lemmaq}[{\cite[Claim 4.5]{KNW24}}] \label{clm:occ_old}
    Let $\mX : P \onto T\fragmentco{y}{y'} \in S$,
    and define $y_\mX, y_{\mX}' \in \fragment{0}{n_S}$ as the number
    of characters of $T_{|S}$ contained in $T\fragmentco{0}{y}$ and $T\fragmentco{0}{y'}$,
    respectively.

    Then, we have $P_{|S}=T_{|S}\fragmentco{y_{\mX}}{y'_{\mX}}$,
    and $\mX$ induces edges between $P_{|S}\position{p}$ and $T_{|S}\position{y_{\mX} +
    p}$ for every $p\in \fragmentco{0}{|P_{|S}|}$, and no other edges incident to
    characters of $P_{|S}$ or $T_{|S}$.
\end{lemmaq}

Many overlapping exact occurrences induce periods of the pattern; thus, in \cite{KNW24} we
enforce the following condition on $S$ to ensure the induced exact matchings overlap.

\begin{definition}[{\cite[Definition 4.2]{KNW24}}]
    We say \emph{$S$ encloses $T$} if $|T| \leq 2|P|-2k$ and there exist two alignments
    $\mXpref,\mXsuf \in S$ such that $\mXpref$ aligns $P$ with a prefix of $T$ and
    $\mXsuf$ aligns $P$
    with a suffix of $T$, or equivalently $(0,0) \in \mXpref$ and $(|P|,|T|) \in \mXsuf$.
\end{definition}

Given this condition on $S$, we show that not only are $P_{|S}$ and $T_{|S}$ periodic, but
even the membership of their characters in black components follows a periodic structure.

\begin{lemmaq}[{\cite[Lemma 4.4]{KNW24}}]
    \label{lem:period_old}
    If $S$ encloses $T$, then, for every $c \in \fragmentco{0}{\bc(\bG_S)}$, there is a
    black connected component with node set
    $\{P_{|S}\position{i} \mid i \equiv_{\bc(\bG_S)} c \} \cup
        \{T_{|S}\position{i} \mid i \equiv_{\bc(\bG_S)} c \}$,
    that is, a black connected component containing all characters of $P_{|S}$ and
    $T_{|S}$ appearing at positions congruent to $c$ modulo $\bc(\bG_S)$. Moreover, the
    last characters of $P_{|S}$ and $T_{|S}$ are contained in the same black connected
    component, that is, $|T_{|S}| \equiv_{\bc(\bG_S)} |P_{|S}|$.
\end{lemmaq}
\Cref{lem:period_old} allows us to define the following quantities related to the periodic
structure.
\begin{itemize}
    \item For $c \in \fragmentco{0}{\bc(\bG_S)}$, we define the \emph{$c$-th black
        connected component} as the black connected component containing $P_{|S}[c]$ and
        set $ m_c \coloneqq \left\lceil(m_S - c)/\bc(\bG_{S}) \right\rceil$ and $n_c \coloneqq
        \left\lceil (n_S - c)/\bc(\bG_{S})  \right\rceil$ as the number of characters in
        $P$ and $T$, respectively, that belong to the $c$-th black connected component.
    \item For $c \in \fragmentco{0}{\bc(\bG_S)}$ and $j \in \fragmentco{0}{m_c}$, we
        define $\pi_j^c \in \fragmentco{0}{m}$ as the position of $P_{|S}\position{c + j
        \cdot \bc(\bG_S)}$ in $P$. Similarly, for $c \in \fragmentco{0}{\bc(\bG_S)}$ and
        $i \in \fragmentco{0}{n_c}$, we define $\tau_i^c \in \fragmentco{0}{n}$ as the
        position of $T_{|S}\position{c + i \cdot \bc(\bG_S)}$ in $T$.
        Note that, for any $c \in \fragmentco{0}{\bc(\bG_S)}$, the characters
        \(\{T\position{\tau_i^c}\}^{n_c-1}_{i=0} \cup
        \{P\position{\pi_j^c}\}^{m_c-1}_{j=0}\)
        are exactly those in the $c$-th black component. Consequently, they are all
        identical.
\end{itemize}

\subsubsection*{\boldmath Constructing $S$ and the Encoding}

In \cite{KNW24}, based on $S$ and a weight $w=\Oh(k|S|)$, we identify a subset of black
components $C_S \subseteq \fragmentco{0}{\bc(\bG_S)}$ so that the characters
$P\position{\pi_0^c}$ with $c\in C_S$ belong to fragments of $T$ whose LZ77 parses consist
of $w=\Oh(k|S|)$ phrases in total.
The construction ensures several desirable properties specified below.
This is the most technical component of \cite{KNW24}, and it is not significantly affected
by our modifications, so we omit deeper insights here.

These desirable properties concern all fragments $T\fragmentco{t}{t'}$ such that if we
were to align $P$ onto $T\fragmentco{t}{t'}$, then $\pi_0^0$ would be close enough to
$\tau_i^0$ for some $i\in \fragmentco{0}{n_0}$.

\begin{definition}[{\cite[Definition 4.28]{KNW24}}]\label{def:scomplete_old}
    We say that $S$ \emph{captures} $T\fragmentco{t}{t'}$ if $S$ encloses $T$
    and either $\bc(\bG_{S}) = 0$ or $|\tau_i^{0} - t - \pi_0^{0}| \leq w + 3k$ holds for
    some $i\in \fragmentco{0}{n_0}$.
\end{definition}

For all such $T\fragmentco{t}{t'}$, the information carried by $S$ and $C_S$ is sufficient
to infer whether $\ed(P, T\fragmentco{t}{t'})\le k$ and, if so, provide the exact distance
and infer the edit sequence.

\begin{theoremq}[{\cite[Corollary 4.30]{KNW24}}]
    \label{thm:occ_old}
    Let $S$ be a set of $k$-edit alignments of $P$ onto fragments of $T$
    such that $S$ encloses $T$ and $\bc(\bG_{S}) > 0$.
    Construct $P^\#$ and $T^\#$ by replacing, for every $c \notin C_S$, every character in
    the $c$-th black component with a unique character $\#_c$.

    If $S$ captures all $k$-error occurrences, then
    \begin{enumerate}
        \item $\ed(P, T\fragmentco{t}{t'}) = \ed(P^\#, T^\#\fragmentco{t}{t'})$ and
            $\sE_{P, T}(\mX) = \sE_{P^\#, T^\#}(\mX)$ for all optimal alignments $\mX \mid
            P \onto T\fragmentco{t}{t'}$ of cost at most $k$; and
        \item $\ed(P, T\fragmentco{t}{t'}) \leq \ed(P^\#, T^\#\fragmentco{t}{t'})$ for all
            integers $0 \le t \le t' \le n$.
        \qedhere
    \end{enumerate}
\end{theoremq}

\Cref{thm:occ_old} makes sure that one can encode the desired information for all captured
$k$-error occurrences. For those that are not captured, the following result is shown in
\cite{KNW24}.

\begin{lemmaq}[{\cite[Lemma 4.27]{KNW24}}]\label{lem:halve_old}
    Let $\mY : P \onto T\fragmentco{t}{t'}$ be an alignment of cost at most $k$.
    If $|\tau_{i}^{0} - t - \pi_0^{0}|>w+2k$ holds for every $i\in \fragment{0}{n_0-m_0}$,
    then there is no $c \in \fragmentco{0}{\bc(\bG_{S})}$ such that  $\mY$ aligns
    $P\position{\pi_0^{c}}$ with
    $T\position{\tau_i^c}$ for some $i\in \fragmentco{0}{n_c}$.
\end{lemmaq}

An alignment $\mY$ as in \cref{lem:halve_old} satisfies $\bc(\bG_{S \cup \{\mY\}}) \leq
\bc(\bG_{S})/2$, since each black component becomes red or is merged with another black
component.
Thus, adding $\mY$ to $S$ allows for significant progress towards $\bc(\bG_S) = 0$.
Altogether, the construction~proceeds~as~follows.
\begin{itemize}
    \item In the beginning, we set $S = \{\mXpref, \mXsuf\}$.
    \item While $S$ does not capture all $k$-error occurrences, select an uncaptured
        $k$-error occurrence $T\fragmentco{t}{t'}$ and add to $S$ an optimal alignment
        $\mY : P \onto T\fragmentco{t}{t'}$.
    \item Return the edit information $\sE_{P, T}(\mX)$ for alignments $\mX\in S$, encoded
        using \cref{prp:enc_gs}, and the set $\{(c, P\position{\pi_0^c}) \mid c \in
        C_S\}$, encoded using the LZ77-compressed fragments of $P$.
\end{itemize}
Since the second step can be executed at most $\Oh(\log m)$ times before $S$ captures all
$k$-error occurrences, the entire encoding requires $\Oh(k \log m \log (m|\Sigma|/k))$
bits.
To decode, it suffices to retrieve the inference graph $\bG_S$ from the edit information
and, using the encoding of $C_S$ as described in \cref{thm:occ_old}, construct $P^\#$ and
$T^\#$, which preserve the solution to the original problem.

\subsection{Toward an Improved Encoding}
\label{subsec:improving}

To tighten our upper bound on the communication complexity, we revisit the step where an
uncaptured alignment $\mY$ is added to $S$.
To this end, we strengthen the conclusion of \cref{lem:halve_old}, which applies only to
the initial character $P\position{\pi_0^c}$ of each black component $c \in
\fragmentco{0}{\bc(\bG_S)}$.
In \cref{sec:halves}, we prove that $\mY$ cannot match any characters within the same
black component.

\begin{restatable*}{lemma}{periodhalves}\label{lem:periodhalves}
    Let $\mY : P \onto T\fragmentco{t}{t'}$ be an optimal alignment with $\cost(\mY) \leq
    k$.
    If $C_S \ne \fragmentco{0}{\bc(\bG_S)}$ and $|\tau_{i}^{0} - t - \pi_0^{0}|>w+2k$
    holds for every $i\in \fragmentco{0}{n_0}$, then there is no $c \in
    \fragmentco{0}{\bc(\bG_{S})}$ such that  $\mY$ aligns $P\position{\pi_j^{c}}$ with
    $T\position{\tau_i^c}$ for some $i\in \fragmentco{0}{n_c}$ and $j \in
    \fragmentco{0}{m_c}$.
\end{restatable*}

\Cref{lem:periodhalves} brings us to the main idea of our improvement: rather than adding
an entire uncaptured alignment $\mY$ to $S$, it suffices to add only the subset of edges
required to merge every black component with another component. More specifically, if we
decompose $P$ into
\[
    P\fragmentco{0}{\pi_0^0} \circ \left( \bigcirc_{j=0}^{m_0-2}
    P\fragmentco{\pi_j^0}{\pi_{j+1}^0} \right) \circ P\fragmentco{\pi_{m_0-1}^0}{|P|},
\]
then at least one of the middle $m_0-1$ fragments of $P$ in this decomposition carries at
most $\Oh(k/m_0)$ edits in $\mY$.
So ideally, we would like to add to $S$ and the corresponding edges $\mY$ restricted to
this fragment.
If we managed to show that the periodic structure is preserved already by adding such a
subset, then not only $\bc(\bG_S)$ would halve, but also the parameter $m_0$ would at
least double in each iteration.
Hence, with this new strategy, we could hope that the total cost of alignment or partial
alignment in $S$ is bounded by  $\sum_{i=1}^{\infty} k / 2^i = \Oh(k)$.

On a technical level, the proof is delicate: The selection of the fragment
$P\fragmentco{\pi_j^0}{\pi_{j+1}^0}$, as described so far, depends on $\mY$ and the black
components defined relative to $S$. However, the structure of the black components changes
when going from $S$ to $S \cup \mY$, and the fragment is not guaranteed to contain the
same black components as in $S \cup \mY$.
Thus, for our proof of correctness, we need to be careful to get rid of any such undesired
dependence on the future.

Finally, we note that \cref{lem:periodhalves} is not a strict strengthening of
\cref{lem:halve_old}, as it introduces three additional conditions: now $i \in
\fragmentco{0}{n_0}$ instead of $i \in \fragment{0}{n_0 -m_0}$, $\mY$ is optimal, and that
$C_S \ne \fragmentco{0}{\bc(\bG_S)}$. We observe that all conditions are relatively
harmless. Since we are only concerned with capturing all optimal alignments, the first two
conditions are perfectly acceptable. Moreover, whenever $C_S =
\fragmentco{0}{\bc(\bG_S)}$, we have $P = P^{\#}$ and $T = T^{\#}$, where $P^{\#}$ and
$T^{\#}$ are defined as in \cref{thm:occ_old}. At this point, we know that $P$ and $T$
already have a sufficiently cheap encoding to be sent directly.

\subsubsection*{\boldmath Redefining the Set $S$}

Let us start making our improvement formal.
To this end, we relax the definition of $S$: we let $S$ be a set of alignments of
\emph{fragments of $P$} onto fragments of $T$ with cost at most $k$, and we denote by
$\cost(S) \coloneqq \sum_{\mX \in S} \cost(\mX)$ the sum of all costs of the alignments
contained in $S$.

Based on such $S$, we can define the inference graph $\bG_S$ as before and argue that all
characters in red components follow from the edit information for all $\mX \in S$.
We conclude again that $\bc(\bG_S)=0$ is easy, and we assume that $\bc(\bG_S)>0$ in this
(sub)section.
We also define $P_{|S}$ and $T_{|S}$ as before.
\Cref{clm:occ_old} can be extended quite easily to fragments of $P$; see
\cref{sec:sec_gs}.

\begin{restatable*}{lemma}{occ}\label{clm:occ}
    Let $\mX : P\fragmentco{x}{x'} \onto T\fragmentco{y}{y'} \in S$
    and define $y_\mX, y_{\mX}' \in \fragment{0}{n_S}$ as the number
    of characters of $T_{|S}$ contained in $T\fragmentco{0}{y}$ and $T\fragmentco{0}{y'}$.
    Similarly, define $x_\mX, x_{\mX}' \in \fragment{0}{m_S}$ as the number
    of characters of $P_{|S}$ contained in $P\fragmentco{0}{x}$ and $P\fragmentco{0}{x'}$.

    Then, $P_{|S}\fragmentco{x_\mX}{x'_\mX}=T_{|S}\fragmentco{y_{\mX}}{y'_{\mX}}$,
    and $\mX$ induces edges between $P_{|S}\position{x_{\mX} + p}$ and
    $T_{|S}\position{y_{\mX} + p}$ for $p\in \fragmentco{0}{x_\mX' - x_\mX}$, and no other
    edges incident to characters of $P_{|S}$ or $T_{|S}$.
\end{restatable*}

We next give the new condition on $S$ needed to observe the periodic structure.

\begin{restatable*}{definition}{enclose}\label{def:enclose}
    We say $S$ \emph{encloses $T$} if $|T| \le 2 \cdot |P| - 2k$
    and $S$ can be written as
    \[
        S = \{\mXpref, \mXsuf, \mA_1, \ldots, \mA_{a}, \mB_1, \ldots, \mB_b\}
    \]
    such that the following conditions hold:
    \begin{enumerate}
        \item $\mXpref$ aligns the whole $P$ with a prefix of $T$;
        \label{def:enclose:a}
        \item $\mXsuf$ aligns the whole $P$ with a suffix of $T$;
        \label{def:enclose:b}
        \item $\mA_1, \ldots, \mA_{a}$ align the whole $P$ with fragments of $T$.
        \label{def:enclose:c}
    \end{enumerate}

    We say $S$ is \emph{degenerate} if $y_{\mXpref} = y_{\mXsuf} = y_{\mA_1} = \ldots =
    y_{\mA_a} = 0$. Moreover, we define
    \[
        g_0 \coloneqq
        \begin{cases}
            m_S & \text{if $S$ is degenerate,} \\
            \gcd\{y_{\mXpref}, y_{\mXsuf}, y_{\mA_1}, \ldots, y_{\mA_a}\} & \text{otherwise}.
        \end{cases}
    \]
    For each \( i \in \fragment{1}{b} \), we define
    \( g_i \coloneqq \gcd(g_{i-1}, y_{\mB_i} - x_{\mB_i}). \)
    We further say that \( S \) \emph{succinctly encloses} \( T \) if the following
    condition holds for every \( i \in \fragment{1}{b} \).
    \begin{enumerate}[resume]
        \item\label{def:enclose:d} The alignment $\mB_i : P\fragmentco{x_i}{x_i'} \onto
            T\fragmentco{y_i}{y_i'}$ satisfies
        $x_{\mB_i}' - x_{\mB_i} \geq g_{i-1}+1$.
        \qedhere
    \end{enumerate}
\end{restatable*}

If \(S\) encloses \(T\), then each of the alignments \(\mXpref, \mXsuf, \mA_1, \dots,
\mA_a\) aligns the entire string \(P\) with a fragment of \(T\), and the edges they induce
connect to all vertices of \(P_{|S}\). Consequently, we have \(x_{\mXpref} = x_{\mXsuf} =
x_{\mA_1} = \dots = x_{\mA_a} = 0\) and \(x_{\mXpref}' = x_{\mXsuf}' = x_{\mA_1}' = \dots
= x_{\mA_a}' = m_S\).
Moreover, note that whenever \(S\) succinctly encloses \(T\) and $S$ is degenerate, then
\(b = 0\), since the condition \eqref{def:enclose:d} cannot hold for $i = 1$ if $b > 0$.

We can also prove something more about the exact matching induced by \(\mXpref\) and
\(\mXsuf\).

\begin{restatable*}{lemma}{nosingle}\label{clm:nosingle}
    If $S$ encloses $T$, then $y_{\mXpref}=0$, $y_{\mXsuf}=n_S-m_S$, and $n_S\le 2m_S$.
\end{restatable*}
The proof, provided for completeness in \cref{sec:sec_gs}, is the same as \cite[Claim
4.6]{KNW24}.
This is because the argument relies exclusively on the alignments $\mXpref$ and $\mXsuf$
(and not any $\mB_i$).

\Cref{clm:nosingle} also implies that, when $S$ is degenerate, then $m_S = n_S$ and every
alignment $\mX \in \{\mXpref, \mXsuf, \mA_1, \ldots, \mA_a\}$ induces in $\bG_S$ an edge
$\{P_{|S}[p], T_{|S}[t]\}$ iff $p = t$.
In this case, for each $p \in \fragmentco{0}{m_S}$, there is a black component consisting
solely of $P_{|S}[p]$ and $T_{|S}[p]$.

The following generalization of the inference graph $\bG_S$ lets us work with graphs of
more predictable structure.

\begin{restatable*}{definition}{gsi}\label{def:gsi}
    Suppose $S$ encloses $T$.
    We define the sets \( S^0 \coloneqq \{\mXpref, \mXsuf, \mA_1, \ldots, \mA_{a}\}\) and
    \( S^i \coloneqq \{\mXpref, \mXsuf, \mA_1, \ldots, \mA_{a}, \mB_{1}, \ldots,
    \mB_{i}\}\) for \( i \in \fragment{1}{b} \).

    Moreover, for \( i \in \fragment{0}{b} \), we define the graph \( \bG^i_S \) obtained
    from \( \bG_S \) as follows.
    We remove vertices (along with their incident edges) that are not in black components
    in \( \bG_S \) (so only vertices in \( P_{|S}\) and \( T_{|S}\) remain), and we
    additionally remove any edge that is not induced by alignments in the set \( S^i \).
\end{restatable*}

We now proceed to argue this definition successfully preserves the periodic structure.

\begin{restatable}{lemma}{periodicity}\label{lem:periodicity}
    Suppose $S$ encloses $T$. Then, the following hold for each $i \in \fragmentco{0}{b}$.
    \begin{enumerate}
        \item\label{lem:periodicity:i}
            If condition \eqref{def:enclose:d} in \cref{def:enclose} holds for all indices
            in $\fragment{1}{i}$, then, for each $c \in \fragmentco{0}{g_i}$, the graph
            $\bG^i_S$ has a connected component with node set
            $C^i_c \coloneqq \{P_{|S}[j] : j \equiv_{g_i} c\} \cup \{T_{|S}[j] : j
            \equiv_{g_i} c\}$,
            where $j$ ranges over $\fragmentco{0}{m_S}$ and $\fragmentco{0}{n_S}$,
            respectively.

            In particular, in $\bG_S = \bG^b_S$,
            for every $c \in \fragmentco{0}{\bc(\bG_S)}$, there exists a black connected
            component with node set
            $\{P_{|S}[i] : i \equiv_{\bc(\bG_S)} c\} \cup \{T_{|S}[i] : i \equiv_{\bc(\bG_S)} c\}$.
        \item\label{lem:periodicity:ii}
            The last characters of $P_{|S}$ and $T_{|S}$ are in the same component of
            $\bG^i_S$.
            \qedhere
    \end{enumerate}
\end{restatable}

\begin{proof}
    We first prove \eqref{lem:periodicity:i}.
    If $S$ is degenerate (and hence $b = 0$), we have already argued that $\bG_S =
    \bG_{S^0} = \bG_S^0$ has the claimed structure. Therefore, we assume that $S$ is not
    degenerate.

    Let us assign a unique label $\$_C$ to each component $C$ of $\bG_S^i$ and define
    strings $P_{\cc}$ and $T_{\cc}$ of
    length $m_S$ and $n_S$, respectively, as follows.
    For $i\in \fragmentco{0}{m_S}$, set $P_{\cc}[i] = \$_C$, where $C$ is the component of
    $\bG_S^i$ containing $P_{|S}[i]$.
    Similarly, for $i\in \fragmentco{0}{n_S}$, set $T_{\cc}[i] = \$_C$, where $C$ is the
    black connected component of $\bG_S^i$ containing $T_{|S}[i]$.

    By \cref{clm:occ,clm:nosingle}, we have
    \[\{0, |T_{\cc}| - |P_{\cc}|\} \subseteq \{y_{\mXpref}, y_{\mXsuf}, y_{\mA_1}, \ldots,
    y_{\mA_a}\} \subseteq \Occ(P_{\cc}, T_{\cc}).\]
    Since \(|T_{\cc}| \leq 2|P_{\cc}|\), we can use the following claim already proved in \cite{KNW24}.

    \begin{claimq}[{\cite[Lemma 3.2]{KNW24}}]\label{fct:periodicity}
        Consider a non-empty pattern $P'$ and a text $T'$ with $|T'|\le 2|P'|+1$.
        If $\{0,|T'|-|P'|\}\subseteq \Occ(P',T')$, that is, $P'$ occurs both as a prefix
        and as a suffix of $T'$, then $\gcd(\Occ(P',T'))$ is a period of $T'$.
    \end{claimq}

    It follows that \(\gcd(\Occ(P_{\cc}, T_{\cc}))\) is a period of \(T_{\cc}\). Hence,
    \(\gcd\{y_{\mXpref}, y_{\mXsuf}, y_{\mA_1}, \ldots, y_{\mA_a}\}\) is a period of both
    \(T_{\cc}\) and its prefix \(P_{\cc}\), and thus \(g_0\) is a period of both
    \(T_{\cc}\) and \(P_{\cc}\).

    We now prove by induction that \( T_{\cc} \) and \( P_{\cc} \) have period \( g_i \)
    for all \( i \in \fragment{0}{b} \). The base case \( i = 0 \) follows from the
    argument above. For \( i > 0 \), assuming \( T_{\cc} \) and \( P_{\cc} \) have period
    \( g_{i-1} \), we observe that the condition on \( \mB_i \) ensures $x_{\mB_i}' -
    x_{\mB_i} \geq g_{i-1} + 1 \ge g_{i-1}$.
    The following new claim allows us to conclude that \( T_{\cc} \) and \( P_{\cc} \)
    have period $g_i = \gcd(g_{i-1}, y_{\mB_i} - x_{\mB_{i}})$.

    \begin{claim}\label{fct:periodicity_extended}
        Consider strings $P'$ and $T'$ with a common string period $Q$.
        If there are matching fragments $P'\fragmentco{x}{x'} = T'\fragmentco{y}{y'}$ of
        length $x' - x \geq |Q|$,
        then $P'$ and $T'$ also have a common period $\gcd(|Q|, y-x)$.
    \end{claim}
    \begin{claimproof}
        The claim is trivial if $y = x$, so we henceforth assume otherwise.
        By symmetry between the fragments, we can also assume without loss of generality that $y>x$.
        Since $P'$ and $T'$ are prefixes of $Q^\infty$, we have
        $Q^\infty\fragmentco{x}{x+|Q|}=Q^\infty\fragmentco{y}{y+|Q|}$.
        Observe that $y-x$ is a period of $Q^\infty$: for each $i\in \Zz$, we indeed have
        $Q^\infty\position{i}=Q^\infty\position{x+(i-x)\bmod |Q|} =
        Q^\infty\position{y+(i-x)\bmod |Q|} = Q^\infty\position{i+y-x}$.
        From the Periodicity Lemma (\cref{lem:perlemma}) applied for
        $Q^\infty\fragmentco{0}{|Q|+y-x}$,
        we conclude that $\gcd(|Q|,y-x)$ is a period of $Q^\infty\fragmentco{0}{|Q|+y-x}$
        and hence of $Q$.
        As a divisor of $|Q|$, it is also a common period of $Q^\infty$, $P'$, and $T'$.
    \end{claimproof}

    Consequently, for each $c\in \fragmentco{0}{g_i}$, the set $C^i_c$ belongs to a single
    connected component of $\bG_S^i$.
    It remains to prove $g_i = \bc(\bG_S^i)$. We do this by demonstrating that no edge of
    $\bG_S^i$ leaves $C^i_c$.
    Note that \cref{clm:occ} further implies that every edge incident to $P_{|S}$ or
    $T_{|S}$ connects $P_{|S}\position{p}$ with $T_{|S}\position{t}$ such that
    $t=p+(y_{\mX} - x_{\mX})$ for some $\mX \in S$ and $p \in
    \fragmentco{x_{\mX}}{x_{\mX}'}$.
    In particular, $t \equiv_{g_i} p$, so $P_{|S}\position{p}\in C^i_c$ holds if and only
    if $T_{|S}\position{t}\in C^i_c$.

    As for \eqref{lem:periodicity:ii}, since $P_{\cc}$ is a suffix of $T_{\cc}$, the last
    characters of $P_{|S}$ and $T_{|S}$ are connected.
\end{proof}

\subsubsection*{\boldmath A New Iterative Construction of $S$}

For the new construction of $S$, we use the following key lemma.

\begin{restatable}{lemma}{additer}\label{lem:add_iter}
    Suppose $S$ encloses $T$ and condition \eqref{def:enclose:d}
    of \cref{def:enclose} holds up to $i$ (inclusive) for some $i \in \fragmentco{0}{b-1}$
    (if $i = 0$, then this statement is void).
    Further, suppose that there exists an integer $z$ such that all characters of
    $P_{|S^{i}}\fragmentco{z \cdot \bc(\bG_{S^{i}})}{(z + 2) \cdot \bc(\bG_{S^{i}})}$ are
    in the pre-image of $\mB_{i+1}$.

    Then, either $\bc(\bG_S) = 0$ or condition \eqref{def:enclose:d} also holds for $i+1$.
\end{restatable}
\begin{proof}
    By \cref{lem:periodicity}\eqref{lem:periodicity:i}, in the inference graph
    \(\bG_{S^{i}}\), for each \(c \in \fragmentco{0}{\bc(\bG_{S^{i}})}\), there is a black
    component with node set
    \[\{P_{|{S^{i}}}[j] : j \equiv_{\bc(\bG_{S^{i}})} c\} \cup \{T_{|S^{i}}[j] : j
    \equiv_{\bc(\bG_{S^{i}})} c\}.\]
    Consider the graph \(\bG_S^{i}\). This graph is identical to \(\bG_{S^{i}}\) except
    that the components that become red in \(\bG_S\) are removed. Let \(R \subseteq
    \fragmentco{0}{\bc(\bG_{S^{i}})}\) index these removed components.

    If \(R = \fragmentco{0}{\bc(\bG_{S^{i}})}\), then \(\bc(\bG_S) = 0\), and the proof is
    complete. Otherwise, we proceed as follows.
    Choose an arbitrary \(r \in \fragmentco{0}{\bc(\bG_{S^{i}})} \setminus R\), and
    consider the fragment
    \[
        P_{|{S^{i}}}\fragment{z \cdot \bc(\bG_{S^{i}}) + r}{(z+1) \cdot
        \bc(\bG_{S^{i}})+r} \subseteq  P_{|S^{i}}\fragmentco{z \cdot \bc(\bG_{S^{i}})}{(z
        + 2) \cdot \bc(\bG_{S^{i}})},
    \]
    whose characters, by our assumption, are also contained in the pre-image of
    $\mB_{i+1}$.
    The characters $P_{|{S^{i}}}\position{z \cdot \bc(\bG_{S^{i}}) + r}$ and
    $P_{|{S^{i}}}\position{(z+1) \cdot \bc(\bG_{S^{i}}) + r}$ belong to the same black
    component of \(\bG_{S^{i}}\).
    Since \(r \notin R\), these characters survive in \(\bG_S^{i}\)
    and correspond to some characters \(P_{|S}\position{p}\) and \(P_{|S}\position{p'}\).

    Note that we must have \( |p' - p| \geq g_{i} \) as, by
    \cref{lem:periodicity}\eqref{lem:periodicity:i}, characters in the same component of
    \(\bG_S^{i}\) lie at least $g_{i}$ positions apart.
    Thus, $|P_{|S}\fragment{p}{p'}| \geq g_{i}+1$.
    Since $P_{|S}\fragment{p}{p'}$ is contained in
    $P_{|S}\fragmentco{x_{\mB_{i+1}}}{x_{\mB_{i+1}}'}$, we get $x_{\mB_{i+1}}' -
    x_{\mB_{i+1}} \geq g_{i}+1$.
\end{proof}

\begin{algorithm}[t]
\DontPrintSemicolon
\LinesNumbered
\KwInput{Strings $P$ and $T$, and a threshold $k$ such that $P$ has $k$-error occurrences
as both a prefix and a suffix of $T$.}
\KwOutput{A set $S$ that succinctly encloses $T$ such that $\cost(S) = \Oh(k)$.}
Initialize $S \coloneqq \{\mXpref, \mXsuf\}$, where $\mXpref, \mXsuf$ are alignments of
cost at most $k$ aligning $P$ with a prefix and a suffix of $T$, which exist by the input
assumptions\;
\For{$\ell \in \fragment{1}{2}$}{ \label{loop:add_A}
    Construct $\bG_S$, $P_{|S}$, $T_{|S}$ and $C_S$ for $S$, and introduce the corresponding notation\;
    \lIf{$S$ captures all $k$-error occurrences \KwSty{or} $C_S = \fragmentco{0}{\bc(\bG_S)}$}{\Return{$S$}}
    Let $\mY \mid P \onto T\fragmentco{t}{t'}$ be an optimal $k$-edit alignment that $S$ does not capture\; \label{alg:sels}
    Add the full alignment $\mY$ to $S$ as the alignment $\mA_{\ell}$\;
}
\For{$\ell \in \mathbb{Z}_{\ge 1}$}{\label{loop:add_B}
    Construct $\bG_S$, $P_{|S}$, $T_{|S}$ and $C_S$ for $S$, and introduce the corresponding notation\;
    \lIf{$S$ captures all $k$-error occurrences \KwSty{or} $C_S = \fragmentco{0}{\bc(\bG_S)}$}{\Return{$S$}}
    Let $\mY \mid P \onto T\fragmentco{t}{t'}$ be an optimal $k$-edit alignment that $S$
    does not capture\; \label{alg:selstwo}
    Select $j \in \fragmentco{0}{m_0 - 2}$ minimizing the cost of the
    partial alignment $\mY_j \subseteq \mY$, where $\mY_j :
    P\fragmentco{\pi^0_{j}}{\pi^0_{j+2}} \onto
    \mY(P\fragmentco{\pi^0_{j}}{\pi^0_{j+2}})$\;
    \label{alg:construction_S:select_i}
    Add $\mY_j$ to $S$ as the alignment $\mB_{\ell}$\;
}
\caption{Construction of $S$ that succinctly encloses $T$.}\label{alg:construction_S}
\end{algorithm}

Finally, we can provide the new construction of $S$. As in \cite{KNW24}, we construct,
from $S$, a set $C_S$ of black components and a parameter $w$ based on a weight function
covering $S$. We use the following facts (proved in \cref{sec:full}).
\begin{enumerate*}
    \item \cref{lem:periodhalves} still holds, and after replacing ``enclosure'' with
        ``succinct enclosure'' in \cref{def:scomplete_old}, \cref{thm:occ_old} still
        holds.
    \item We can improve the bound for encoding $\{(c, T\position{\tau_0^c}) \mid c \in
        C_S\}$ from $\Oh(w \log m)$ additional bits (on top of the edit information of
        $S$) to
        $\Oh\left((w+k) \log (1+m|\Sigma|/(w+k))\right)$ bits.
    \item All of this is possible with $w = \Oh(\cost(S))$ instead of $w = \Oh(k|S|)$.
\end{enumerate*}

\begin{restatable*}{lemma}{construction}\label{lem:construction_S:loop}
    Suppose that $n \leq 2m-2k$ and that there are $k$-error occurrences of $P$ appearing
    as both a prefix and a suffix of $T$. Then, \cref{alg:construction_S} computes a set
    $S$ of size $|S| = \Oh(\min(\log m,|\floor{\OccE_k(P,T)/k}|))$ and $\cost(S) = \Oh(k)$
    such that the following hold:
    \begin{enumerate}
        \item $S$ succinctly encloses $T$,
        \item $C_S = \fragmentco{0}{\bc(\bG_S)}$ or $S$ captures all $k$-error occurrences
            of $P$ in $T$,
        \item the set $\{\sE_{P, T}(\mX) \mid S \ni \mX \mid P\fragmentco{p}{p'}\onto
            T\fragmentco{t}{t'}\}$ together with all starting/ending points of the
            alignments in $S$ can be encoded in $\Oh(k\log(m|\Sigma|/k) + |S|\log m)$
            bits.
            \qedhere
    \end{enumerate}
\end{restatable*}

\subparagraph*{Proof sketch (full proof in \cref{subsec:together}).}
In \cref{alg:construction_S}, we maintain the invariant that \( S \) succinctly
encloses \( T \) at the start of every iteration of the loops in
Lines~\ref{loop:add_A} and~\ref{loop:add_B}.

This invariant holds at the beginning as we set $S\coloneqq \{\mXpref,\mXsuf\}$.
We proceed by adding at most two alignments $\mA_\ell$ to $S$ corresponding to
uncaptured $k$-error occurrences in the loop of Line~\ref{loop:add_A}.
After the $\ell$-th iteration, one of $\bc(\bG_S) = 0$, $C_S =
\fragmentco{0}{\bc(\bG_S)}$, or $s_P \geq 2^{\ell}$ holds, where $s_P$ is the minimum
number of characters of $P$ in a black component of $\bG_S$.
Indeed, by \cref{lem:periodhalves}, after adding $\mY$ to $S$, each black component of
$\bG_S$ either becomes red or is merged with another component.
Hence, unless $\bc(\bG_{S\cup\{\mY\}})=0$, every black component of
$\bG_{S\cup\{\mY\}}$ contains at least two black components of $\bG_S$, so $s_P$ at
least doubles.

Adding the two alignments of the form $\mA_\ell$ ensures that \( m_0 \geq s_P \geq 2^2
= 4 \) when we enter the loop at Line~\ref{loop:add_B} (so that \( m_0 \geq 3 \) in
Line~\ref{alg:construction_S:select_i}, and the quantifier \( j \in \fragmentco{0}{m_0
- 2} \) is well-defined). There, we start adding alignments of the form $\mB_\ell$ to
$S$. Each time we add such an alignment, we use \cref{lem:add_iter,lem:periodhalves}
to prove that $S$ still succinctly encloses~$T$.
Using a similar argument as before, $\bc(\bG_S) = 0$, $C_S =
\fragmentco{0}{\bc(\bG_S)}$, or $m_0 \geq s_P \geq 2^{2 + \ell}$ at the end of the
$\ell$-th iteration.
Since $s_P \leq m$, this ensures that the loop finishes after $\Oh(\log m)$ iterations
and $|S| = \Oh(\log m)$.
When it does, all $k$-error occurrences are captured.

Next, we argue why the bound on $\cost(S)$ holds.
Since each of the alignments $\mY_j$ can overlap only with at most two others, we
obtain that the sum of the costs of the alignments $\mY_j$ for all $j \in
\fragmentco{0}{m_0 - 2}$ is at most $2k$.
Consequently, in the $\ell$-th iteration the cost of the alignment $\mY_j$ selected at
Line~\ref{alg:construction_S:select_i} is at most
$\cost(\mY_j) \leq 2k / (m_0-2) \leq 2k / (2^{\ell+1} - 2) \leq k / 2^{\ell-1}$. Thus,
$\cost(S) \leq 4k + k \cdot \sum_{\ell=1}^{\Oh(\log m)} {1}/{2^{\ell-1}} \le 6k =
\Oh(k)$.

Lastly, we give the bound on the encoding size.
The endpoints of the alignments $\mX\in S$ can be encoded in $\Oh(|S|\log m)$ bits.
On top of that, by \cref{prp:enc_gs}, the edit information $\{\sE_{P, T}(\mX) \mid
\mX\in S\}$ can be encoded in space asymptotically bounded as follows, using $S^+
\coloneqq \{\mX \in S \mid \cost(\mX) > 0\}$ and the monotonicity of $x \log (k/x)$
for $0 < x \le k/e$:
\begin{align*}
    \sum_{\mX \in S^+} \cost(\mX) \cdot \log \left( {m|\Sigma|}/{ \cost (\mX)}
    \right)
        &= \cost(S) \cdot \log \left( {m|\Sigma|}/{k} \right)
        + \sum_{\mX \in S^+} \cost(\mX) \cdot \log \left( {k}/{ \cost (\mX)} \right) \\
        &\leq k \log \left({m|\Sigma|}/{k}\right) + \Oh(k)
        + k \cdot \sum_{\ell=1}^{\Oh(\log m)} ({\ell-1})/{2^{\ell-1}}\\
        &\leq \Oh\left(k \log \left({m|\Sigma|}/{k}\right)\right)
        \text{ bits.}
\end{align*}
To show $|S| = \Oh(|\floor{\OccE_k(P,T)/k}|)$, we argue in the full proof that the
alignments $\mY$ in Line~\ref{alg:selstwo} start at least $k$ positions away from each
other.
In total, this completes the proof sketch of \cref{lem:construction_S:loop}.

Now, following the approach in \cite{KNW24}, we can encode the set $\{\sE_{P, T}(\mX) \mid
\mX \in S\}$ and the pairs $\{(c, T\position{\tau_0^c}) \mid c \in C_S\}$ using $\Oh(k
\log (m|\Sigma| / k) + |S|\log m)$ bits.
Decoding would then proceed via an appropriately modified version of \cref{thm:occ_old}.
Unfortunately, this falls short of proving \cref{thm:ccompl} because the $|S|\log m$ term
might dominate the encoding size; consequently, we must develop an alternative for cases
where this term is prohibitively large.

\subsubsection*{\boldmath Workaround}
Note that
\(
    |S|\log m = \Oh(|S|\log {m}/{\log m}) \le \Oh(\log m\cdot \log ({m}/{\log m})),\)
so $|S|\log m = \Oh(k \log ({m}/{k}))$ holds as long as $k \geq \log
m$ or $|S| = \Oh(k)$.
Thus, the critical scenario arises when $k < \log m$ and $|S|=\omega(k)$.
Not coincidentally, \cref{lem:construction_S:loop} also yields $|S| = \Oh(|\lfloor
\OccE_k(P,T)/k \rfloor|)$,
which lets us derive $|\lfloor \OccE_k(P,T)/k \rfloor|=\omega(k)$ in this case.
This condition characterizes one of the two fundamental structural results for PMwE from
\cite{CKW20}: in a setting (almost) identical to ours, there is a primitive string $Q$
with $|Q| \leq m/128k$ such that the edit distance between $P$ and a prefix of $Q^\infty$
does not exceed $2k$.
Moreover, \cite[Theorem 5.2]{CKW20} upper-bounds the distance between $T$ and a prefix of
$Q^\infty$ by $6k$.
The subsequent characterization in \cite[Section 5]{CKW20} implies that, for every
$k$-error occurrence $T\fragmentco{t}{t'}$ of $P$, there exists an exact occurrence of $Q$
in $P$ that is matched perfectly both in the alignment from $P$ to the prefix of
$Q^\infty$ and in the alignment from $P$ via $T\fragmentco{t}{t'}$ to a fragment of
$Q^\infty$.

Our workaround is to initialize $S$ so that the periodic structure of $\bG_S$
\emph{essentially coincides} with the structure induced by $Q$.
Namely, every black component corresponds to a position $r\in \fragmentco{0}{|Q|}$ and
consists of all characters of $P$ and $T$ that the alignments with prefixes of $Q^\infty$
match with characters of the form $Q^\infty\position{r+j|Q|}$ for $j\in \Zz$.
By \cref{lem:periodhalves} and the above characterization of $k$-error occurrences, this
guarantees that all $k$-error occurrences are captured by $S$.
To achieve such $\bG_S$, we include in $S$ alignments for the two approximate occurrences
of $P$ as a prefix and a suffix of $T$, and a third $14k$-edit alignment that, for
every~$j$, aligns the $j$th copy of $Q$ in $P$ with the $(j+1)$th copy of $Q$ in $T$. This
yields the following lemma proved in \cref{subsec:together}.

\begin{restatable*}{lemma}{constructSper}\label{lem:constructSper}
    Suppose that $n \leq 3/2 \cdot m - 28k$, that there are $k$-error occurrences of $P$
    appearing as both a prefix and a suffix of $T$, and that there is a primitive string
    $Q$ with $|Q| \leq m/128k$ and $\edp{P}{Q} \leq 2k$.
    Then, we can construct a set $S$ of at most three $14k$-edit alignments of $P$ onto
    fragments of $T$ such that, when the notions of \emph{succinctly enclosing} and
    \emph{capturing} are interpreted with threshold $14k$, all of the following hold.
    \begin{enumerate}
        \item $S$ succinctly encloses $T$,
        \item $S$ captures all $k$-error occurrences of $P$ in $T$,
        \item the set $\{\sE_{P, T}(\mX) \mid S \ni \mX \mid P\fragmentco{p}{p'}\onto
            T\fragmentco{t}{t'}\}$ together with all starting/ending points of the
            alignments in $S$ can be encoded in $\Oh(k\log(m|\Sigma|/k))$ bits.
            \qedhere
    \end{enumerate}
\end{restatable*}

\section{Lower Bound}
\label{subsec:lb}

\setcounter{mtheorem}{1}
\cclb
\begin{proof}
    Set $p \coloneqq \lfloor n/m \rfloor$ and $T \coloneqq S_0\cdot S_1 \cdots S_{p-1}
    \cdot \zero^{n-pm}$, where $S_0,\ldots,S_{p-1} \in \Sigma^{m}$ are strings that
    contain at most $k$ characters from $\Sigma \setminus \{\zero\}$ and all remaining
    characters are equal to $\zero$. Clearly, $\ed(P,T\fragmentco{q m}{(q+1)m}) \leq k$
    for every $q \in \fragmentco{0}{p}$.
    Moreover, for every $q \in \fragmentco{0}{p}$, there is a unique optimal alignment
    $\mA_q : P \onto T\fragmentco{q m}{(q+1)m}$ that performs substitutions exactly in the
    positions where the non-zero characters appear (it is not difficult to verify that any
    deletions or insertions would increase the cost). Thus, $\mA_q$ allows us to fully
    recover $S_q$ from the edit information $\sE_{P, T}(\mA_q)$, and consequently $T$ can
    be fully recovered.

    Finally, the number of possibilities for each block $S_q$ is $\sum_{i=0}^{k}
    \binom{m}{i} (|\Sigma| - 1)^i$.
    If $k \leq m/2$, then this quantity is at least $\binom{m}{k} (|\Sigma| - 1)^k$, and
    the standard estimate $\binom{m}{k} = 2^{\Omega(k\log(m/k))}$ yields $\binom{m}{k} =
    2^{\Omega(k\log(|\Sigma| m/k))}$ for $|\Sigma| = 2$ and $\binom{m}{k} (|\Sigma| - 1)^k
    \ge \binom{m}{k}(|\Sigma|/2)^k= 2^{\Omega(k\log(|\Sigma| m/k))}$ for $|\Sigma| > 2$.
    Otherwise, the sum contains the term $\binom{m}{\lceil m/2 \rceil} (|\Sigma| -
    1)^{\lceil m/2 \rceil}$.
    If $|\Sigma| = 2$, then this term is simply $\binom{m}{\lceil m/2 \rceil} \ge 2^m/m
    \ge 2^{\Omega(m)} \ge 2^{\Omega(m\log |\Sigma|)}$.
    If $|\Sigma| > 2$, then the term is at least $(\Sigma-1)^{\lceil m/2 \rceil} \ge
    (|\Sigma|/2)^{ m/2} = 2^{\Omega(m\log |\Sigma|)}$.
    In either case, we have $2^{\Omega(m\log |\Sigma|)} \ge 2^{\Omega(k\log(|\Sigma|
    m/k))}$ when $m/2 < k \le m$.
    Therefore, the number of possibilities for $T$ is $2^{\Omega(n/m\cdot k\log(|\Sigma| m/k))}$.
\end{proof}

\bibliography{main}

\appendix
\clearpage

\section{Additional Preliminaries}

\subsubsection*{Edit Distance Between Strings and Periodic Extensions}
We denote the minimum edit distance between a string $S$ and any prefix of a string
$T^\infty$ by $\edp{S}{T} \coloneqq \min\{\ed(S,T^\infty\fragmentco{0}{j}) \mid j \in
\Zz\}$.
We denote the minimum edit distance between a string $S$ and any substring of $T^\infty$
by $\edl{S}{T} \coloneqq \min\{\ed(S,T^\infty\fragmentco{i}{j}) \mid i, j \in \Zz\text{
and }i \le j\}$.
Lastly, we set $\eds{S}{T} \coloneqq \min \{\ed(S,T^{\infty}\fragmentco{i}{j|T|}) \mid i,j
\in \Zz \text{ and } i\leq j|T|\}$.

\subsubsection*{More on Alignments}
An alignment $\mA'':X\fragmentco{x}{x'}\onto Z\fragmentco{z}{z'}$ is a \emph{product} of
alignments $\mA:X\fragmentco{x}{x'}\onto Y\fragmentco{y}{y'}$ and
$\mA':Y\fragmentco{y}{y'}\onto Z\fragmentco{z}{z'}$ if, for every $(\bar{x},\bar{z})\in
\mA''$, there is $\bar{y}\in \fragment{y}{y'}$ such that $(\bar{x},\bar{y})\in \mA$ and
$(\bar{y},\bar{z})\in \mA'$.
A product alignment always exists, and every product alignment satisfies
\[
    \edal{\mA''}(X\fragmentco{x}{x'},Z\fragmentco{z}{z'})\le
    \edal{\mA}(X\fragmentco{x}{x'},Y\fragmentco{y}{y'})+\edal{\mA'}(Y\fragmentco{y}{y'},\allowbreak
    Z\fragmentco{z}{z'}).
\]

For an alignment $\mA:X\fragmentco{x}{x'}\onto Y\fragmentco{y}{y'}$ with \(\mA =
(x_i,y_i)_{i=0}^\ell\), we define the \emph{inverse alignment} $\mA^{-1} :
Y\fragmentco{y}{y'}\onto X\fragmentco{x}{x'}$ as $\mA^{-1} \coloneqq
(y_i,x_i)_{i=0}^\ell$.
The alignment $\mA^{-1}$ satisfies
\[
    \edal{\mA^{-1}}(Y\fragmentco{y}{y'},X\fragmentco{x}{x'})=\edal{\mA}(X\fragmentco{x}{x'},Y\fragmentco{y}{y'}).
\]

\begin{lemma}\label{obs:drift}
    Every alignment $\mA:X\fragmentco{x}{x'}\onto Y\fragmentco{y}{y'}$ satisfies
    \[
        |(y'-y)-(x'-x)| \le \edal{\mA}(X\fragmentco{x}{x'},Y\fragmentco{y}{y'}).
    \]
\end{lemma}
\begin{proof}
    Each insertion and deletion changes the difference between the lengths of the aligned
    fragments by one, whereas each match and substitution leaves it unchanged.
\end{proof}

\subsubsection*{Self-Edit Distance of a String}
A measure of compressibility that will be instrumental in this paper is the
\emph{self-edit distance} of a string, recently introduced by Cassis, Kociumaka, and
Wellnitz~\cite{CKW23}.
An alignment $\mA : X \onto X$ is a \emph{self-alignment} if $\mA$ does not align any
character $X\position{x}$ to itself.
The \emph{self-edit distance} of $X$, denoted by $\selfed(X)$, is the minimum cost of a
self-alignment; formally, we set $\selfed(X) \coloneqq \min_{\mA} \edal{\mA} (X, X)$,
where the minimization ranges over all self-alignments $\mA : X \onto X$.
A small self-edit distance implies that we can efficiently encode the string.

\begin{lemma}\label{prp:self_ed_enc}
    Let $X$ be a string over $\Sigma$.
    Further, let $\fragment{a_1}{b_1}, \fragment{a_2}{b_2}, \ldots,
    \fragment{a_\ell}{b_\ell} \subseteq \fragmentco{0}{|X|}$ be pairwise disjoint. Then,
    we can encode $\{(x, X\position{x}) \mid x \in \bigcup_{i=1}^{\ell}
    \fragment{a_i}{b_i}\}$ using $\Oh(d \cdot \log (1+|\Sigma||X|/ d))$ bits, where $d
    \coloneqq \sum_{i=1}^{\ell} \selfed(X\fragment{a_i}{b_i})$.
\end{lemma}
\begin{proof}
    Without loss of generality, assume $a_1 \leq b_1 < a_2 \leq b_2 < \cdots < a_{\ell}
    \leq b_{\ell}$.
    For $i \in \fragment{1}{\ell}$, let $d_i \coloneqq \selfed(X\fragment{a_i}{b_i})$ and
    let $\mX_i : X\fragment{a_i}{b_i} \onto X\fragment{a_i}{b_i}$ be an optimal alignment
    of cost $d_i$.
    We may further assume that each point $(x,y)\in \mX_i$ satisfies $x \ge y$.%
    \footnote{If this is not the case, we write $\mX_i = (x_t,y_t)_{t=0}^m$ and observe
        that $\mA' = (\max(x_t,y_t), \min(x_t,y_t))_{t=0}^m$ is a self-alignment of the
        same cost.
        Geometrically, this means that we replace parts of $\mX_i$ below the main diagonal
        with their mirror image.}

    First, we claim that the edit information $\sE_{X,X}(\mX_i)$ can be encoded using two
    sequences of $\Oh(d_i)$ unique and non-decreasing indices in $\fragment{a_i}{b_i}$ and
    two sequences of $\Oh(d_i)$ characters from $\Sigma \cup \{\bot\}$, where $\bot$ is a
    special symbol representing the empty string $\varepsilon$. To see this, it suffices
    to notice that $\sE_{X,X}(\mX_i)$ is monotone in its first and third component.

    Next, we show that through $\sE_{X,X}(\mX_i)$ we can recover the set $\{(x,
    X\position{x}) \mid x \in \fragment{a_i}{b_i}\}$.
    To this end, notice that $\sE_{X,X}(\mX_i)$ allows us to recover a partition of
    $\fragment{a_i}{b_i}$ into individual positions that $\mX_i$ deletes or substitutes,
    and maximal fragments that $\mX_i$ matches perfectly.
    Observe that for any fragment $X\fragmentco{x}{x'}$ that $\mX_i$ matches perfectly to
    $X\fragmentco{y}{y'}$, two conditions hold: the lengths are identical ($y'-y = x'-x$)
    and the perfect match precedes the current position ($y < x$).
    The latter holds because $(x,y) \in \mX_i$, $x \ge y$, and $\mX_i$ does not match
    $X\fragmentco{x}{x'}$ with itself.
    This strictly prior dependency allows us to recover the sequence iteratively from left
    to right.
    Assuming we have recovered the prefix $\{(x, X\position{x}) \mid x \in
    \fragmentco{a_i}{x}\}$, we can extend the recovery to $x'$: if the partition
    corresponds to an explicit character (insertion/deletion), we assign it directly; if
    it corresponds to a perfect match, we copy the values from the previously recovered
    segment $X\fragmentco{y}{y'}$.\footnote{Note that what we are describing is nothing
    else than a \emph{LZ77-like factorization}.}

    By combining the encodings for all $i \in \fragment{1}{\ell}$, we can describe the set
    $\{(x, X\position{x}) \mid x \in \bigcup_{i=1}^{\ell} \fragment{a_i}{b_i}\}$ using a
    sequence of non-decreasing indices from $\fragment{0}{|X|}$ of size
    $\Oh(\sum_{i=1}^{\ell} d_i) \le \Oh(d)$ (here we use the disjointness of the
    $\fragment{a_i}{b_i}$), together with $\Oh(d)$ characters from $\Sigma \cup \{\bot\}$.
    Moreover, $\selfed(Y) \leq 2|Y|$ holds for every string $Y$, so by the disjointness of
    the intervals we have
    \(
        d = \sum_{i=1}^{\ell} \selfed(X\fragment{a_i}{b_i}) \leq 2 \sum_{i=1}^{\ell} |X\fragment{a_i}{b_i}| \leq 2|X|
    \).

    Consequently, $\log \binom{|X|+d}{d} = \Oh(d \log(1+|X|/d)) \le \Oh(d
    \log(1+|\Sigma||X|/d))$. Moreover, since $d \leq 2|X|$, we have $1+|\Sigma||X|/d \geq
    1+|\Sigma|/2$, and therefore $\log |\Sigma| = \Oh(\log(1+|\Sigma||X|/d))$. Thus,
    $\Oh(\log \binom{|X|+d}{d} + d \log |\Sigma|) \le \Oh(d \log(1+|\Sigma||X|/d))$ bits
    suffice.
\end{proof}

We use the following known properties of $\selfed$.

\begin{lemmaq}[Properties of $\selfed$, {\cite[Lemma~4.2]{CKW23}}] \label{prp:prop_selfed}
    For any string  $X$:
    \begin{description}
        \item[Monotonicity.] For any $\ell' \leq \ell \leq r \leq r' \in \fragment{0}{|X|}$, we have
        \[
            \selfed(X\fragmentco{\ell}{r}) \leq \selfed(X\fragmentco{\ell'}{r'}).
        \]
        \item[Sub-additivity.] For any $m \in \fragment{0}{|X|}$, we have
        \[
            \selfed(X) \leq \selfed(X\fragmentco{0}{m}) + \selfed(X\fragmentco{m}{|X|}).
        \]
        \item[Triangle inequality.] For any string $Y$, we have
        \[
            \selfed(Y) \leq \selfed(X) + 2\ed(X,Y).
            \qedhere
        \]
    \end{description}
\end{lemmaq}

Cassis, Kociumaka, and Wellnitz~\cite{CKW23} also proved the following lemma that bounds
the self-edit distance of $Y$ in the presence of two disjoint alignments mapping $Y$ to
nearby fragments of another string~$X$.

\begin{lemmaq}[{\cite[Lemma~4.5]{CKW23}}]\label{prp:edimpliesselfed}
    Consider strings $X, Y$ and alignments $\A : Y \onto X\fragmentco{i}{j}$ and $\A' : Y \onto X\fragmentco{i'}{j'}$.
    If there is no $(y, x) \in \A \cap \A'$  such that both $\A$ and $\A'$ match $Y\position{y}$ with $X\position{x}$, then
    \[
        \selfed(Y) \leq |i - i'| + \edal{\A}(Y, X\fragmentco{i}{j}) + \edal{\A'}(Y, X\fragmentco{i'}{j'}) + |j-j'|.
        \qedhere
    \]
\end{lemmaq}

\subsubsection*{Periodicity Lemma}
For completeness, we restate Fine and Wilf's Periodicity Lemma~\cite{FW65}.

\begin{lemmaq}[Periodicity Lemma,~\cite{FW65}] \label{lem:perlemma}
    If \(p, q\) are periods of a string \(X\) of length \(|X| \geq p + q - \gcd(p, q)\),
    then \(\gcd(p, q)\) is a period of \(X\).
\end{lemmaq}

\section{Full Proof}\label{sec:full}

In this section, we provide a refined and comprehensive version of the proofs
from \cite{KNW24}. For each theorem, we either introduce a new proof if it
is essential for our improvements, or restate the original proof from
\cite{KNW24} with targeted clarifications and corrections where necessary.
If proofs are omitted from this section, then they have already been given in the main body of the paper.

To this end, we fix two strings $P \in \Sigma^m$ and $T \in \Sigma^n$, and a positive
threshold $k$. Furthermore, we let $S$ be a set of alignments of fragments of $P$ onto
fragments of $T$ with cost at most $k$, and we denote by $\cost(S) \coloneqq \sum_{\mX \in
S} \cost(\mX)$ the sum of all costs of the alignments contained in $S$.

\subsection{The Inference Graph $\bG_S$ of a Set of Alignments}
\label{sec:sec_gs}

We begin by defining the \emph{inference graph} $\bG_S$ (refer to \cref{fig:ex_def,fig:2}
for examples).

\bg*

Note that every vertex within a black component corresponds to a character
from $P$ or~$T$.
Furthermore, all characters within a single black
component are identical. This follows from the fact that a black edge
indicates that an alignment in $S$ matches the two specific characters. Because a black component is formed
exclusively by these edges, transitivity ensures that all characters
represented in a given component must be the same.

\begin{figure}[thbp]
    \renewcommand\tabularxcolumn[1]{m{#1}}
    \centering
    \begin{tabularx}{\linewidth}{*{2}{>{\centering\arraybackslash}X}}
        \includegraphics[page=1, scale=.8]{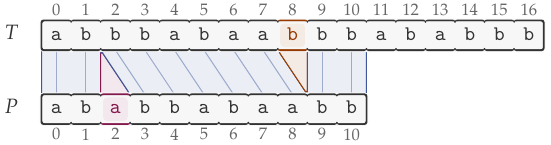}
        &
        \includegraphics[page=4, scale=.74]{figures/g01.pdf}
        \\[-2ex]
        \begin{subfigure}[t]{\linewidth}
            \caption{%
                An alignment $\mA_1 : P \onto T\fragmentco{0}{11}$
                of cost 2.}
        \end{subfigure}
        &
        \begin{subfigure}[t]{\linewidth}
            \caption{%
                The inference graph $\bG_{S_{1}}$ for $S_{1} \coloneqq \{\mA_1\}$ has 16
                black components and one red component.
                Removing from \(P\) and \(T\) the
                character red (that correspond to vertices of~$\bG_{S_{1}}$ of red
                components), one obtains the strings
                $T_{|S_{1}} = \mathbf{\mathtt{abbbabaabbababbb}}$
                and
                $P_{|S_{1}}= \mathbf{\mathtt{abbbabaabb}}$.
            }
        \end{subfigure}\\
        &\\
        \includegraphics[page=2, scale=.8]{figures/g01.pdf}
        &
        \includegraphics[page=22, scale=.74]{figures/g01.pdf}
        \\[-2ex]
        \begin{subfigure}[t]{\linewidth}
            \caption{%
                An alignment $\mA_2 : P \onto T\fragmentco{6}{17}$
                of cost 3.}
        \end{subfigure}
        &
        \begin{subfigure}[t]{\linewidth}
            \caption{%
                The inference graph $\bG_{S_{12}}$ for $S_{12} \coloneqq \{\mA_1, \mA_2\}$
                has five  black components and two red components;
                we highlight vertices that are red due to \(\mA_2\).
                Removing from \(P\) and \(T\) the
                character red (that correspond to vertices of~$\bG_{S_{12}}$ of red
                components), one obtains the strings
                $T_{|S_{12}} = \mathbf{\mathtt{abbababbababb}}$
                and
                $P_{|S_{12}}= \mathbf{\mathtt{abbababb}}$.
                For $k=2$ we have that $S_{12}$ encloses $T$.
            }\label{fig:1d}
        \end{subfigure}\\
        &\\
        \includegraphics[page=3, scale=.8]{figures/g01.pdf}
        &
        \includegraphics[page=30, scale=.74]{figures/g01.pdf}
        \\[-2ex]
        \begin{subfigure}[t]{\linewidth}
            \caption{%
                An alignment $\mA_3 : P\fragmentco{3}{6} \onto T\fragmentco{5}{7}$
                of cost 1.}
        \end{subfigure}
        &
        \begin{subfigure}[t]{\linewidth}
            \caption{%
                The inference graph $\bG_{S_{123}}$ for $S_{123} \coloneqq \{\mA_1, \mA_2,
                \mA_3\}$
                has two  black components and two red components;
                we highlight vertices that are red due to \(\mA_3\).
                Removing from \(P\) and \(T\) the
                character red (that correspond to vertices of~$\bG_{S_{123}}$ of red
                components), one obtains the strings
                $T_{|S_{123}} = \mathbf{\mathtt{ababababab}}$
                and
                $P_{|S_{123}}= \mathbf{\mathtt{ababab}}$.
                For $k=2$ we have that $S_{123}$ encloses $T$, but $S_{123}$ does not
            succinctly enclose $T$.}\label{fig:1f}
        \end{subfigure}
    \end{tabularx}
    \caption{Three alignments of fragments of
        $P = \mathtt{ababbabaabb}$ onto fragments of $T = \mathtt{abbbabaabbbababbb}$ are consecutively
    added to the corresponding inference graph from \cref{def:bg}.}\label{fig:ex_def}
\end{figure}
\begin{figure}[thbp]
    \begin{subfigure}[t]{\linewidth}
        \centering
        \includegraphics[page=23, scale=.74]{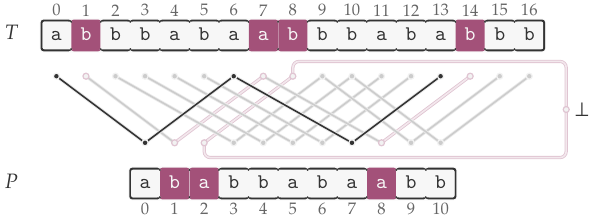}~%
        \includegraphics[page=24, scale=.74]{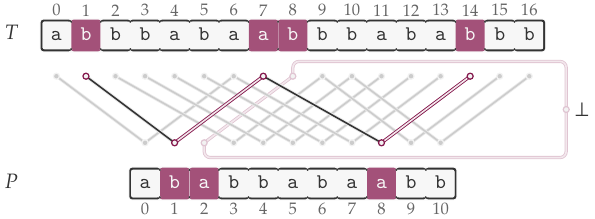}\\[2ex]
        \includegraphics[page=25, scale=.74]{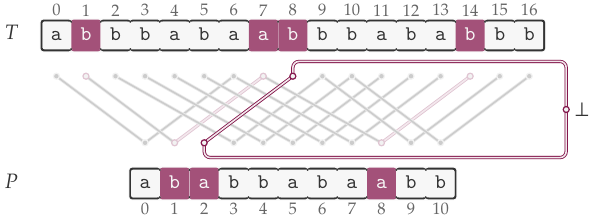}~%
        \includegraphics[page=26, scale=.74]{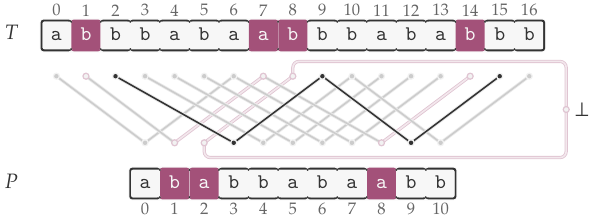}\\[2ex]
        \includegraphics[page=27, scale=.74]{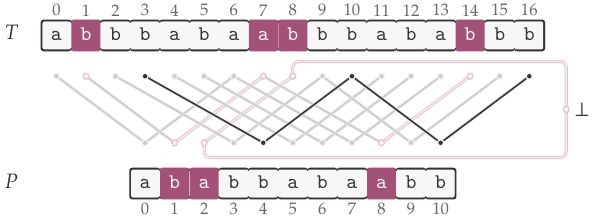}~%
        \includegraphics[page=28, scale=.74]{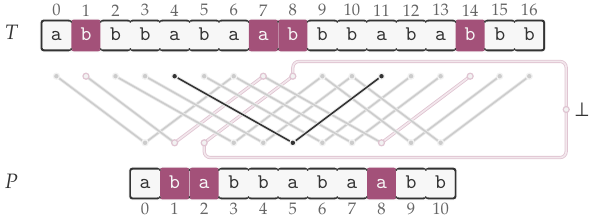}\\[2ex]
        \includegraphics[page=29, scale=.74]{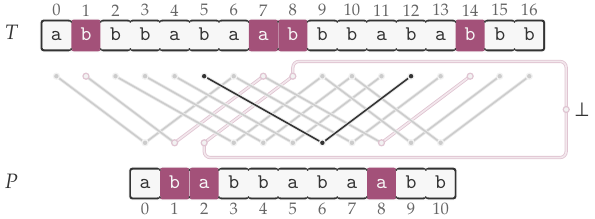}
        \caption{%
            The seven connected components of $\bG_{S_{12}}$ from \cref{fig:1d}.
        }
    \end{subfigure}\\[2ex]
    \begin{subfigure}[t]{\linewidth}
        \centering
        \includegraphics[page=31, scale=.74]{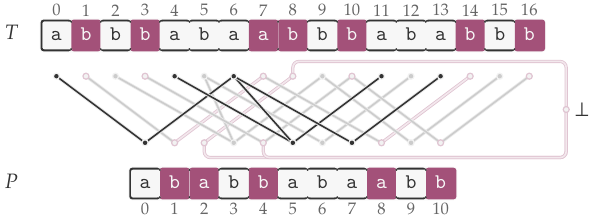}~%
        \includegraphics[page=32, scale=.74]{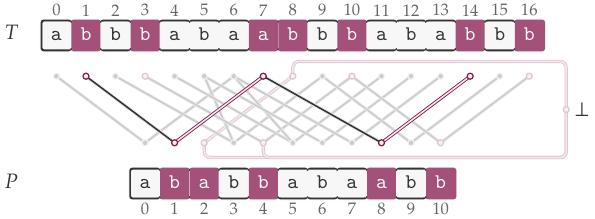}\\[2ex]
        \includegraphics[page=33, scale=.74]{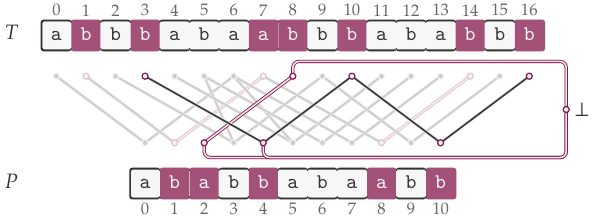}~%
        \includegraphics[page=34, scale=.74]{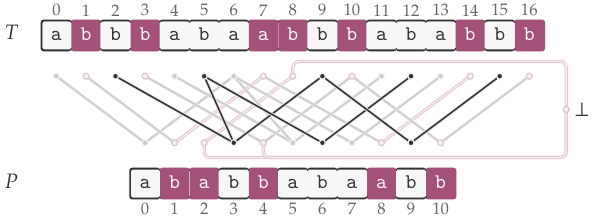}
        \caption{%
            The four connected components of $\bG_{S_{123}}$ from \cref{fig:1f}.
        }
    \end{subfigure}
    \caption{The connected components of the graphs of \cref{fig:ex_def} in
    detail.}\label{fig:2}
\end{figure}

The inference graph $\bG_S$ is encoded implicitly via the edit information $\sE_{P,
T}(\mX)$ for each alignment $\mX: P\fragmentco{p}{p'} \onto T\fragmentco{t}{t'}$ in $S$.

\encbg*

From \cref{prp:enc_gs} we can directly infer that \( \bc(\bG_S) = 0 \) is an easy case to
handle.

\subsubsection*{\boldmath The (easy) case: \( \bc(\bG_S) = 0 \)}
Suppose that \( \bc(\bG_S) = 0 \) and that, in the inference graph $\bG_S$, each character
of $P$ and $T$ in $\bc(\bG_S)$ has at least one outgoing edge (we will elaborate later how
to ensure this).
Then, each character is in a red component.
By \cref{prp:enc_gs}, we can fully recover \( P\) and \(T \) from $\{\sE_{P, T}(\mX) \mid
\mX \mid P\fragmentco{p}{p'} \onto T\fragmentco{t}{t'} \in S\}$.

\subsubsection*{\boldmath The (challenging) case: \( \bc(\bG_S) \neq 0 \)}
For the remainder of this section, unless stated otherwise, we assume that \( \bc(\bG_S) >
0 \), as this is the case that remains to be addressed.
In this case, we carefully select \( S \) so that a periodic structure can be observed in
two additional strings, \( T_{|S} \) and \( P_{|S} \). These strings are obtained by
retaining only the characters that belong to black connected components.

\tsps*

The key observation is that the black edges of \(\bG_S\) that are induced by a single
alignment \(\mX \in S\) correspond to an \emph{exact occurrence of a fragment of \( P_{|S}
\) in \( T_{|S} \)}.

\occ
\begin{proof}
    By \cref{def:bg}, for every character of $P_{|S}$ contained in $P\fragmentco{x}{x'}$,
    the alignment $\mX$ induces a black edge in the inference graph $\bG_S$ whose other
    endpoint belongs to the same black connected component of $\bG_S$.
    The other endpoint is therefore a character of $T_{|S}$ contained in
    $T\fragmentco{y}{y'}$.
    By symmetry and since the black edges induced by $\mX$ form a matching, we conclude
    that $\mX$ induces a perfect matching between the characters of $P_{|S}$ contained in
    $P\fragmentco{x}{x'}$ and the characters of $T_{|S}$ contained in
    $T\fragmentco{y}{y'}$.
    These characters appear consecutively in $P_{|S}$ and $T_{|S}$, respectively, and form
    fragments $P_{|S}\fragmentco{x_\mX}{x'_\mX}$ and
    $T_{|S}\fragmentco{y_{\mX}}{y'_{\mX}}$.
    Moreover, $\mX$ is non-crossing, so the perfect matching induced by $\mX$ consists of
    edges between $P_{|S}\position{x_{\mX} + p}$ and $T_{|S}\position{y_{\mX} + p}$ for
    every $p\in \fragmentco{0}{x_\mX' - x_\mX}$.
\end{proof}

In the following definition, we give the conditions on $S$ needed to observe the periodic
structure.

\enclose

This last property the we assume about \( S \), that is, \( S \) succinctly enclosing \( T \), is
designed to ensure that the number of alignments in \( S \) remains small while still
guaranteeing that the period induced in \( P_{|S} \) and \( T_{|S} \) by the exact matches
decreases exponentially with the size of $S$.

If \(S\) encloses \(T\), then each of the alignments \(\mXpref, \mXsuf, \mA_1, \dots,
\mA_a\) aligns the entire string \(P\) with a fragment of \(T\), and the edges they induce
connect to all vertices of \(P_{|S}\). Consequently, we have \(x_{\mXpref} = x_{\mXsuf} =
x_{\mA_1} = \dots = x_{\mA_a} = 0\) and \(x_{\mXpref}' = x_{\mXsuf}' = x_{\mA_1}' = \dots
= x_{\mA_a}' = m_S\).
Moreover, note that whenever \(S\) succinctly encloses \(T\) and $S$ is degenerate, then
\(b = 0\), since the condition \eqref{def:enclose:d} cannot hold for $i = 1$ if $b > 0$.

We can also prove something more about the exact matching induced by \(\mXpref\) and
\(\mXsuf\).

\nosingle
\begin{proof}
    Every alignment $\mX\in S$ induces a matching between $P_{|S}$ and the characters of
    $T_{|S}$ contained in the image $\mX(P)$, so we must have $y_{\mXpref}=0$ and
    $y_{\mXsuf}=n_S-m_S$.
    The costs of $\mXpref(P)$ and $\mXsuf(P)$ are at most $k$, so $|\mXpref(P)|\ge |P|-k$
    and $|\mXsuf(P)|\ge |P|-k$.
    Due to the assumption $|T|\le 2|P|-2k$, this means that every character of $T$ is
    contained in $\mXpref(P)$ or $\mXsuf(P)$.
    Consequently, the two matchings induced by these alignments jointly cover $T_{|S}$,
    and thus $n_S\le 2m_S$.
\end{proof}

As a side remark, in the case where \( \bc(\bG_S) = 0 \), the presence of \( \mXpref \)
and \( \mXsuf \) in \( S \) ensures that
each character of $P$ and $T$ in $\bG_S$ has at least one outgoing edge.

\Cref{clm:nosingle} also implies that when $S$ is degenerate, we have $m_S = n_S$, and
every alignment $\mX \in {\mXpref, \mXsuf, \mA_1, \ldots, \mA_a}$ induces in the inference
graph $\bG_S$ an edge $(P_{|S}[p], T_{|S}[t])$ if and only if $p = t$.
In this case, for each $p \in \fragmentco{0}{m_S}$, there is a black connected component
consisting solely of $P_{|S}[p]$ and $T_{|S}[p]$.
We now proceed to prove that the black connected components are structured even when $S$
is not degenerate.

\gsi

\periodicity*

Next, we demonstrate that \cref{lem:periodicity} remains valid after removing the
rightmost edge from each alignment. This refinement is required for
technical reasons in the subsequent sections of this proof.

\newcommand{\bbG}{\bar{\bG}}
\newcommand{\bcc}{\bar{\rm{cc}}}

\begin{lemma}\label{lem:periodicity_gbar}
    Suppose $S$ succinctly encloses $T$. Consider the graph $\bbG_S$ obtained from $\bG_S
    = \bG^b_S$ by removing nodes $P_{|S}\position{m_S-1}$ and $T_{|S}\position{n_S-1}$ and
    the rightmost edge induced by $\mX$ for each $\mX \in S$, that is, the edge that
    connects $P_{|S}\position{x_{\mX}'-1}$ with $T_{|S}\position{y_{\mX}' - 1}$. Set
    \[
        \bar{g}_b \coloneqq
        \begin{cases}
            m_S - 1 & \text{if $S$ is degenerate}, \\
            g_b & \text{otherwise}.
        \end{cases}
    \]
    Then, for each $c \in \fragmentco{0}{\bar{g}_b}$, the graph $\bbG_S$ has a connected
    component with node set
    \[
        \bar{C}_c \coloneqq \{P_{|S}\position{j} : j\in \fragmentco{0}{m_S-1} \text{ and } j
        \equiv_{\bar{g}_b} c\} \cup \{T_{|S}\position{j} : j\in \fragmentco{0}{n_S-1} \text{ and }
     j \equiv_{\bar{g}_b} c\}.
    \]
\end{lemma}

\begin{proof}
    Similarly to \cref{lem:periodicity}, we only focus on the case where $S$ is not
    degenerate.

    Since $\bbG_S$ is a subgraph of $\bG_S = \bG_S^b$, the description of the connected
    components of $\bG_S$ from \cref{lem:periodicity}\eqref{lem:periodicity:i} implies
    that $\bbG_S$ contains no edge leaving the subgraph induced by $\bar{C}_c$.
    Therefore, it suffices to show that the subgraph induced by $\bar{C}_c$ is connected.

    Let us define strings $P_{\bcc}$ and $T_{\bcc}$ analogously to $P_{\cc}$ and $T_{\cc}$
    in the proof of \cref{lem:periodicity}\eqref{lem:periodicity:i}, but with respect to
    $\bbG_S$ rather than $\bG_S$.
    Specifically, $|P_{\bcc}|=m_S-1$ and $|T_{\bcc}|=n_S-1$, and the characters
    $P_{\bcc}[j]$ and $T_{\bcc}[i]$ are equal if and only if the corresponding characters
    $P_{|S}[j]$ and $T_{|S}[i]$ belong to the same connected component of $\bbG_S$.

    Note that, if an alignment $\mX \in S$ induces in $\bG_S$ a perfect matching between
    the fragments $P_{|S}\fragmentco{x_{\mX}}{x'_{\mX}}$ and
    $T_{|S}\fragmentco{y_{\mX}}{y'_{\mX}}$, then in $\bbG_S$ it still induces a perfect
    matching between the shortened fragments $P_{|S}\fragmentco{x_{\mX}}{x'_{\mX}-1}$ and
    $T_{|S}\fragmentco{y_{\mX}}{y'_{\mX}-1}$. Consequently, the corresponding fragments
    $P_{\bcc}\fragmentco{x_{\mX}}{x'_{\mX}-1}$ and
    $T_{\bcc}\fragmentco{y_{\mX}}{y'_{\mX}-1}$ are equal.

    Therefore, using \cref{clm:occ} together with the preceding observation, each of the
    alignments $\mXpref,\mXsuf,\mA_1,\ldots,\mA_a$ still induces an exact occurrence of
    $P_{\bcc}$ in $T_{\bcc}$ at the same shift as before, and each $\mB_i$ still induces a
    matching fragment shortened by one on the right.
    Hence, we next adapt the proof of \cref{lem:periodicity}\eqref{lem:periodicity:i} with minor modifications.

    By \cref{clm:occ,clm:nosingle}, we have
    $\{0,|T_{\bcc}| - |P_{\bcc}|\} \subseteq \{y_{\mXpref}, y_{\mXsuf}, y_{\mA_1}, \ldots,
    y_{\mA_a}\} \subseteq \Occ(P_{\bcc},T_{\bcc})$ and
    $|T_{\bcc}| = n_S - 1 \leq 2m_S - 1 = 2|P_{\bcc}| + 1$.
    If \( |P_{\bcc}| = 0 \), then \( |T_{\bcc}| \leq 1 \), and \( \bbG_S \) is either
    empty or consists of a single vertex constituting the only connected component.
    Otherwise, \cref{fct:periodicity} implies that \( \gcd(\Occ(P_{\bcc},T_{\bcc})) \) is
    a period of \( T_{\bcc} \), and thus also \( g_0 \) is a period of both \( T_{\bcc} \)
    and its prefix \( P_{\bcc} \).
    We can thus proceed to prove inductively that \( T_{\bcc} \) and \( P_{\bcc} \) have
    period \( g_i \). Note that we can still apply \cref{fct:periodicity_extended}, as
    instead of \( P_{\cc}\fragmentco{x_{\mB_i}}{x_{\mB_i}'} \), we apply it to \(
    P_{\bcc}\fragmentco{x_{\mB_i}}{x_{\mB_i}'-1} \).
    Thus, the length \( g_{i-1} + 1 \) decreases by at most one, under which we can still
    apply \cref{fct:periodicity_extended}.
    Consequently, the subgraph induced by \( \bar{C}_c \) is a connected subgraph of \(
    \bbG_S \).
\end{proof}

We have observed that the assumption that \( S \) succinctly encloses \( T \) induces a
well-defined periodic structure in \( P_{|S} \) and \( T_{|S} \). However, the definition
of succinct enclosure, as stated, has a drawback: it does not naturally support an
iterative construction.
Ideally, we would like to maintain an \( S \) that succinctly encloses \( T \) while
adding a new alignment \( \mY \) such that \( S \cup \{ \mY \} \) still succinctly
encloses \( T \). However, verifying this is challenging, as condition
\eqref{def:enclose:d} of \cref{def:enclose} depends on the final form of \( S \). The
following lemma helps to address this difficulty.

\additer*

For the remaining part of this section, up to \cref{sec:recocc},
we assume that $S$ succinctly encloses $T$.
By \cref{lem:periodicity}\eqref{lem:periodicity:i}, this allows us, for the sake of
convenience, to index black connected components with integers in
$\fragmentco{0}{\bc(\bG_S)}$ and to introduce some additional notation.

\begin{definition}\label{def:pitau}
    We define the following quantities.
    \begin{itemize}
        \item For $c \in \fragmentco{0}{\bc(\bG_S)}$, we define the \emph{$c$-th black
            connected component} as the black connected component containing $P_{|S}[c]$
            and set
        \[
        m_c \coloneqq \left\lceil\frac{m_S - c}{\bc(\bG_{S})} \right\rceil
        \quad\text{and}\quad n_c := \left\lceil \frac{n_S - c}{\bc(\bG_{S})} \right\rceil,
        \]
        as the number of characters in $P$ and $T$, respectively, that belong to the
        $c$-th black connected component.
        \item We define $c_{\last} \in \fragmentco{0}{\bc(\bG_{S})}$ as the black
            component containing the last characters of $P_{|S}$ and $T_{|S}$.
        \item We define $\pi(j) \in \fragmentco{0}{|P|}$ as the position of
            $P_{|S}\position{j}$ in $P$.
            Similarly, $\tau(i) \in \fragmentco{0}{|T|}$ is the position of
            $T_{|S}\position{i}$ in $T$.
        \item Starting from \cref{sec:close}, it will be useful to index the positions
            w.r.t. the remainder and quotient of $\bc(\bG_S)$. To this end, we set
            $\pi_{j}^c \coloneqq \pi(c + j \cdot \bc(\bG_S))$ for $c \in
            \fragmentco{0}{\bc(\bG_S)}$ and $j \in \fragmentco{0}{m_c}$. Similarly, we set
            $\tau_{i}^c \coloneqq \tau(c + i \cdot \bc(\bG_S))$ for $c \in
            \fragmentco{0}{\bc(\bG_S)}$ and $i \in \fragmentco{0}{n_c}$.
        \item For the sake of convenience, we extend some notation, namely $m_c$, $n_c$,
            $\pi_j^c$ for $j \in \fragmentco{0}{m_c}$, and $\tau_i^c$ for $i \in
            \fragmentco{0}{n_c}$, to $c = \bc(\bG_S)$ using the same formulas as above.
            \qedhere
    \end{itemize}
\end{definition}

\begin{remark} \label{rmk:perstr}
    We conclude this (sub)section with some observations.
    \begin{enumerate}
        \item Consider $c \in \fragmentco{0}{\bc(\bG_{S})}$.
            Notice that the characters $\{T\position{\tau(i)} \mid i \in
            \fragmentco{0}{n_S}, i \equiv_{\bc(\bG_S)} c\} = \{T\position{\tau_i^c} \mid i
        \in \fragmentco{0}{n_c}\}$ and $\{P\position{\pi(j)} \mid j \in
    \fragmentco{0}{m_S}, j \equiv_{\bc(\bG_S)} c\} = \{P\position{\pi_j^c} \mid j \in
\fragmentco{0}{m_c}\}$ are precisely those within the $c$-th black connected component.
            Consequently, they are all identical.
        \item For all $c \in \fragmentco{0}{\bc(\bG_{S})}$, we have $n_c \in \{n_0, n_0 - 1\}$ and $m_c \in \{m_0, m_0 - 1\}$.
        \item The indices $\tau$ and $\pi$ induce a partition of $T$ and $P$ into
            substrings. More precisely (assuming an empty concatenation is interpreted as
            the empty string),
            \[
                T\fragmentco{0}{\tau(0)} \circ \left( \bigcirc_{i=0}^{n_S-2}\;
                T\fragmentco{\tau(i)}{\tau(i+1)} \right) \circ
                T\fragmentco{\tau(n_S-1)}{|T|}
            \]
            and
            \[
                P\fragmentco{0}{\pi(0)} \circ \left( \bigcirc_{j=0}^{m_S-2}\;
                P\fragmentco{\pi(j)}{\pi(j+1)} \right) \circ
                P\fragmentco{\pi(m_S-1)}{|P|}.
            \]
        \item Recall that the inference graph $\bG_S$ contains an edge between
            $P\position{\pi(j)}$ and $T\position{\tau(i)}$ with $[i]_b = [j]_b$ only if
            there is an alignment $\mX\in S$ such that $(\pi(j),\tau(i))\in \mX$.
            If the edge between $P\position{\pi(j)}$ and $T\position{\tau(i)}$ is also
            contained in $\bbG$ (which happens if and only if $j<m_S-1$ and $i<n_S-1$), we
            also have $(\pi(j+1),\tau(i+1))\in \mX$.
            \qedhere
            \label{rmk:perstr:4}
    \end{enumerate}
\end{remark}

\subsection{Covering Weight Functions}

\begin{definition}[{\cite[Definition 4.9]{KNW24}}]
    \label{def:funcover}
    Abbreviate $\bc(\bG_S)$ with $\bc$.
    We say that a \emph{weight function} $\w_S : \fragmentco{0}{\bc}\to
    \mathbb{Z}_{\ge 0}$ \emph{covers $S$} if all of the following conditions hold:
    \begin{enumerate}
        \item for all $c \in \fragmentco{0}{\bc}$, all $j\in \fragmentco{0}{m_S-1}$ and
            $i\in \fragmentco{0}{n_S-1}$ such that $i \equiv_{\bc} j \equiv_{\bc} c$,
            we have  \label{it:funcover:1}
            \begin{align}
                \label{eq:funcover}
                \w_S(c) \ge \ed(P\fragmentco{\pi(j)}{\pi(j+1)},T\fragmentco{\tau(i)}{\tau(i+1)});
            \end{align}
        \item $\w_S(\bc-1) \geq \ed(P\fragmentco{0}{\pi(0)}, T\fragmentco{0}{\tau(0)})$;
            \label{it:funcover:2}
        \item $\w_S(\bc-1)\ge \ed(P\fragmentco{0}{\pi(0)},T\fragmentco{t}{\tau(i)})$ holds
            for every $i \in \fragmentoo{0}{n_S}$ such that $i \equiv_{\bc} 0$ and some
            $t\in \fragment{\tau(i-1)}{\tau(i)}$;
            \label{it:funcover:3}
        \item $\w_S(c_{\last}) \geq \ed(P\fragmentco{\pi(m_S-1)}{|P|},
            T\fragmentco{\tau(n_S-1)}{|T|})$; and
            \label{it:funcover:4}
        \item $\w_S(c_{\last}) \ge \ed(P\fragmentco{\pi(m_S-1)}{|P|},
            T\fragmentco{\tau(i)}{t'})$ holds for every $i \in \fragmentco{0}{n_{S}-1}$
            such that $i \equiv_{\bc} c_{\last}$ and some $t'\in
            \fragment{\tau(i)}{\tau(i+1)}$.
            \label{it:funcover:5}
    \end{enumerate}
    We define $\sum_{c=0}^{\bc-1} \w_S(c)$ to be the \emph{total weight} of $\w_S$.
    Moreover, we define $\w_S(-1) \coloneqq \w_S(\bc-1)$.
\end{definition}

In this (sub)section, we describe the construction of a weight function $\w_S$ covering
$S$ with a total weight of $\Oh(\cost(S))$. To facilitate a
more concise presentation of \cref{alg:construction_w}, we first introduce a preliminary
definition that simplifies the exposition of the construction process.

\begin{algorithm}[t]
\KwInput{Two strings $P$ and $T$, a threshold $k$, and a set $S$ that succinctly encloses $T$.}
\KwOutput{A weight function $\w_S$ that covers $S$.}
Compute $\bc(\bG_S)$ and functions $\Pi_P, \Pi_T$\;
Initialize the weight function $\w_S(c) = 0$ for all $c\in \fragmentco{0}{\bc(\bG_S)}$\;
\For{$\mX : P\fragmentco{x}{x'} \onto T\fragmentco{y}{y'} \in S$}{
    \For{each edit $(\hat{x}, \hat{y})$ performed by $\mX : P\fragmentco{x}{x'} \onto T\fragmentco{y}{y'}$}{
        \If{$\mX$ deletes $P\position{\hat{x}}$ or substitutes $P\position{\hat{x}}$ with $T\position{\hat{y}}$}{
           \CommentSty{//charge edit to the first black connected component to the left of $\hat{x}$\\}
            \If{$\Pi_P(\hat{x}) = \bot$}{Increase $\w_S(\bc(\bG_S)-1)$ by one\;}
            \Else{Increase $\w_S(c)$ by one, where $c$ is the black component of $\Pi_P(\hat{x})$\;}
        }\ElseIf{$\mX$ inserts $T\position{\hat{y}}$}{
               \CommentSty{//charge edit to the first black connected component to the left of $\hat{y}$\\}
            \If{$\Pi_T(\hat{y}) = \bot$}{Increase $\w_S(\bc(\bG_S)-1)$ by one\;}
            \Else{Increase $\w_S(c)$ by one, where $c$ is the black component of $\Pi_T(\hat{y})$\;}
        }
    }
}
 \Return{$\w_S$}\;
\caption{A deterministic algorithm to construct $\w_S$.}\label{alg:construction_w}
\end{algorithm}

\begin{definition}[{\cite[Definition 4.10]{KNW24}}]
    Given a position $x \in \fragmentco{0}{|P|}$, we define a function $\Pi_P(x)$ that
    maps $x$ to the largest position in $P$ that is less than or equal to $x$ and belongs
    to a black connected component. If no such character exists in $P$ preceding $x$, then
    $\Pi_P(x) = \bot$. We define a symmetric function $\Pi_T(y)$ for all positions $y \in
    \fragmentco{0}{|T|}$.
\end{definition}

\begin{lemma}[Improvement of {\cite[Lemma 4.11]{KNW24}}] \label{lem:exfuncover}
    \Cref{alg:construction_w} yields $\w_S : \fragmentco{0}{\bc(\bG_S)}\to
    \mathbb{Z}_{\ge 0}$ that covers $S$ of total weight $\mathcal{O}(\cost(S))$.
\end{lemma}

\begin{proof}
    The total weight of $\w_S$ is clearly $\Oh(\cost(S))$ because each edit in an
    alignment increases the total weight of $\w_S$ by at most 1. To demonstrate that
    $\w_S$ covers $S$, we first construct another function $\w_S' :
    \fragmentco{0}{\bc(\bG_S)} \to \mathbb{Z}_{\ge 0}$. We show that $\w_S'$ covers $S$
    and satisfies $\w_S'(c) \leq \w_S(c)$ for all $c \in \fragmentco{0}{\bc(\bG_S)}$.
    Since the properties defined in \cref{def:funcover} are all componentwise monotone, it
    follows that $\w_S$ also covers $S$.

    We begin by considering the graph \( \bbG_S \) as defined in
    \cref{lem:periodicity_gbar}.
    By \cref{lem:periodicity_gbar} the graph $\bbG_S$ contains for all $c \in
    \fragmentco{0}{\bc(\bbG_S)}$ a connected component $\bbG_S(c)$ induced by vertices
    $P\position{\pi(j)}$ and $T\position{\tau(i)}$ with $j\in \fragmentco{0}{m_S-1}$,
    $i\in \fragmentco{0}{n_S-1}$ and $i \equiv_{\bc(\bbG_S)} j \equiv_{\bc(\bbG_S)} c$.
    Note, by \cref{lem:periodicity_gbar}, we have that $\bc(\bbG_S) = \bc(\bG_S)$,
    if \(S\) is not degenerate and $\bc(\bbG_S) = \bc(\bG_S) - 1$, otherwise. Let us
    assign weights to edges of $\bbG_S$.
    If $\bbG_S$ contains an edge between $P\position{\pi(j)}$ and $T\position{\tau(i)}$,
    we define its weight to be $\edal{\mX}(P\fragmentco{\pi(j)}{\pi(j+1)}, T\fragmentco{\tau(i)}{\tau(i+1)})$
    where $\mX$ is any $\mX \in S$ such that $(\pi(j), \tau(i)) \in \mX$.
    (Note that all of this is indeed well-defined because of how we defined $\bbG_S$.)
    By $\bar{\w}(c)$, we denote the total weight of edges in $\bbG_S(c)$ for $c \in \fragmentco{0}{\bc(\bbG_S)}$.

    Now, we define
    \[
    \w_S'(c) = \bar{\w}(c) \cdot \mathbf{1}[c\neq \bc(\bG_S)-1 \text{ or } \bc(\bG_S) =
    \bc(\bbG_S)] + \alpha\cdot \mathbf{1}[c=\bc(\bG_S)-1]+\alpha'\cdot
    \mathbf{1}[c=c_{\last}],
    \]
    where $\mathbf{1}[\phi]$ denotes the indicator function (equal to $1$ if the statement
    $\phi$ holds and to $0$ otherwise), and $\alpha, \alpha'$ are defined as follows.
    \begin{itemize}
        \item We set
            \[
                \alpha \coloneqq \edal{\mXpref}(P\fragmentco{0}{\pi(0)}, T\fragmentco{0}{\tau(0)})
                + \edal{\mXsuf}(P\fragmentco{0}{\pi(0)}, T\fragmentco{y}{\tau(\hi)}),
            \]
            assuming $\mXsuf$ aligns $P\fragmentco{0}{\pi(0)}$ onto
            $T\fragmentco{y}{\tau(\hi)}$ for some $\hi \in \fragmentco{0}{n_S}$ with $\hi
            \equiv_{\bc(\bG_S)} 0$ and $y \in \fragment{0}{\tau(\hi)}$.
        \item Similarly, we set
            \[
                \alpha' \coloneqq \edal{\mXsuf}(P\fragmentco{\pi(m_S-1)}{|P|},
                T\fragmentco{\tau(n_S-1)}{|T|}) +
                \edal{\mXpref}(P\fragmentco{\pi(m_S-1)}{|P|},
                T\fragmentco{\tau(\hi')}{y'}),
            \]
            assuming $\mXpref$ aligns $P\fragmentco{\pi(m_S-1)}{|P|}$ onto
            $T\fragmentco{\tau(\hi')}{y'}$
            for some $\hi' \in \fragmentco{0}{n_S}$ with $\hi' \equiv_{\bc(\bG_S)}
            c_{\last}$ and $y'\in \fragment{\tau(\hi')}{|T|}$.
    \end{itemize}

    We proceed to show that $\w_S'$ satisfies properties~\eqref{it:funcover:1},
    \eqref{it:funcover:2}, and~\eqref{it:funcover:3} of \cref{def:funcover}.

    \begin{claim}
        \label{claim:funcover:1}
        The function $\w_S'$ satisfies property~\eqref{it:funcover:1} of
        \cref{def:funcover}.
    \end{claim}
    \begin{claimproof}
        Whenever \( S \) is degenerate, by \cref{lem:periodicity}\eqref{lem:periodicity:i}
        we have \( m_{0} = n_{0} = 1 \), and we do not need to verify
        property~\eqref{it:funcover:1} for $c  = \bc(\bbG_S)-1$. Consequently, it always
        suffices to verify property~\eqref{it:funcover:1}
        only for \( c \in \fragmentco{0}{\bc(\bbG_S)} \) rather than for all \( c \in
        \fragmentco{0}{\bc(\bG_S)} \).

        By \cref{lem:periodicity_gbar}, $\bbG_S(c)$ is connected for all $c \in
        \fragmentco{0}{\bc(\bbG_S)}$.
        Thus, for any
        $c \in \fragmentco{0}{\bc(\bbG_S)}$,  $j\in \fragmentco{0}{m_S-1}$, $i\in
        \fragmentco{0}{n_S-1}$ with $i \equiv_{\bc(\bbG)} j \equiv_{\bc(\bbG)} c$,
        we can construct an alignment of cost at most $\bar{\w}(c)$
        aligning $P\fragmentco{\pi(j)}{\pi(j+1)}$ onto $T\fragmentco{\tau(i)}{\tau(i+1)}$.
        To obtain such alignment, we compose
        the alignments inducing the weights on the edges
        contained in the path in $\bbG_S(c)$ between $P\position{\pi(j)}$ and $T\position{\tau(i)}$.
    \end{claimproof}

    \begin{claim} \label{claim:funcover:2}
        The function $\w_S'$ satisfies
        properties~\eqref{it:funcover:2} and~\eqref{it:funcover:3} of \cref{def:funcover}.
    \end{claim}
    \begin{claimproof}
        Clearly, property~\eqref{it:funcover:2} holds because
        \[
            \w_S'(\bc(\bG_S)-1) \geq \alpha \geq \edal{\mXpref}(P\fragmentco{0}{\pi(0)},
            T\fragmentco{0}{\tau(0)}) \geq \ed(P\fragmentco{0}{\pi(0)},
            T\fragmentco{0}{\tau(0)}).
        \]

        To address property~\eqref{it:funcover:3} of \cref{def:funcover}, we first show
        that the (edge) case \((\pi(0), \tau(0)) \in \mXsuf\), that is \(\hi = 0\), cannot
        occur.
        Observe that \((\pi(0), \tau(0)) \in \mXsuf\) implies \(y_{\mXsuf} = 0\). Together
        with \(y_{\mXsuf}' = n_S\), this yields \(m_S = n_S\).
        Since all alignments \(\mX_{\pref}, \mX_{\suf}, \mA_1, \ldots, \mA_a\) align the
        entire string \(P\) with a fragment of \(T\), it follows that
        $S$ is degenerate.
        Thus, there is nothing to prove, as the quantifier $i$ in
        property~\eqref{it:funcover:3} ranges over an empty set.

        Thus, we can restrict ourselves to the case $\hi > 0$ and that $S$ is
        non-degenerate.
        Choose an arbitrary $i \in \fragmentoo{0}{n_S}$ such that $i \equiv_{\bc(\bG_S)}
        0$.
        Since $\tau(\hi)$ and $\tau(i)$ belong to the same connected component
        $\bbG_S(\bc(\bG_S)-1)$, we can employ a similar argument as in
        \cref{claim:funcover:1}.
        Thus, there exists an alignment $\mY$ with a cost of at most
        $\bar{\w}(\bc(\bG_S)-1)$, aligning $T\fragmentco{\tau(\hi-1)}{\tau(\hi)}$ to
        $T\fragmentco{\tau(i-1)}{\tau(i)}$.
        Assume $\mY(T\fragmentco{y}{\tau(\hi)}) = T\fragmentco{t}{\tau(i)}$ for some $t\in
        \fragment{\tau(i-1)}{\tau(i)}$.
        By composing $\mXsuf : P\fragmentco{0}{\pi(0)} \onto T\fragmentco{y}{\tau(\hi)}$
        and $\mY : T\fragmentco{y}{\tau(\hi)} \onto T\fragmentco{t}{\tau(i)}$, we obtain
        an alignment $P\fragmentco{0}{\pi(0)} \onto T\fragmentco{t}{\tau(i)}$ of a cost of
        at most
        \[
            \edal{\mXsuf}(P\fragmentco{0}{\pi(0)}, T\fragmentco{t}{\tau(\hi)}) +
            \bar{\w}(\bc(\bG_S)-1) \leq \alpha + \bar{\w}(\bc(\bG_S)-1) \leq
            \w_S'(\bc(\bG_S)-1).\claimqedhere
        \]
    \end{claimproof}

    The arguments for properties~\eqref{it:funcover:4} and~\eqref{it:funcover:5} of
    \cref{def:funcover} are almost identical to those used in \cref{claim:funcover:2}.
    Property~\eqref{it:funcover:4} holds because in $\alpha'$ we added the term depending
    on $\mXsuf$. On the other hand, for property~\eqref{it:funcover:5}, one can first rule
    out the case when $(\pi(m_S-1), \tau(n_S-1)) \in \mXpref$. This allows us to restrict
    ourselves to the case $\hi' < n_S-1$ and that $S$ is non-degenerate.
    Property~\eqref{it:funcover:5} follows by composing alignments in
    $\bar{\w}(c_{\last})$ and the alignment that contributes to $\alpha'$.

    \begin{claim}\label{claim:funcover:3}
        $\w_S'(c) \leq \w_S(c)$ for all $c \in \fragmentco{0}{\bc(\bG_S)}$.
    \end{claim}
    \begin{claimproof}
        Note that $\w_S'(c)$ is the sum of $\bar{\w}(c)$ and possibly $\alpha$ and/or
        $\alpha'$. Thus, $\w_S'(c)$ is the sum of costs of partial alignments that can be
        of up to five different types.
        \begin{enumerate}
            \item $\edal{\mX}(P\fragmentco{\pi(j)}{\pi(j+1)},
                T\fragmentco{\tau(i)}{\tau(i+1)})$
                for some distinct triplet $(\mX, i, j)$ such that $\mX \in S$, $j\in
                \fragmentco{0}{m_S-1}$, $i\in \fragmentco{0}{n_S-1}$ and $i
                \equiv_{\bc(\bbG)} j \equiv_{\bc(\bbG)} c$;
                \label{claim:funcover:3:1}
            \item $\edal{\mXpref}(P\fragmentco{0}{\pi(0)}, T\fragmentco{0}{\tau(0)})$;
                \label{claim:funcover:3:2}
            \item $\edal{\mXsuf}(P\fragmentco{0}{\pi(0)}, T\fragmentco{y}{\tau(\hi)})$;
                \label{claim:funcover:3:3}
            \item $\edal{\mXsuf}(P\fragmentco{\pi(m_{S}-1)}{|P|}, T\fragmentco{\tau(n_S-1)}{|T|})$; and
                \label{claim:funcover:3:4}
            \item $\edal{\mXpref}(P\fragmentco{\pi(m_{S}-1)}{|P|}, T\fragmentco{\tau(\hi')}{y'})$.
                \label{claim:funcover:3:5}
        \end{enumerate}
        Note that all partial alignments that are (possibly) involved in the sum are
        disjoint.
        To prove the claim, it suffices to show that
        each edit contained in one of these partial alignments corresponds to a unique
        increase of $\w_S$ by one.
        We derive the following.
        \begin{description}
            \item[\eqref{claim:funcover:3:1}:] Each edit in
                $\edal{\mX}(P\fragmentco{\pi(j)}{\pi(j+1)},
                T\fragmentco{\tau(i)}{\tau(i+1)})$ for some distinct triplet $(\mX, i, j)$
                corresponds to the increase by one of $\w_S(c)$ when that edit in $\mX$ is
                processed.
            \item[\eqref{claim:funcover:3:2} and \eqref{claim:funcover:3:3}:] Each edit in
                $\edal{\mXpref}(P\fragmentco{0}{\pi(0)}, T\fragmentco{0}{\tau(0)})$ or
                $\edal{\mXsuf}(P\fragmentco{0}{\pi(0)}, T\fragmentco{y}{\tau(\hi)})$
                corresponds
                to the increase by one of $\w_S(\bc(\bG_S)-1)$ when that edit in
                $\mXsuf,\mXpref$ is processed, respectively.
            \item[\eqref{claim:funcover:3:4} and \eqref{claim:funcover:3:5}:] Similarly,
                each edit in $\edal{\mXsuf}(P\fragmentco{\pi(m_{S}-1)}{|P|},
                T\fragmentco{\tau(n_S-1)}{|T|})$ or in
                $\edal{\mXpref}(P\fragmentco{\pi(m_{S}-1)}{|P|},
                T\fragmentco{\tau(\hi')}{y'})$ corresponds
                to the increase of $\w_S(c_{\last})$ when that edit in $\mXsuf,\mXpref$ is
                processed, respectively. \claimqedhere
        \end{description}
    \end{claimproof}
    This concludes the proof of \cref{lem:exfuncover}.
\end{proof}

In the remaining part of this section, we let $\w_S$ denote a \emph{weight
function} that covers $S$ of total weight at most $w$.
We do not restrict ourselves to the case $w=\mathcal{O}(\cost(S))$ to be more precise
about the dependency on the total weight of $\w_S$.

\subsection{Period Covers}

In the following subsection, we formalize the characters contained in the black components
that we need to learn.
We specify this through a subset of indices $C_S \subseteq \fragmentco{0}{\bc(\bG_S)}$ of
black components,
which we refer to as \emph{period cover}.

\begin{definition}[{\cite[Definition 4.16]{KNW24}}]\label{def:periodcover_alt}
    A set $C_S \subseteq \fragmentco{0}{\bc(\bG_S)}$ is a \emph{period cover with respect
    to $\w_S$} if $\fragment{a}{b}\subseteq C_S$ holds for every interval
    $\fragment{a}{b}\subseteq \fragmentco{0}{\bc(\bG_S)}$
    that satisfies at least one of the following.
    \begin{enumerate}
        \item $\selfed(T\fragment{\tau(a)}{\tau(b)}) \le 6w+11k$ and $a=0$;
            \label{it:periodcover_alt:1}
        \item $\selfed(T\fragment{\tau(a)}{\tau(b)}) \le 6w+11k$ and $b=\bc(\bG_S)-1$;
            \label{it:periodcover_alt:2}
        \item $\selfed(T\fragment{\tau(a)}{\tau(b)}) \le 6w+11k$ and $b=c_{\last}$;
            \label{it:periodcover_alt:3}
        \item $\selfed(T\fragment{\tau(a)}{\tau(b)}) \le 6w+11k$ and $a=c_{\last}+1$; or
            \label{it:periodcover_alt:4}
        \item $\selfed(T\fragment{\tau(a)}{\tau(b)}) \le 6 \sum_{c=a-1}^{b} \w_S(c)$.
            \label{it:periodcover_alt:5}  \qedhere
    \end{enumerate}
\end{definition}

In the remainder of this (sub)section,
we illustrate that by constructing a period cover $C_S$ in two different ways,
we can efficiently (w.r.t. parameters $w$, $\cost(S)$) encode the information $\{(c,
T[\tau_0^c]):c\in C_S\}$.

First, we establish that this is possible if we construct $C_S$
by straightforwardly verifying whether all intervals $\fragment{a}{b}$
satisfy any of the conditions
\eqref{it:periodcover_alt:1},\eqref{it:periodcover_alt:2},\eqref{it:periodcover_alt:3},\eqref{it:periodcover_alt:4},\eqref{it:periodcover_alt:5}
of \cref{def:periodcover_alt}.

\begin{lemma}[Improvement of {\cite[Lemma 4.17]{KNW24}}] \label{prp:encode_simple_funcover}
    Suppose we are given the information from \cref{prp:enc_gs}. Let
    ${\{\fragment{a_i}{b_i}\}}_{i=0}^{\ell-1}$ be the set of intervals satisfying any of
    the conditions \eqref{it:periodcover_alt:1}--\eqref{it:periodcover_alt:5} of
    \cref{def:periodcover_alt}. Defining the period cover as $C_S \coloneqq
    \bigcup_{i=0}^{\ell-1} \fragment{a_i}{b_i}$, we can encode the set $\{(c, T[\tau(c)])
    \mid c \in C_S\}$ using an additional $\Oh((w+k) \log (1+m|\Sigma| / (w+k)))$ bits.
\end{lemma}

\begin{proof}
    Recall that using the information from \cref{prp:enc_gs}, we can derive all properties
    of the inference graph $\bG_S$, such as the connected component to which a character
    belongs.

    To prove the lemma, we show that it is possible to select a subset of intervals
    indexed by $I \subseteq \fragmentco{0}{\ell}$ such that $\bigcup_{i\in I}
    \fragment{a_i}{b_i} = C_S$ and $\sum_{i \in I}
    \selfed(T\fragment{\tau(a_i)}{\tau(b_i)}) = \Oh(w+k)$.
    This is sufficient to prove that $\{(c, T\position{\tau(c)}) \mid c\in C_S\}$ can be
    encoded with $\Oh((w+k) \log (1+m|\Sigma| / (w+k)))$ bits. Indeed, the interval
    representation $J$ of $C_S$, i.e.,
    the representation of $C_S$ using the smallest number of non-overlapping intervals,
    satisfies
    \[
        \sum_{i \in J} \selfed(T\fragment{\tau(a_i)}{\tau(b_i)}) \leq \sum_{i \in I}
        \selfed(T\fragment{\tau(a_i)}{\tau(b_i)}) = \Oh(w+k),
    \]
    where the first inequality follows by applying sub-additivity and monotonicity
    (\cref{prp:prop_selfed}) repeatedly. Hence, by \cref{prp:self_ed_enc}, we can encode
    $\{(c, T\position{\tau(c)}) \mid c\in C_S\}$ using
    \(
        \Oh\left(D \log(1+|\Sigma||T|/D)\right)
    \)
    bits, where $D \coloneqq \sum_{i \in J} \selfed(T\fragment{\tau(a_i)}{\tau(b_i)}) =
    \Oh(w+k)$. Since $|T| \leq 2m$, this is $\Oh((w+k) \log(1+m|\Sigma|/(w+k)))$.

    Among all intervals that satisfy condition~\eqref{it:periodcover_alt:1} of
    \cref{def:periodcover_alt},
    that is, among all intervals $\fragment{a_i}{b_i}$ such that $a_i=0$, we add to $I$ the
    index $j$ that maximizes $b_j$.
    Clearly, for all other intervals $\fragment{a_i}{b_i}$ such that $a_i=0$,
    we have $\fragment{a_i}{b_i} \subseteq \fragment{a_j}{b_j}$, and we can safely leave
    them out from $I$.
    By using a similar approach for the intervals that satisfy any of the
    conditions~\eqref{it:periodcover_alt:2},~\eqref{it:periodcover_alt:3},~\eqref{it:periodcover_alt:4},
    henceforth, we may assume that we have taken care of all the intervals
    that satisfy any of the
    conditions~\eqref{it:periodcover_alt:1},~\eqref{it:periodcover_alt:2},~\eqref{it:periodcover_alt:3},~\eqref{it:periodcover_alt:4},
    and that ${\{\fragment{a_i}{b_i}\}}_{i=0}^{\ell-1}$ only contains
    intervals that satisfy~\eqref{it:periodcover_alt:5}.

    Now, we proceed by iteratively adding indices to $I$.
    The first index we add to $I$ is $\argmin_i \{a_i\}$.
    Next, we continue to add indices to $I$ based on the last index we added, $j$, as
    follows.
    \begin{itemize}
        \item If $\{i : a_j < a_i \leq b_j < b_i\} \neq \emptyset$, then we add the index
            $\argmax_i \{b_i : a_j < a_i \leq b_j < b_i\}$ to $I$;
        \item otherwise, we add the index $\argmin_i \{a_i : b_j < a_i\}$ to $I$.
    \end{itemize}
    We terminate if we cannot add any index to $I$ anymore, i.e., if $b_j = \max_i
    \{b_i\}$.
    Note, this iterative selection process is guaranteed to end, because
    in every iteration for the $i$ that we add to $I$,
    we have $a_j < a_i$.

    We want to argue that $\bigcup_{i\in I} \fragment{a_i}{b_i} = C_S$.
    To do this, consider an interval $\fragment{x}{y}$ in the interval representation $J$
    of $C_S$.
    We observe that an index $i$ with $a_i = x$ is added either initially or through the
    first rule.
    Subsequently, the selection process applies a combination of the two rules until we
    include an index $i$ with $b_i = y$ into $I$.

    Finally, we want to show that
    $\sum_{i \in I} \selfed(T\fragment{\tau(a_i)}{\tau(b_i)}) \leq 24w$.
    Consider three indices $i,i',i''$ that are added one after the other to $I$.
    If for the selection of $i'$ or $i''$ we need to apply the second rule, then clearly
    $b_i < a_{i''}$.
    Otherwise, if we always apply the first rule, then from $b_{i'} < b_{i''}$ follows
    that $b_i < a_{i''}$,
    because otherwise, we would have selected $i''$ instead of $i'$ in the $\argmin$.
    Hence, for every $c \in \fragmentco{0}{\bc(\bG_S)}$
    it holds $|\{i \in I : c \in \fragment{a_i}{b_i}\}| \leq 2$,
    from which follows $|\{i \in I : c \in \fragment{a_i-1}{b_i}\}| = |\{i \in I : c \in
    \fragment{a_i}{b_i}\}| + |\{i \in I : c + 1 = a_i\}| \leq 4$.
    Note, the previous property also holds when we add less than three indices to $I$.
    Using a double counting argument, we obtain
    \begin{align*}
        \sum_{i \in I} \selfed(T\fragment{\tau(a_i)}{\tau(b_i)})
        \leq 6 \sum_{i \in I} \sum_{c=a_i-1}^{b_i} \w_S(c)
        = 6\sum_{c=0}^{\bc(\bG_S)-1} \sum_{\substack{i \in I \\ c \in \fragment{a_i-1}{b_i}}} \w_S(c)
        \leq 24w,
    \end{align*}
    which concludes the proof.
\end{proof}

A limitation of this approach is the requirement to inspect every character within the
black components before encoding. To address this, in \cite{KNW24} we provide an
alternative construction for a period cover $C_S$
optimized to reduce the total number of string queries. This strategy proves useful within
the quantum setting.

\subsection{Close Candidate Positions}\label{sec:close}

This (sub)section is devoted to proving the following result.

\begin{restatable*}[{\cite[Proposition 4.26]{KNW24}}]{lemma}{prpclose}\label{prp:close}
    Let $\mX : P \onto T\fragmentco{t}{t'}$ be an optimal alignment of $P$ onto a fragment
    $T\fragmentco{t}{t'}$ such that $\ed(P, T\fragmentco{t}{t'})\le k$.
    If there exists $i \in \fragment{0}{n_0-m_0}$ such that $|\tau_i^0 - t - \pi_0^0| \leq
    w+3k$ holds, then the following holds for every $c\in
    \fragmentco{0}{\bc(\bG_S)}\setminus C_S$:
    \begin{enumerate}
        \item\label{it:close:in} $\mX$ aligns $P\position{\pi^c_{j}}$ to
            $T\position{\tau^c_{i+{j}}}$ for every $j \in \fragmentco{0}{m_c}$, and
        \item\label{it:close:out} $\tau^c_{i'}\notin \fragmentco{t}{t'}$ for every $i'\in
            \fragmentco{0}{n_c}\setminus \fragmentco{i}{i+m_c}$.\qedhere
    \end{enumerate}
\end{restatable*}

\subsubsection*{Partition of \texorpdfstring{$P$}{P} and \texorpdfstring{$T$}{T} into Blocks}

\Cref{prp:close} heavily relies on the fact that $P$ and $T$ exhibit a periodic structure.
We divide $P$ and $T$ into the blocks ${\{P_j\}}_{j \in \fragmentco{0}{m_0}}$ and
${\{T_i\}}_{i \in \fragmentco{0}{n_0}}$ having a similar structure.

\begin{definition}[{\cite[Definition 4.20]{KNW24}}]\label{def:blocks}
    For $j \in \fragmentco{0}{m_0}$ and $i \in \fragmentco{0}{n_0}$ we define
    \begin{align*}
        P_j \coloneqq
        \begin{cases}
            P\fragmentco{\pi_j^{0}}{\pi_{j+1}^{0}} & \text{if $j \neq m_{0}-1$,} \\
            P\fragment{\pi_{m_0-1}^{0}}{\pi_{m_0-1}^{c_{\last}}} & \text{otherwise,}
        \end{cases}
        \quad
        \text{ and }
        \quad
         T_i  \coloneqq
        \begin{cases}
            T\fragmentco{\tau^{0}_{i}}{\tau^{0}_{i+1}} & \text{if $i \neq n_{0}-1$,} \\
            T\fragment{\tau^{0}_{n_0-1}}{\tau^{c_{\last}}_{n_0-1}} & \text{otherwise.}
        \end{cases}
        \tag*{\qedhere}
    \end{align*}
\end{definition}

The notion of similarity between these blocks is also captured by the fact that
for any $j,i$ we
can construct an alignment of $P_j$ onto $T_i$ whose cost is upper bounded by the total
weight of the weight function $\w_S$.
More formally, we can prove the following proposition.

\begin{lemma}[{\cite[Proposition 4.21]{KNW24}}]
    \label{prop:asji}
    Let $i \in \fragmentco{0}{n_0}, j \in \fragmentco{0}{m_0}$ be arbitrary such that if
    $i = n_0-1$ then $j = m_0-1$.
    Then, there is an alignment $\mA_{S}^{j,i}$ with the following properties.
    \begin{enumerate}
        \item $\mA_{S}^{j,i} : P_j \onto T_i$ if $j \neq m_0-1$,
            and $\mA_{S}^{j,i} : P_{m_0-1} \onto T\fragment{\tau_i^0}{\tau_i^{c_{\last}}}$ otherwise.
            \label{it:asji:i}
        \item Let $c \in \fragmentco{0}{\bc(\bG_S)}$ if $j \neq m_0-1$, and $c \in
            \fragment{0}{c_{\last}}$ if $j = m_0-1$.
            Then, $\mA_{S}^{j,i}$ matches $P\position{\pi_{j}^{c}}$ and $T\position{\tau_{i}^{c}}$.
            \label{it:asji:ii}
        \item
            $\mA_{S}^{j,i}$ has cost at most $w$.
            \label{it:asji:iii}
    \end{enumerate}
\end{lemma}

\begin{proof}
    Let $c\in \fragmentco{0}{\bc(\bG_S)}$ if $j \neq m_{0}-1$ and let $c \in
    \fragmentco{0}{c_{\last}}$ if $j = m_{0}-1$.
    From \Cref{def:funcover}\eqref{eq:funcover} follows that there is an
    alignment $\mX_c : P\fragmentco{\pi_j^{c}}{\pi_j^{c+1}} \onto T\fragmentco{\tau_i^{c}}{\tau_i^{c+1}}$
    with cost at most $\w_S(c)$.
    Note, since $P\position{\pi_j^{c}} = T\position{\tau_i^{c}}$, we
    may assume that $\mX_c$ matches $P\position{\pi_j^{c}}$ and $T\position{\tau_i^{c}}$.
    Now, it suffices to set
    \[
        \mA_{S}^{j,i} \coloneqq
        \begin{cases}
            \bigcup_{c=0}^{\bc(\bG_S)-1} \mX_c & \text{if $j \neq m_{0}-1$}, \\
            \left(\bigcup_{c=0}^{c_{\last}-1} \mX_c \right) \cup
            \{(\pi_{m_0-1}^{c_{\last}}+1, \tau_i^{c_{\last}}+1)\} & \text{if $j =
            m_{0}-1$}.
        \end{cases}
    \]
    Observe that \( \mA_{S}^{j,i} \) is a valid alignment because the last two characters
    aligned by \( \mX_{c-1} \) are exactly the first two characters aligned by \( \mX_{c}
    \).

    Clearly, \eqref{it:asji:i} holds.
    To verify~\eqref{it:asji:ii}, observe that we may assume each \( \mX_c \) matches in
    its first two characters, since they are equal as  they belong to the same black
    connected component. These are precisely the characters referenced in
    \eqref{it:asji:ii}.
    Regarding~\eqref{it:asji:iii}, note that the cost of \( \mA_{S}^{j,i} \) equals the
    sum of the costs of all \( \mX_{c} \) whose union defines \( \mA_{S}^{j,i} \).
    This sum is upper bounded by the sum of the corresponding \( \w_S(c) \), and is
    therefore at most \( w \).
\end{proof}

In the next lemma, we prove the key property that we use when relating one of the
alignments $\mA_{S}^{j,i}$ and an optimal alignment acting on the same fragments.

\begin{lemma}[{\cite[Lemma 4.22]{KNW24}}]
    \label{prp:coverededit}
    Let $j \in \fragmentco{0}{m_0}$, $i \in \fragmentco{0}{n_0}$ such that if $i = n_0-1$
    then $j = m_0-1$, and let $\mA := \mA_S^{j,i}$.
    Let $P\fragmentco{x}{x'}$ and $T\fragmentco{y}{y'}$ be fragments of $P_j$ and $T_i$,
    respectively, such that $\mA$ aligns $P\fragmentco{x}{x'}$ onto $T\fragmentco{y}{y'}$.

    Let $\mX$ be an optimal alignment of $P\fragmentco{x}{x'}$ onto $T\fragmentco{y}{y'}$.
    If there exists no $(\hx, \hy) \in \mX \cap \mA$ such that both $\mX$ and $\mA$ match
    $P[\hx]$ and $T[\hy]$, then $\{c\in \fragmentco{0}{\bc(\bG_S)} : \pi^c_j \in
    \fragmentco{x}{x'}\} \subseteq C_S$.
\end{lemma}

\begin{proof}
    Denote  $\fragment{a}{b}=\{c\in \fragmentco{0}{\bc(\bG_S)} : \pi^c_j \in
    \fragmentco{x}{x'}\}$
    and assume that this interval is non-empty (otherwise, there is nothing to prove).
    Let $a'\coloneqq a+1$ if $\pi^a_j = x$ and $a'\coloneqq a$ otherwise so that
    $\fragment{a'}{b}=\{c\in \fragmentco{0}{\bc(\bG_S)} : \pi^c_j \in
    \fragmentoo{x}{x'}\}$.
    Observe that $a'>0$ because $P\fragmentco{x}{x'}$ is a fragment of $P_j$ and therefore
    $x \geq \pi_j^0$.
    Now, we can apply the following chain of inequalities
    \begin{align}
            \selfed(T\fragment{\tau_i^a}{\tau_i^b}) &\leq \selfed(T\fragmentco{y}{y'})
            \label{eq:coverededit:1} \\
            &\leq 2 \edal{\mA}(P\fragmentco{x}{x'}, T\fragmentco{y}{y'})
            \label{eq:coverededit:2} \\
            &\leq 2 \edal{\mA}(P\fragmentco{\pi_j^{a'-1}}{\pi_j^{b+1}},
            T\fragmentco{\tau_i^{a'-1}}{\tau_i^{b+1}}) \label{eq:coverededit:3} \\
            &= 2 \sum_{c=a'-1}^{b}\edal{\mA}(P\fragmentco{\pi_j^{c}}{\pi_j^{c+1}},
            T\fragmentco{\tau_i^{c}}{\tau_i^{c+1}}) \label{eq:coverededit:4} \\
            &\leq 2 \sum_{c=a'-1}^{b} \w_S(c) \label{eq:coverededit:5}\\
            &\leq 2 \sum_{c=a-1}^{b} \w_S(c)  \label{eq:coverededit:6}
    \end{align}
    where we have used
    \begin{itemize}
        \item[(\ref{eq:coverededit:1})] monotonicity of $\selfed$
            (\cref{prp:prop_selfed});
        \item[(\ref{eq:coverededit:2})] \cref{prp:edimpliesselfed} and
            $\edal{\mX}(P\fragmentco{x}{x'}, T\fragmentco{y}{y'}) \leq
            \edal{\mA}(P\fragmentco{x}{x'}, T\fragmentco{y}{y'})$;
        \item[(\ref{eq:coverededit:3})] the definition of $\fragment{a'}{b}$ and the fact
            that  $(\pi_j^{a'-1}, \tau_i^{a'-1}),(\pi_j^{b+1}, \tau_i^{b+1}) \in \mA$;
        \item[(\ref{eq:coverededit:4})] the fact that $(\pi_j^{c}, \tau_i^{c}) \in \mA$
            for all $c \in \fragment{a}{b}$; and
        \item[(\ref{eq:coverededit:5})] \cref{def:funcover}.
    \end{itemize}

    Now, notice that
    \[
        \ed(T\fragment{\tau_i^{a}}{\tau_i^{b}}, T\fragment{\tau(a)}{\tau(b)}) \leq
        \sum_{c=a}^{b-1} \ed(T\fragment{\tau_i^{c}}{\tau_i^{c+1}},
        T\fragment{\tau(c)}{\tau(c+1)}) \leq 2\sum_{c=a}^{b-1} \w_S(c)
    \]
    follows by using twice \cref{def:funcover} for each $c \in \fragmentco{a}{b}$.
    From this together with \cref{prp:prop_selfed} and
    $\selfed(T\fragment{\tau_i^{a}}{\tau_i^{b}}) \leq 2 \sum_{c=a-1}^{b} \w_S(c)$ follows
    $\selfed(T\fragment{\tau(a)}{\tau(b)}) \leq 6 \sum_{c=a-1}^{b} \w_S(c)$.
    Hence, we can use \cref{def:periodcover_alt} to deduce $\fragment{a}{b} \subseteq
    C_S$.
\end{proof}

\subsubsection*{Recovering the Edit Distance for a Single Candidate Position}

\Cref{prp:coverededit} allows us to prove that an optimal alignment $P_j$ onto a fragment
$T\fragmentco{y}{y'}$ consistently aligns all characters in black components with each
other (period cover excluded), provided $y$ is close enough to a starting point of any
$T_i$.

\begin{lemma}[{\cite[Lemma 4.23]{KNW24}}]\label{prp:recperioded}
    Let $j\in \fragmentco{0}{m_0}$ and $i\in \fragmentco{0}{n_0}$ be such that $j=m_0-1$
    if $i=n_0-1$.
    Let $\mX: P_j \onto T\fragmentco{y}{y'}$ be an optimal alignment of $P_j$ onto an
    arbitrary fragment $T\fragmentco{y}{y'}$.
    If $|\tau_i^0 - y|  \leq w + 4k$ and $\ed(P_j, T\fragmentco{y}{y'})\le k$,
    then $\mX$ aligns $P[\pi^c_j]$ to $T[\tau^c_{i}]$ for all $c\in
    \fragmentco{0}{\bc(\bG_S)}\setminus C_S$ such that $j\in \fragmentco{0}{m_c}$.
\end{lemma}

\begin{proof}
    We first assume that $j< m_0-1$ and then briefly argue that the case of $j=m_0-1$ can
    be handled similarly.
    Assume $\fragmentco{0}{\bc(\bG_S)}\setminus C_S \neq \emptyset$; otherwise there is
    nothing to prove.
    Denote $c_{\ell}= \min (\fragmentco{0}{\bc(\bG_S)}\setminus C_S)$ and $c_{r} =
    \max(\fragmentco{0}{\bc(\bG_S)}\setminus C_S)$.
    Henceforth, we set $\mA \coloneqq \mA_{S}^{j,i}$.

    \begin{claim}
        There exist $(x_{\ell},y_{\ell}), (x_{r},y_{r}) \in \mA \cap \mX$
        such that $(x_{\ell},y_{\ell}) \leq (\pi_j^{c_{\ell}}, \tau_i^{c_{\ell}})$
        and $(\pi_j^{c_{r}}, \tau_i^{c_{r}}) \leq (x_{r},y_{r})$.
    \end{claim}

    \begin{claimproof}
        First, we want to argue that there exists $(x_{\ell},y_{\ell}) \in \mA \cap \mX$
        such that $(x_{\ell},y_{\ell}) \leq (\pi_j^{c_{\ell}}, \tau_i^{c_{\ell}})$.
        Note that, following \cref{def:periodcover_alt}, we must have
        $\selfed(T\fragment{\tau_0^0}{\tau_0^{c_{\ell}}}) > 6w + 11k$; otherwise
        $c_{\ell}$ would have been included in $C_S$.
        Suppose that $\mX$ aligns $P\fragmentco{\pi_j^0}{\pi_j^{c_{\ell}}}$ with
        $T\fragmentco{y}{\hy}$.
        For the sake of contradiction, suppose that $\mA :
        P\fragmentco{\pi_j^0}{\pi_j^{c_{\ell}}} \onto
        T\fragmentco{\tau_i^0}{\tau_i^{c_{\ell}}}$ and $\mX :
        P\fragmentco{\pi_j^0}{\pi_j^{c_{\ell}}} \onto T\fragmentco{y}{\hy}$ are disjoint.
        Applying \cref{prp:edimpliesselfed} to the two alignments $\mA :
        P\fragmentco{\pi_j^0}{\pi_j^{c_{\ell}}} \onto
        T\fragmentco{\tau_i^0}{\tau_i^{c_{\ell}}}$ and $\mX :
        P\fragmentco{\pi_j^0}{\pi_j^{c_{\ell}}} \onto T\fragmentco{y}{\hy}$ yields
        \begin{align*}
            \MoveEqLeft \selfed(P\fragmentco{\pi_j^0}{\pi_j^{c_{\ell}}}) \\
        &\leq |\tau_i^0 - y| + \edal{\mA}(P\fragmentco{\pi_j^0}{\pi_j^{c_{\ell}}},
        T\fragmentco{\tau_i^0}{\tau_i^{c_{\ell}}}) +
        \edal{\mX}(P\fragmentco{\pi_j^0}{\pi_j^{c_{\ell}}}, T\fragmentco{y}{\hy}) +
        |\tau_i^{c_{\ell}} - \hy| \\
        &\leq 2|\tau_i^0 - y| + 2\edal{\mA}(P\fragmentco{\pi_j^0}{\pi_j^{c_{\ell}}},
        T\fragmentco{\tau_i^0}{\tau_i^{c_{\ell}}}) +
        2\edal{\mX}(P\fragmentco{\pi_j^0}{\pi_j^{c_{\ell}}}, T\fragmentco{y}{\hy}) \\
        &\leq 2w + 8k + 2w + 2k = 4w + 10k,
        \end{align*}
        where we have used
        \begin{align*}
            |\tau_i^{c_{\ell}} - \hy| &\leq  |\tau_i^0 - y| + |(\tau_i^{c_{\ell}} -
            \tau_i^0) - (\pi_j^{c_{\ell}} - \pi_j^0)| + |(\hy - y) - (\pi_j^{c_{\ell}} -
            \pi_j^0)| \\
                                      &\leq |\tau_i^0 - y| +
                                      \edal{\mA}(P\fragmentco{\pi_j^0}{\pi_j^{c_{\ell}}},
                                      T\fragmentco{\tau_i^0}{\tau_i^{c_{\ell}}}) +
                                      \edal{\mX}(P\fragmentco{\pi_j^0}{\pi_j^{c_{\ell}}},
                                      T\fragmentco{y}{\hy}).
        \end{align*}
        Here, the second inequality follows from \cref{obs:drift}.

        An application of \cref{prp:prop_selfed} on
        $P\fragmentco{\pi_j^0}{\pi_j^{c_{\ell}}}$ and
        $\mA^{j,0}_S$ yields
        \[
            \selfed(T\fragment{\tau_0^0}{\tau_0^{c_{\ell}}}) \leq
            \selfed(P\fragmentco{\pi_j^0}{\pi_j^{c_{\ell}}}) + 2w = 6w + 10k.
        \]
        As a consequence, we have
        \[
            \selfed(T\fragment{\tau_0^0}{\tau_0^{c_{\ell}}}) \leq
        \selfed(T\fragmentco{\tau_0^0}{\tau_0^{c_{\ell}}}) + 1 \leq 6w + 11k,\]
        and we obtain a contradiction.

        Since the alignments $\mA : P_j \onto T_{i}$ and $\mX : P_j \onto
        T\fragmentco{y}{y'}$ intersect and have costs
        at most $w$ and $k$, respectively, we conclude from \cref{obs:drift} that
        $|\tau_{i+1}^0 - y'| \le w+k \le w+4k$.
        Using an argument symmetric to the above, there exists $(x_{r},y_{r}) \in \mA \cap
        \mX$ such that $(\pi_j^{c_{r}}, \tau_i^{c_{r}}) \leq (x_{r},y_{r})$.
    \end{claimproof}

    Now, for the sake of contradiction, suppose that there exists $c\in
    \fragmentco{0}{\bc(\bG_S)}\setminus C_S$ such that $\mX$ does not align $P[\pi^c_j]$
    to $T[\tau^c_{i}]$.
    Let $(x_c,y_c) \in \mA \cap \mX$ be the largest $(x_c,y_c)$ such that $(x_c,y_c) \leq
    (\pi^c_j, \tau^c_{i})$, and let $(x_c',y_c') \in \mA \cap \mX$ be the smallest
    $(x_c',y_c')$ such that $(\pi^c_j, \tau^c_{i}) \leq (x_c',y_c')$.
    Note, we know that such $(x_c,y_c), (x_c',y_c')$ exist because $(x_{\ell},y_{\ell}),
    (x_{r},y_{r}) \in \mA \cap \mX$ are such that $(x_{\ell},y_{\ell}) \leq
    (\pi_j^{c_{\ell}}, \tau_i^{c_{\ell}})\leq (\pi^c_j, \tau^c_{i}) \leq (\pi_j^{c_{r}},
    \tau_i^{c_{r}}) \leq (x_{r},y_{r})$.

    \begin{claim}
        $\mX$ and $\mA$ do not both align $P[x_c]$ to $T[y_c]$ at the same time.
    \end{claim}

    \begin{claimproof}
        For the sake of contradiction, assume $\mX$ and $\mA$ both align $P[x_c]$ to
        $T[y_c]$.
        Thus, $(x_c+1,y_c+1) \in \mX \cap \mA$.
        Since $(\pi^c_j, \tau^c_{i}) \notin \mX$, we must have $(x_c,y_c) \neq (\pi^c_j,
        \tau^c_{i})$,
        and $(x_c,y_c)$ appears strictly before $(\pi^c_j, \tau^c_{i}) \in \mA$ in $\mA$.
        Consequently, $(x_c+1,y_c+1) \leq (\pi^c_j, \tau^c_{i})$.
        However, this is a contradiction with $(x_c,y_c)$
        being the largest $(x_c,y_c) \in \mX \cap \mA$ such that $(x_c,y_c) \leq (\pi^c_j,
        \tau^c_{i})$.
    \end{claimproof}

    As a consequence,
    there is no $(\hx, \hy) \in \mA \cap \mX$ such that
    both $\mA : P\fragmentco{x_c}{x_c'} \onto T\fragmentco{y_c}{y_c'}$ and $\mX :
    P\fragmentco{x_c}{x_c'} \onto T\fragmentco{y_c}{y_c'}$
    align $P\position{\hx}$ to $T\position{\hy}$.
    Therefore, we can use \cref{prp:coverededit} obtaining $c \in C_S$.
    However, this is a contradiction with the definition of $c$.

    The proof for $j=m_0-1$ is almost identical to the proof for $j<m_0-1$: in the proof,
    we replace every occurrence of $\fragmentco{0}{\bc(\bG_S)}$  with
    $\fragment{0}{c_{\last}}$ and every occurrence of $\tau_{i+1}^0$ with
    $\tau_i^{c_{\last}}+1$.
\end{proof}

Next, we prove that for an optimal alignment of $P$ onto $T$ it suffices that any $P_j$ is
aligned close enough to any $T_i$ that the consistent alignment of black components as
proven above extends to all blocks before $P_j$ and $T_i$.

\begin{lemma}\label{prp:seq_match}
    Let $\mX : P \onto T\fragmentco{t}{t'}$ be an optimal alignment of $P$ onto a fragment
    $T\fragmentco{t}{t'}$ such that $\ed(P, T\fragmentco{t}{t'})\le k$.
    If there exists $j \in \fragmentco{0}{m_0}$ and $i \in \fragmentco{0}{n_0}$ such that
    $|\tau_{i}^0 - t - \pi_{j}^0| \leq w + 3k$, then $\mX$ aligns
    $P\position{\pi^c_{\hj}}$ to $T\position{\tau^c_{i+{\hj-j}}}$ for every $c\in
    \fragmentco{0}{\bc(\bG_S)}\setminus C_S$ and  $\hj \in \fragmentco{\max(0,
    j-i)\allowbreak }{\min(m_c, n_c+j-i)}$.
\end{lemma}

\begin{proof}
    We assume that $C_S \subsetneq  \fragmentco{0}{\bc(\bG_S)}$; otherwise, there is
    nothing to prove.

    Let $j \in \fragmentco{0}{m_0}$ and $i \in \fragmentco{0}{n_0}$ be such that
    $|\tau_{i}^0 - t - \pi_{j}^0| \leq w + 3k$.
    For each $\hj \in \fragmentco{\max(0, j-i)}{\min(m_c, n_c+j-i)}$, we must show: the
    alignment $\mX$ aligns $P\position{\pi^c_{\hj}}$ to $T\position{\tau^c_{i+{\hj-j}}}$
    for all $c \in \fragmentco{0}{\bc(\bG_S)}\setminus C_S$.
    We begin by proving the case for $\hj = j$. This serves as the base case for an
    inductive argument showing that the statement holds for both $\hj > j$ and $\hj < j$
    (we provide the explicit proof for $\hj > j$, as the case for $\hj < j$ is almost
    identical).

    Before going into the details of the proofs, we set up some notation.
    Define $y_0,\ldots, y_{m_0}\in \fragment{t}{t'}$ so that
    $\mX(P_{\hj})=T\fragmentco{y_{\hj}}{y_{\hj+1}}$
    holds for all $\hj \in \fragmentco{0}{m_0}$.

    Now, we prove the base case when $\hj = j$.
    We observe that $|\pi_{j}^0-(y_{j}-t)|\le k$ because $\mX$ aligns
    $P\fragmentco{0}{\pi_{j}^0}$ with $T\fragmentco{t}{y_j}$ at a cost not exceeding $k$.
    Combining this inequality with the assumption $|\tau_{i}^0 - t - \pi_j^0| \leq w +
    3k$,
    we conclude that $|\tau_{i}^0-y_j|\le w+4k$.
    Thus, we can use \cref{prp:recperioded} to conclude that $\mX$ aligns  $P[\pi_j^c]$ to
    $T[\tau_{i}^c]$ for all $c\in \fragmentco{0}{\bc(\bG_S)}\setminus C_S$ such that $j
    \in \fragmentco{0}{m_c}$.

    For the inductive step consider $\hj \in \fragmentco{0}{m_0}$ such that $\hj > j$.
    Moreover, consider an arbitrary $\bar{c} \in \fragmentco{0}{\bc(\bG_S)}\setminus C_S$.
    By the inductive hypothesis, we have $(\pi^{\bar{c}}_{\hj-1},
    \tau^{\bar{c}}_{i+\hj-1-j}) \in \mX$.
    Combined with $(\pi^{0}_{\hj}, y_{\hj}) \in \mX$, this yields $|(\pi_{\hj}^0 -
    \pi_{\hj-1}^{\bar{c}}) - (y_{\hj} - \tau^{\bar{c}}_{i+\hj-1-j})| \leq k$.
    On the other side, we observe that
    $\ed(P\fragmentco{\pi_{\hj-1}^{\bar{c}}}{\pi_{\hj}^0},T\fragmentco{\tau_{{i+\hj-1-j}}^{\bar{c}}}{\tau_{i+\hj-j}^0})\le
    w$ follows from \cref{def:funcover} by composing alignments of cost at most $\w_S(c)$
    for each $c \in \fragmentco{\bar{c}}{\bc(\bG_S)}$, which means that $|(\pi_{\hj}^0 -
    \pi_{\hj-1}^{\bar{c}})-(\tau_{i+\hj-j}^0 - \tau_{{i+\hj-1-j}}^{\bar{c}})|\le w$.
    Consequently,
    \[
        |\tau_{i+\hj-j}^0 - y_{\hj}| \leq |(\pi_{\hj}^0 - \pi_{\hj-1}^{\bar{c}}) -
        (y_{\hj} - \tau^{\bar{c}}_{i+\hj-1-j})| + |(\tau_{i+\hj-j}^0 -
        \tau_{{i+\hj-1-j}}^{\bar{c}}) - (\pi_{\hj}^0 - \pi_{\hj-1}^{\bar{c}})| \leq w + k.
    \]
    This lets us apply \cref{prp:recperioded} for $\hj$, $i+\hj-j$, and $\mX : P_{\hj}
    \onto T\fragmentco{y_{\hj}}{y_{\hj+1}}$,
    which implies that $\mX$ aligns $P[\pi_{\hj}^c]$ to $T[\tau_{i+\hj-j}^c]$ for all
    $c\in \fragmentco{0}{\bc(\bG_S)}\setminus C_S$ such that $\hj\in \fragmentco{0}{m_c}$,
    completing the inductive argument.
\end{proof}

Finally, to conclude this (sub)section, we prove \cref{prp:close}.

\prpclose
\begin{proof}
    We assume that $C_S \subsetneq  \fragmentco{0}{\bc(\bG_S)}$; otherwise, there is nothing to prove.

    For ~\eqref{it:close:in} notice that it suffices to apply \cref{prp:seq_match},
    so we proceed directly to the proof of~\eqref{it:close:out}.
    First, consider $i'\in \fragmentco{0}{i}$ and suppose, for a proof by contradiction,
    that $\tau_{i'}^c\ge t$.
    \cref{prp:close}\eqref{it:close:in} implies $(\pi_0^c,\tau_i^c)\in \mX$, so $\mX$
    aligns $P\fragmentco{0}{\pi_0^c}$
    onto $T\fragmentco{t}{\tau_i^c}$ at cost at most $k$.
    By \cref{def:funcover}\eqref{it:funcover:3}, there exists an alignment $\mA$ of cost
    at most $w$
    that aligns $P\fragmentco{0}{\pi_0^c}$ onto $T\fragmentco{\hat{t}}{\tau_i^c}$
    for some $\hat{t}\in \fragment{\tau_{i-1}^{\bc(\bG_S)-1}}{\tau_i^0}$.
    Consequently, $|\hat{t}-t|\le w+k$.
    Since $t \le \tau_{i'}^c \le \tau_{i'}^{\bc(\bG_S)-1} \le \tau_{i-1}^{\bc(\bG_S)-1}
    \le \hat{t}$,
    we conclude $|T\fragmentco{\tau_{i'}^c}{\tau_{i'}^{\bc(\bG_S)-1}}|\le \hat{t}-t \le
    w+k$.
    Moreover, $|T\fragmentco{\tau_{0}^c}{\tau_{0}^{\bc(\bG_S)-1}}| \le
    |T\fragmentco{\tau_{i'}^c}{\tau_{i'}^{\bc(\bG_S)-1}}|+\ed(T\fragmentco{\tau_{0}^c}{\tau_{0}^{\bc(\bG_S)-1}},T\fragmentco{\tau_{i'}^c}{\tau_{i'}^{\bc(\bG_S)-1}})
    \le w+k+2w=3w+k$
    and $|T\fragment{\tau_{0}^c}{\tau_{0}^{\bc(\bG_S)-1}}|\le 3w+k+1 \le 3w+2k$.
    Observe that for any string $X$, by only applying insertions/deletions, we have
    $\selfed(X) \leq 2|X|$.
    Consequently, $\selfed(T\fragment{\tau_{0}^c}{\tau_{0}^{\bc(\bG_S)-1}})\le 6w+4k$,
    contradicting $c\notin C_S$.

    The argument for $i'\in \fragmentco{i+m_c}{n_c}$ is fairly similar.
    For a proof by contradiction, suppose that $\tau_{i'}^{c} < t'$.
    \cref{prp:close}\eqref{it:close:in} implies $(\pi_{m_c-1}^c,\tau_{i+m_c-1}^c)\in \mX$,
    so $\mX$ aligns $P\fragmentco{\pi_{m_c-1}^c}{|P|}$
    onto $T\fragmentco{\tau_{i+m_c-1}^c}{t'}$ at cost at most $k$.
    By \cref{def:funcover}\eqref{it:funcover:5}, there exists an alignment $\mA$ of cost
    at most $w$
    that aligns $P\fragmentco{\pi_{m_c-1}^c}{|P|}$ onto
    $T\fragmentco{\tau_{i+m_c-1}^c}{\hat{t}}$
    for some $\hat{t}\in
    \fragment{\tau_{i+m_0-1}^{c_{\last}}}{\tau_{i+m_0-1}^{c_{\last}+1}}$.
    Consequently, $|\hat{t}-t'|\le w+k$. Now, we consider two cases.

    First, suppose that $c \le c_\last$ so that $i'\ge i+m_c = i+m_0$.
    Since $\hat{t} \le \tau_{i+m_0-1}^{c_{\last}+1} \le \tau_{i'}^0 \le \tau_{i'}^c <  t'$,
    we conclude $|T\fragmentco{\tau_{i'}^{0}}{\tau_{i'}^{c}}|\le t'-\hat{t} \le w+k$.
    Moreover, $|T\fragmentco{\tau_{0}^{0}}{\tau_{0}^{c}}| \le
    |T\fragmentco{\tau_{i'}^{0}}{\tau_{i'}^{c}}|+\ed(T\fragmentco{\tau_{0}^{0}}{\tau_{0}^{c}},T\fragmentco{\tau_{i'}^{0}}{\tau_{i'}^{c}})
    \le w+k+2w=3w+k$
    and $|T\fragment{\tau_{0}^{0}}{\tau_{0}^{c}}|\le 3w+k+1 \le 3w+2k$.
    Consequently, $\selfed(T\fragment{\tau_{0}^{0}}{\tau_{0}^{c}})\le 6w+4k$,
    contradicting $c\notin C_S$.

    Next, suppose that $c > c_\last$ so that $i' \ge i+m_c = i+m_0-1$.
    Since $\hat{t} < \tau_{i+m_0-1}^{c_{\last}+1} \le \tau_{i'}^{c_\last+1} \le
    \tau_{i'}^c <  t'$,
    we conclude $|T\fragmentco{\tau_{i'}^{c_\last+1}}{\tau_{i'}^{c}}|\le t'-\hat{t} \le
    w+k$.
    Moreover, $|T\fragmentco{\tau_{0}^{c_\last+1}}{\tau_{0}^{c}}| \le
    |T\fragmentco{\tau_{i'}^{c_\last+1}}{\tau_{i'}^{c}}|+\ed(T\fragmentco{\tau_{0}^{c_\last+1}}{\tau_{0}^{c}},T\fragmentco{\tau_{i'}^{c_\last+1}}{\tau_{i'}^{c}})
    \le w+k+2w=3w+k$
    and $|T\fragment{\tau_{0}^{c_\last+1}}{\tau_{0}^{c}}|\le 3w+k+1 \le 3w+2k$.
    Consequently, $\selfed(T\fragment{\tau_{0}^{c_\last+1}}{\tau_{0}^{c}})\le 6w+4k$,
    contradicting $c\notin C_S$.
\end{proof}

\subsection{Adding Candidate Positions to \texorpdfstring{$S$}{S}}\label{sec:halves}

In this (sub)section we prove the lemma that we use when adding a new alignment to $S$.

\periodhalves

\begin{proof}
    For a proof by contradiction, suppose that $(\pi_j^c,\tau_i^c)\in \mY$ for some $i\in
    \fragmentco{0}{n_c}$ and $j \in \fragmentco{0}{m_c}$.
    Moreover, since $C_S \subsetneq \fragmentco{0}{\bc(\bG_S)}$, we can consider an
    arbitrary $c' \in  \fragmentco{0}{\bc(\bG_S)} \setminus C_S$.
    We consider two cases.

    If $i \geq j$, we note that $\mY$ aligns $P\fragmentco{0}{\pi_j^c}$ to
    $T\fragmentco{t}{\tau_i^c}$. Since $\mY$ has cost at most $k$, we have $|\tau_i^c - t
    - \pi_j^c| \leq k$.
    At the same time,
    $\ed(P\fragmentco{\pi_j^0}{\pi_j^c},T\fragmentco{\tau_i^0}{\tau_i^c})\le w$ follows
    from \cref{def:funcover} by composing alignments of cost at most $\w_S(\bar{c})$ for
    each $\bar{c} \in \fragmentco{0}{c}$,
    which means that $|(\pi_j^c-\pi_j^0)-(\tau_i^c-\tau_i^0)|\le w$.
    Combining the two inequalities, we derive $|\tau_i^0 - t - \pi_j^0| \le w+k$.
    This allows us to use \cref{prp:seq_match} to get that $(\pi_0^{c'},\tau_{i-j}^{c'})
    \in \mY$ because $c' \in  \fragmentco{0}{\bc(\bG_S)} \setminus C_S$.

    Now, since $\mY$ has cost at most $k$, we have $|\tau_{i-j}^{c'} - t - \pi_0^{c'}|
    \leq k$.
    By composing alignments of cost at most $\w_S(\bar{c})$ from
    \cref{def:funcover} for each $\bar{c} \in \fragmentco{0}{c'}$, we get
    $\ed(P\fragmentco{\pi_0^0}{\pi_0^{c'}},T\fragmentco{\tau_{i-j}^{0}}{\tau_{i-j}^{c'}})
    \le w$, which implies $|(\pi_0^{c'}-\pi_0^0)-(\tau_{i-j}^{c'} - \tau_{i-j}^{0})|\le
    w$.
    Combining the previous two inequalities via the triangle inequality, we get
    $|\tau_{i-j}^{0} - t - \pi_0^{0}| \leq w + k$, which contradicts the assumption that
    $|\tau_{i}^{0} - t - \pi_0^{0}|>w+2k$ holds for every $i\in \fragmentco{0}{n_0}$.

    If $i < j$, then we distinguish two further cases.
    If $\pi_0^0 > \tau_0^0 - t$, we use the same argument as before up until the
    application of \cref{prp:seq_match} to get that $(\pi_{j-i}^{c'}, \tau_{0}^{c'}) \in
    \mY$. Thus, $|\tau_0^{c'} - t - \pi_{j-i}^{c'}| \le k$. Using once again the
    alignments from \cref{def:funcover}, and combining through the triangle inequality, we
    get $|\tau_0^0 - t - \pi_{j-i}^0| \le w + k$. This allows us to reach a contradiction
    using the following chain of inequalities:
    \[
        |\tau_0^0 - t - \pi_0^0| = \pi_0^0 - (t - \tau_0^0) \leq \pi_{j-i}^0 - (t -
        \tau_0^0) \leq |\pi_{j-i}^0 - (t - \tau_0^0)| \leq w + k.
    \]
    Otherwise, if $\pi_0^0 \leq \tau_0^0 - t$, we recall that $S$ contains an alignment
    $\mXpref$ of cost at most $k$ such that $(0,0), (\pi_0^0,\tau_{0}^0)\in \mXpref$.
    Thus, $|\tau_0^0 - \pi_0^0| \leq k$, which gives us
    \[
        |\tau_0^0 - t - \pi_0^0| \leq t + |\tau_0^0 - \pi_0^0| = \tau_0^0 - (\tau_0^0 - t)
        + |\tau_0^0 - \pi_0^0| \leq \tau_0^0 - \pi_0^0 + |\tau_0^0 - \pi_0^0| \leq
        2|\tau_0^0 - \pi_0^0| \leq 2k.
    \]
    We conclude that in both cases we have a contradiction with our assumption.
\end{proof}

\subsection{Recovering all \texorpdfstring{$k$}{k}-edit Occurrences of
\texorpdfstring{$P$}{P} in \texorpdfstring{$T$}{T}}
\label{sec:recocc}

In this (sub)section we drop the assumption that $S$ succinctly encloses $T$ and that $\bc(\bG_{S}) > 0$.
Every time we assume that $S$ succinctly encloses $T$ or that $\bc(\bG_{S}) > 0$, we explicitly mention it.

\begin{definition}[Compare with {\cite[Definition 4.28]{KNW24}}] \label{def:scomplete}
Let $S$ be a set of $k$-edit alignments of $P$ onto fragments of $T$
and let $T\fragmentco{t}{t'}$ be a $k$-error occurrence of $P$ in $T$.
We say that $S$ \emph{captures} $T\fragmentco{t}{t'}$ if $S$ succinctly encloses $T$
and exactly one of the two following holds.
\begin{itemize}
    \item We have $\bc(\bG_{S}) = 0$; or
    \item $\bc(\bG_{S}) > 0$ and $|\tau_i^{0} - t - \pi_0^{0}| \leq w + 3k$
    holds for some $i\in \fragmentco{0}{n_0}$.\qedhere
\end{itemize}
\end{definition}

\begin{theorem}[Compare with {\cite[Theorem 4.29]{KNW24}}]\label{prp:subhash}
    Let $S$ be a set of $k$-edit alignments of $P$ onto fragments of $T$ such that $S$
    succinctly encloses $T$ and $\bc(\bG_{S}) > 0$.
    Construct $P^\#$ and $T^\#$ by replacing, for every $c \notin C_S$, every character in
    the $c$-th black component with a unique character $\#_c$. Then, the two following
    hold.
    \begin{enumerate}
        \item For every $a\in \fragmentco{0}{m}$ and $b \in \fragmentco{0}{n}$, if
            $P^\#\position{a} = T^\#\position{b}$, then $P\position{a} = T\position{b}$.\\
            (No new equalities between characters are created.)
            \label{it:ed_subhash:i}
        \item If $S$ captures a $k$-error occurrence $T\fragmentco{t}{t'}$, then we have
            that $\ed(P^\#, T^\#\fragmentco{t}{t'}) \leq \ed(P, T\fragmentco{t}{t'})$.
            Moreover, all optimal alignments $\mX \mid P \onto T\fragmentco{t}{t'}$
            satisfy $\sE_{P, T}(\mX) = \sE_{P^\#, T^\#}(\mX)$.\\
            (For captured $k$-error occurrences, the edit distance and edit information of
            all optimal alignments are preserved.)
            \label{it:ed_subhash:ii}
    \end{enumerate}
\end{theorem}

\begin{proof}
    First, for \eqref{it:ed_subhash:i} observe that we substitute with the same sentinel
    $\#_c$
    characters belonging to the $c$-th black connected component, which we know to be
    uniform, that is,
    containing the same character.

    We proceed with the proof of \eqref{it:ed_subhash:ii}.
    Clearly, when $C_S = \fragmentco{0}{\bc(\bG_S)}$ there is nothing to prove, so assume
    that this is not the case.
    Suppose $\mX : P \onto T\fragmentco{t}{t'}$ is an optimal alignment
    of cost $k' \leq k$. We begin by showing that $\ed(P^\#, T^\#\fragmentco{t}{t'}) \leq
    \ed(P, T\fragmentco{t}{t'})$.

    Since $S$ succinctly encloses $T$ and $\bc(\bG_{S}) > 0$, \cref{def:scomplete} implies
    that there exists $i\in \fragmentco{0}{n_0}$ such that $|\tau_i^0 - t - \pi_0^0| \le w
    + 3k$.
    In the next claim, we prove that \cref{prp:close} applies to $\mX$ because $i \in
    \fragment{0}{n_0-m_0}$.
    \begin{claim}
        We have that $i \le n_0 - m_0$ holds.
    \end{claim}
    \begin{claimproof}
        Recall that $\mXsuf \in S$ has cost at most $k$ and contains both $(\pi_{0}^0,
        \tau_{n_0-m_0}^{0})$ and $(m,n)$.
        Hence, $|(m - \pi_{0}^0) - (n - \tau_{n_0-m_0}^{0})| \leq k$.
        By the definition of $i$ and since $\cost(\mX) \leq k$, we have
        \[
            |(m - \pi_0^0) - (t' - \tau_i^0)| \leq |m - (t' - t)| + |\tau_i^0 - t -
            \pi_0^0| \leq w+3k+k = w+4k.
        \]
        Fix any $c \in \fragmentco{0}{\bc(\bG_S)} \setminus C_S$ and suppose for
        contradiction that $i > n_0 - m_0$. Then
        \begin{multline*}
            |T\fragment{\tau_{n_0-m_0}^{0}}{\tau_{n_0-m_0}^{c}}| =
            \tau_{n_0-m_0}^{c} - \tau_{n_0-m_0}^{0}
            \leq \tau_{i}^{0} - \tau_{n_0-m_0}^{0}
            = |\tau_{i}^{0} - \tau_{n_0-m_0}^{0}|\\
            \leq |(m - \pi_0^0) - (t' - \tau_i^0)| + |(n - \tau_{n_0-m_0}^{0}) - (m - \pi_{0}^0)| \leq w+5k.
        \end{multline*}
        Since \cref{def:funcover} implies $\ed(T\fragment{\tau_0^0}{\tau_0^c},
        T\fragment{\tau_{n_0-m_0}^{0}}{\tau_{n_0-m_0}^{c}})\le 2w$, we get
        \[
            |T\fragment{\tau_0^0}{\tau_0^c}| \leq
            |T\fragment{\tau_{n_0-m_0}^{0}}{\tau_{n_0-m_0}^{c}}| +
            \ed(T\fragment{\tau_0^0}{\tau_0^c},
            T\fragment{\tau_{n_0-m_0}^{0}}{\tau_{n_0-m_0}^{c}}) \leq 3w + 5k.
        \]
        Thus, $\selfed(T\fragment{\tau_0^0}{\tau_0^c}) \leq 6w + 10k$, contradicting $c
        \notin C_S$ in view of \cref{def:periodcover_alt}.
    \end{claimproof}

    By \cref{prp:close}\eqref{it:close:out}, for every $c \in \fragmentco{0}{\bc(\bG_S)}
    \setminus C_S$ we have that
    $\#_c$ appears in $P^{\#}$ at the positions $\Pi_c = \{\pi_j^c \mid j \in
    \fragmentco{0}{m_c}\}$
    and in $T^{\#}\fragmentco{t}{t'}$ at the positions $\Tau_c = \{\tau_{i+j}^c \mid j \in
    \fragmentco{0}{m_c}\}$.

    Observe that the difference in cost between $\mX : P \onto T\fragmentco{t}{t'}$ and
    $\mX : P^{\#} \onto T^{\#}\fragmentco{t}{t'}$ is determined by the characters that
    $\mX$ matches in $P$ and $T$,
    but due to the substitutions they do not match anymore in $P^\#$ and $T^\#$. That is,
    $\edal{\mX}(P^\#, T^\#\fragmentco{t}{t'}) - \ed(P, T\fragmentco{t}{t'}) = |E|$, where
    \[
        E \coloneqq \big\{(x, y) \mid \text{$\mX$ matches $P\position{x}$ with
        $T\position{y}$ but not $P^\#\position{x}$ with $T^\#\position{y}$} \big\},
    \]
    is such that $E \subseteq \big(\textstyle{\bigcup_c \Pi_c \times \bigcup\nolimits_c
    \Tau_c} \big) \cap \mX$.
    By \cref{prp:close}\eqref{it:close:in}, we have that $\mX$ aligns
    $P\position{\pi^c_j}$ to $T\position{\tau^c_{i+j}}$ for every $j \in
    \fragmentco{0}{m_c}$, meaning that
    \[
        E \subseteq \{(\pi^c_j, \tau^c_{i+j}) \mid j \in \fragmentco{0}{m_c}, c \in \fragmentco{0}{\bc(\bG_S)}\}.
    \]
    Moreover, as $\pi^c_j$ and $\tau^c_{i+j}$ are in the same black connected component,
    we have $P^{\#}\position{\pi_j^c} = T^{\#}\position{\tau_{i+j}^c}$.
    Consequently, $E = \emptyset$ and
    \[
        \ed(P^\#, T^\#\fragmentco{t}{t'}) \leq \edal{\mX}(P^\#, T^\#\fragmentco{t}{t'}) = \ed(P, T\fragmentco{t}{t'}).
    \]

    For the second part of \eqref{it:ed_subhash:ii}, it suffices
    to notice that all characters that are substituted in $P,T$
    (or equivalently all the hashes in $P^{\#},T^{\#}$)
    are always matched by such $\mX$.
    Since the characters that are matched by $\mX$
    are stored explicitly neither in $\sE_{P, T}(\mX)$ nor in $\sE_{P^\#, T^\#}(\mX)$,
    we conclude $\sE_{P, T}(\mX) = \sE_{P^\#, T^\#}(\mX)$.
\end{proof}

\subsection{Putting Everything Together}
\label{subsec:together}

Finally, we give the full proof for the construction of $S$ that allows us to prove \cref{thm:ccompl}.

\construction
\begin{proof}
    We begin by verifying that the notation used in \cref{alg:construction_S} is
    well-defined.
    To this end, we must show that at the start of every iteration of the loops in
    Line~\ref{loop:add_A} and Line~\ref{loop:add_B}, the set \( S \) succinctly encloses
    \( T \), since this condition is required for \( S \) to capture a \( k \)-error
    occurrence.
    Additionally, we need to ensure in Line~\ref{alg:construction_S:select_i} that \( m_0
    \geq 3 \), so that the quantifier \( j \in \fragmentco{0}{m_0 - 2} \) is well-defined.
    We prove both of these facts through the next two claims.

    \begin{claim}\label{cl:construction_S:1}
        Before each iteration of the loop in Line~\ref{loop:add_A}, and before the first
        iteration of the loop in Line~\ref{loop:add_B}, the set $S$ succinctly encloses
        $T$.
        Moreover, at the end of the $\ell$-th iteration of the loop at
        Line~\ref{loop:add_A} at least one of $\bc(\bG_S) = 0$, $C_S =
        \fragmentco{0}{\bc(\bG_S)}$, or $s_P \geq 2^{\ell}$ holds, where $s_P$ is the
        minimum number of characters of $P$ in a black component of $\bG_S$.
    \end{claim}
    \begin{claimproof}
        Regarding the first part of the statement of \cref{cl:construction_S:1},
        the conditions for \( S \) succinctly enclosing \( T \) can be easily verified, as
        \( b = 0 \).

        We proceed to prove the second part of the statement. In the $\ell$-th iteration
        of Line~\ref{loop:add_A}, we add to \(S\) an optimal alignment \( \mY \) for a \(
        k \)-error occurrence that \(S\) does not capture.
        Since \( S \) succinctly encloses \( T \) and $C_S \ne
        \fragmentco{0}{\bc(\bG_S)}$, we can apply \cref{lem:periodhalves} to conclude
        there is no $c \in \fragmentco{0}{\bc(\bG_{S})}$ such that $\mY$ aligns
        $P\position{\pi_j^{c}}$ with $T\position{\tau_i^c}$ for some $i\in
        \fragmentco{0}{n_c}$ and $j \in \fragmentco{0}{m_c}$.

        In \( \bG_{S \cup \{\mY\}} \), for every black component $c \in
        \fragmentco{0}{\bc(\bG_{S})}$, the newly added alignment \(\mY\) induces an edge
        incident to \(P\position{\pi_0^c}\).
        Since $\mY$ does not align $P\position{\pi_0^{c}}$ with $T\position{\tau_i^c}$ for
        any $i\in \fragmentco{0}{n_c}$, this edge is either red, making the whole
        component red, or black with its other endpoint outside the component, forcing a
        merge with another (red or black) component.
        Consequently, as long as any black component survives, the value \( s_P \) at
        least doubles.
    \end{claimproof}

    By \cref{cl:construction_S:1}, when we enter the loop at Line~\ref{loop:add_B} for the
    first time, we have \( m_0 \geq s_P \geq 2^2 = 4 \), ensuring that the selection of \(
    j \) at Line~\ref{alg:construction_S:select_i} is well-defined.
    Moreover, unless \( \bc(\bG_S) = 0 \) or $C_S = \fragmentco{0}{\bc(\bG_S)}$, the value
    of \( s_P \) can never decrease.
    Thus, as long as \( S \) continues to succinctly enclose \( T \) in subsequent
    iterations, the selection of \( j \) remains well-defined throughout.

    \begin{claim}\label{cl:construction_S:2}
        Before each iteration in Line~\ref{loop:add_B} the set $S$ succinctly encloses
        $T$.
        Moreover, at the end of the $\ell$-th iteration of the loop at
        Line~\ref{loop:add_B} at least one of $\bc(\bG_S) = 0$, $C_S =
        \fragmentco{0}{\bc(\bG_S)}$, or $s_P \geq 2^{2+\ell}$ holds.
    \end{claim}
    \begin{claimproof}
        If \( \bc(\bG_S) = 0 \), there is nothing to prove.
        We therefore assume \( \bc(\bG_S) \neq 0 \) and focus on showing that the set $S =
        \{\mXpref, \mXsuf, \mA_1, \mA_2, \mB_1, \ldots, \mB_{\ell}\}$
        succinctly encloses \( T \) for \( \ell \geq 1 \).
        The case \( \ell = 0 \) is excluded, as it has already been handled in
        \cref{cl:construction_S:1}.

        Clearly, conditions \eqref{def:enclose:a}, \eqref{def:enclose:b}, and
        \eqref{def:enclose:c} of \cref{def:enclose} are satisfied.
        It remains to verify that condition \eqref{def:enclose:d} of \cref{def:enclose}
        holds for all alignments \( \mB_{1}, \ldots, \mB_{\ell} \). This follows by
        applying \cref{lem:add_iter} iteratively to each alignment \( \mB_{1}, \ldots,
        \mB_{\ell} \).

        For the second part of the claim, let $S_{\ell-1}$ be the set at the beginning of the
        $\ell$-th iteration of the loop at Line~\ref{loop:add_B}. Since the iteration is
        executed, we have $\bc(\bG_{S_{\ell-1}}) \neq 0$ and
        $C_{S_{\ell-1}} \ne \fragmentco{0}{\bc(\bG_{S_{\ell-1}})}$. Let $\mY$ be the uncaptured optimal
        alignment selected in Line~\ref{alg:selstwo}, and let $\mB_\ell=\mY_j$ be the
        partial alignment added to $S_{\ell-1}$ in Line~\ref{alg:construction_S:select_i}.
        By construction, the pre-image of $\mB_\ell$ contains all characters of
        $P_{|S_{\ell-1}}\fragmentco{j \cdot \bc(\bG_{S_{\ell-1}})}{(j+2)\cdot \bc(\bG_{S_{\ell-1}})}$ and hence, in
        particular, at least one character from every black component of $\bG_{S_{\ell-1}}$.

        Moreover, \cref{lem:periodhalves} applied to $\mY$ shows that no edge of $\mY$,
        and therefore no edge of its restriction $\mB_\ell$, joins a character to a text
        character from the same black component of $\bG_{S_{\ell-1}}$. Consequently, every black
        component of $\bG_{S_{\ell-1}}$ either receives a red edge and becomes red, or is joined
        by a black edge to another black component of $\bG_{S_{\ell-1}}$. Hence every surviving
        black component after the $\ell$-th iteration contains at least two black
        components from the beginning of that iteration, so $s_P$ at least doubles.
        Using the induction hypothesis from the end of the $(\ell-1)$-st iteration, we
        conclude that at the end of the $\ell$-th iteration of the loop at
        Line~\ref{loop:add_B}, at least one of $\bc(\bG_S) = 0$,
        $C_S = \fragmentco{0}{\bc(\bG_S)}$, or $m_0 \geq s_P \geq 2^{2 + \ell}$ holds.
    \end{claimproof}
    Since $s_P$ cannot grow indefinitely, \cref{cl:construction_S:2} ensures that the loop
    will eventually terminate (specifically, after $\Oh(\log m)$ iterations). When it
    does, all $k$-error occurrences will be captured.

    Next, we argue why the bound on $\cost(S)$ holds.
    Note that, for each $r \in \{0, 1\}$, the alignments $\mY_j$ satisfying $j \equiv_2 r$
    are disjoint and are all subsets of $\mY$. We obtain that the sum of the costs of the
    alignments $\mY_j$ for all $j \in \fragmentco{0}{m_0 - 2}$ is at most $2k$.
    Consequently, in the $\ell$-th iteration the cost of the alignment $\mY_j$ selected at
    Line~\ref{alg:construction_S:select_i} is at most
    $\cost(\mY_j) \leq 2k / (m_0-2) \leq 2k / (2^{\ell+1} - 2) \leq k / 2^{\ell-1}$. Thus,
    $\cost(S) \leq 4k + k \cdot \sum_{\ell=1}^{\Oh(\log m)} \frac{1}{2^{\ell-1}} \le
    \Oh(k)$.

    Lastly, we give the bound on the encoding size.
    By \cref{prp:enc_gs}, we can encode this information for all $\mX \in S$ such that
    $\cost(\mX) > 0$, which we collect in the set $S^+$, using
    \begin{align*}
        &\sum_{\mX \in S^+} \cost(\mX) \cdot \log \left( \frac{m|\Sigma|}{ \cost (\mX)}
        \right)
        \leq \cost(S) \cdot \log \left( \frac{m|\Sigma|}{k} \right) + \sum_{\mX \in S^+}
        \cost(\mX) \cdot \log \left( \frac{k}{ \cost (\mX)} \right) \\
        &\leq k \log \left(\frac{m|\Sigma|}{k}\right) + 4k + k \cdot
        \sum_{\ell=1}^{\Oh(\log m)} \frac{\ell-1}{2^{\ell-1}}
        \leq \Oh\left(k \log \left(\frac{m|\Sigma|}{k}\right)\right)
        \text{ bits,}
    \end{align*}
    where we use that $x \log (k/x)$ is monotonically increasing for $0 < x \le k/e$. On
    the other side, we bound the number of $\mX$ of cost zero by $|S|$, for which
    $\Oh(\log m)$ bits suffice.

    Lastly, we prove that the returned set satisfies $|S| =
    \Oh(|\floor{\OccE_k(P,T)/k}|)$.
    If no iteration at Line~\ref{loop:add_B} is executed, then there is nothing to prove,
    as $|S| \leq 4$.
    Otherwise, let $\mY_{\ell} : P \onto T\fragmentco{t_{\ell}}{t_{\ell}'}$
    be the selected alignment at Line~\ref{alg:selstwo} in the $\ell$-th iteration of
    Line~\ref{loop:add_B}.
    It suffices to prove that $|t_{\ell} - t_{\ell'}| > k$ for all $\ell' < \ell$, as each
    time we add a new alignment the starting point corresponds to a different element in
    $\floor{\OccE_k(P,T)/k}$.

    Suppose $S_{\ell}$ is $S$ at the point in time just after the $\ell$-th iteration of
    Line~\ref{loop:add_B}, and pick an arbitrary black component $c \in
    \fragmentco{0}{\bc(\bG_{S_{\ell-1}})}$ (note that since we implicitly assume that the
    $\ell$-th iteration is executed, we have $\bc(\bG_{S_{\ell-1}}) > 0$).
    Let $\ell' < \ell$ be arbitrary.
    The $c$-th black component $C_c$ in $\bG_{S_{\ell-1}}$ must be the union of black
    components of $\bG_{S_{\ell'}}$.
    Choose an arbitrary black component $D$ of $\bG_{S_{\ell'}}$ contained in $C_c$.
    Just after adding $\mY_{\ell'}$, the component $D$ is either merged with another black
    component or becomes red.
    The latter is impossible, since then all vertices of $D$ would belong to a red
    component in $\bG_{S_{\ell-1}}$, contradicting $D \subseteq C_c$.
    Hence, $D$ is merged via a black edge induced by $\mY_{\ell'}$ with another black
    component $D'$ of $\bG_{S_{\ell'}}$.
    Moreover, $D'$ must also be contained in $C_c$: otherwise, the black edge of
    $\mY_{\ell'}$ joining $D$ and $D'$ would force $D$ and $D'$ to belong to the same
    black component in $\bG_{S_{\ell-1}}$.
    Consequently, there must be $T\position{\tau_i^c}$ and $P\position{\pi_j^c}$ in $C_c$
    such that $\mY_{\ell'}$ aligns $T\position{\tau_i^c}$ with $P\position{\pi_j^c}$ (here
    $i,j$ are w.r.t. $S_{\ell-1}$).
    Using the contrapositive of \cref{lem:periodhalves} on $\mY_{\ell'}$, this means that
    there is $\hi \in \fragmentco{0}{n_0}$ such that $|\tau_{\hi}^{0} - t_{\ell'} -
    \pi_0^{0}| \leq w+2k$ (again, all $\hi,w,i,j,n_0,m_0$ are w.r.t. $S_{\ell-1}$).
    Therefore, every $k$-error occurrence of $P$ in $T$ with starting point within $k$ of
    $t_{\ell'}$ is captured by $S_{\ell-1}$, because for such a starting point $t$ we have
    $|\tau_{\hi}^{0} - t - \pi_0^{0}| \leq |\tau_{\hi}^{0} - t_{\ell'} - \pi_0^{0}| + |t_{\ell'}-t| \leq w+3k$.
    In particular, the occurrence selected at the $\ell$-th iteration cannot start within
    distance $k$ of $t_{\ell'}$, so $|t_{\ell'} - t_{\ell}| > k$.
\end{proof}

The next construction considers the case where $P$ is close to a highly periodic string.

\constructSper
\begin{proof}
    Let $q_P \in \fragment{m-k}{m+k}$
    be such that $\ed(P, Q^{\infty}\fragmentco{0}{q_P}) \leq 2k$. First, we claim the
    following.

    \begin{claim}
        There is $q_T$ with
        $\ed(T\fragmentco{0}{n}, Q^{\infty}\fragmentco{0}{q_T}) \leq 12k$, $q_T
        \equiv_{|Q|} q_P$ and $q_T \geq q_P$.
    \end{claim}
    \begin{claimproof}
        Let $\ell \in \fragment{n-m-k}{n-m+k}$ be such that $\ed(P,T\fragmentco{\ell}{n})
        \leq k$. By the triangle inequality, we have
        $\ed(T\fragmentco{\ell}{n},Q^{\infty}\fragmentco{0}{q_P}) \leq 3k$.
        Similarly, there is $r \in \fragment{m-k}{m+k}$ such that
        $\ed(P,T\fragmentco{0}{r}) \leq k$. For the same reason as before,
        $\ed(T\fragmentco{0}{r},Q^{\infty}\fragmentco{0}{q_P}) \leq 3k$.

        We apply \cite[Lemma 5.4]{CKW20}%
        \footnote{\cite[Lemma 5.4]{CKW20} states that if
            $\ed(T\fragmentco{0}{q}, Q^{\infty}\fragmentco{x}{y}) \leq k$ and
            $\ed(T\fragmentco{p}{n}, Q^{\infty}\fragmentco{x'}{y'}) \leq k$ hold for some integers
            $p \leq q$, $x \leq y$, and $x' \leq y'$, then, provided that $|Q| = 1$ or $q - p \geq
            (2k+1)|Q|$, we have $\ed(T,Q^{\infty }\fragmentco{x''}{y}) = \ed(T,Q^{\infty
            }\fragmentco{x}{y''}) \leq 2k$ for some $x'' \equiv_{|Q|} x'$ and $y'' \equiv_{|Q|} y'$.}
        on $T$ with threshold $3k$, using the bounds $\ed(T\fragmentco{0}{r},\allowbreak
        Q^{\infty}\fragmentco{0}{q_P}) \leq 3k$ and
        $\ed(T\fragmentco{\ell}{n},Q^{\infty}\fragmentco{0}{q_P}) \leq 3k$.
        Moreover, from $n - \ell \geq m - k$, $r \geq m-k$, and $n \leq 3/2 \cdot m - 2k$,
        we obtain
        $r - \ell \geq 2(m-k)-n \geq m/2 \geq 8k|Q| \geq (6k+1)|Q|$, and hence $r \geq
        \ell$.
        Therefore, the lemma yields $\ed(T,Q^{\infty}\fragmentco{0}{q_T}) \leq 6k$ for
        some $q_T \equiv_{|Q|} q_P$.

        Now, if $q_T \geq q_P$ we are done.
        Otherwise, if $q_T < q_P$, we observe that $|n-q_T| \leq 6k$ and $|r - q_P| \leq
        3k$ because $\ed(T,Q^{\infty}\fragmentco{0}{q_T}) \leq 6k$ and
        $\ed(T\fragmentco{0}{r},Q^{\infty}\fragmentco{0}{q_P}) \leq 3k$.
        Consequently, we have $n-r \leq n-r+q_P-q_T \leq |n-r+q_P-q_T| \leq |n-q_T| + |r -
        q_P| \leq 9k$.
        This means that $\ed(T, Q^{\infty}\fragmentco{0}{q_P}) \leq
        \ed(T\fragmentco{0}{r}, Q^{\infty}\fragmentco{0}{q_P}) + n-r \leq 12k$, and the
        claim follows by actually setting $q_T$ to be $q_P$.
    \end{claimproof}

    Let $\mA : P \onto Q^{\infty}\fragmentco{0}{q_P}$ and  $\mZ :
    Q^{\infty}\fragmentco{0}{q_T} \onto T\fragmentco{0}{n}$ be optimal alignments.
    We define three (not necessarily optimal) alignments:
    \begin{enumerate}
        \item $\mXpref$ obtained by composing $\mA$ and $\mZ$ restricted to
            $Q^{\infty}\fragmentco{0}{q_P} \onto \mZ(Q^{\infty}\fragmentco{0}{q_P})$.
        \item $\mXsuf$ obtained by composing $\mA$
            and $\mZ$ restricted to $Q^{\infty}\fragmentco{q_T-q_P}{q_T} \onto
            \mZ(Q^{\infty}\fragmentco{q_T-q_P}{q_T})$.
            Note that this is a valid composition because $q_T \equiv_{|Q|} q_P$ implies
            $Q^{\infty}\fragmentco{q_T-q_P}{q_T} = Q^{\infty}\fragmentco{0}{q_P}$.
        \item If $q_T = q_P$, then we set $\mXmid = \mXpref$ (and notice that $q_P = q_T$
            also implies $\mXsuf = \mXpref$). Otherwise, define $\mXmid$ as the alignment
            obtained by composing $\mA$
            and $\mZ$ restricted to $Q^{\infty}\fragmentco{|Q|}{|Q|+q_P} \onto
            \mZ(Q^{\infty}\fragmentco{|Q|}{|Q|+q_P})$.
            Notice that this is again a valid composition because
            $Q^{\infty}\fragmentco{|Q|}{|Q|+q_P} = Q^{\infty}\fragmentco{0}{q_P}$.
    \end{enumerate}

    We set $S = \{\mXpref, \mXsuf,\mXmid\}$ and we note that $S$ contains alignments of
    cost at most $12k+2k \leq 14k$.
    Clearly, $S$ succinctly encloses $T$ because $n \leq 3/2 \cdot m - 28k \leq 2m - 2
    \cdot 14k$ and $\mXpref,\mXsuf$ are alignments that align $P$ with a prefix and a
    suffix of $T$, respectively.
    (So $\mA_1 = \mXmid$ and $S$ contains no other alignments.)
    \Cref{prp:enc_gs} and $|S| = 3$ imply that we can encode the desired information of
    the alignments in $S$ using $\Oh(k\log(m|\Sigma|/k))$ bits. It remains to show that
    $S$ captures all $k$-error occurrences.

    Now, if $\bc(\bG_S) = 0$, then $S$ trivially captures all $k$-error occurrences of $P$
    in $T$ and we are done.
    Henceforth, we assume $\bc(\bG_S) > 0$.

    Next define functions $f_{\mA} : \fragmentco{0}{q_P} \rightarrow \fragmentco{0}{|P|}
    \cup \{\bot\}, f_{\mZ} :  \fragmentco{0}{q_T} \rightarrow \fragmentco{0}{|T|} \cup
    \{\bot\}$ as follows.
    For $z \in \fragmentco{0}{q_P}$ set $f_{\mA}(z) = x$, if $\mA$ matches
    $Q^{\infty}\position{z}$ to some $P\position{x}$, otherwise set $f_{\mA}(z) = \bot$.
    Similarly, for $z \in \fragmentco{0}{q_T}$ set $f_{\mZ}(z) = y$, if $\mZ$ matches
    $Q^{\infty}\position{z}$ to some $T\position{y}$, otherwise set $f_{\mZ}(z) = \bot$.

    Next, construct strings $P_Q$ and $T_Q$ of length $q_P$ and $q_T$ as follows. For $z
    \in \fragmentco{0}{q_P}$, if $f_{\mA}(z) \neq \bot$, we set $P_Q\position{z} = \#^c$
    if $P\position{f_{\mA}(z)}$ is contained in the $c$th black component. Otherwise, we
    set $P_Q\position{z} = \$$.
    Similarly, for $z \in \fragmentco{0}{q_T}$, if $f_{\mZ}(z) \neq \bot$, we set
    $T_Q\position{z} = \#^c$ if $T\position{f_{\mZ}(z)}$ is contained in the $c$th black
    component. Otherwise, we set $T_Q\position{z} = \$$.

    We begin by noting that for each black edge $(p,t)$ in $\bG_S$, there must exist $z$
    and $z'$ with $z \equiv_{|Q|} z'$ such that $\mA$ matches $P\position{p}$ with
    $Q^{\infty}\position{z}$, and $\mZ$ matches $Q^{\infty}\position{z'}$ with
    $T\position{t}$. This means that $f_{\mA}(z) = p$ and $f_{\mZ}(z') = t$. Furthermore,
    the inverses $f_{\mA}^{-1}(p) = z$ and $f_{\mZ}^{-1}(t) = z'$ are well-defined for
    such $p$ and $t$.

    Thus, for each character $P\position{p}$ in the $c$-th black component, we have
    $P_Q\position{f^{-1}_{\mA}(p)} = \#^c$, and similarly, for each character
    $T\position{t}$ in the $c$-th black component, we have $T_Q\position{f^{-1}_{\mZ}(t)}
    = \#^c$.
    We observe that this also implies that the character $\#^c$ appears in $P_Q$ and $T_Q$
    only at positions belonging to a single residue class modulo $|Q|$.
    \newcommand{\dol}{\$}
    \newcommand{\hash}{\#}

    \begin{claim}\label{clm:Sper:match}
        Set $\Delta_{\mXpref} = 0$ and $\Delta_{\mXsuf} = q_T - q_P$. Additionally, set
        $\Delta_{\mXmid} = |Q|$ if $q_T > q_P$, and $\Delta_{\mXmid} = 0$ otherwise.
        Then, for every $\mX \in S$ and all $z \in \fragmentco{0}{q_P}$, we have
        $P_Q\position{z} = T_Q\position{z+\Delta_{\mX}}$.
    \end{claim}
    \begin{claimproof}
        Consider any $\mX \in S$ and $z \in \fragmentco{0}{q_P}$.
        We must prove that $P_Q\position{z} = T_Q\position{z+\Delta_{\mX}}$.
        This is immediate if both characters are equal to $\dol$, so we henceforth assume
        otherwise.

        First, suppose that $P_Q\position{z}=\hash^c$ for some black component $c\in
        \fragmentco{0}{\bc(\bG_S)}$.
        Then, for $p\coloneqq f_\mA(z)$, the character $P\position{p}$ belongs to the
        $c$th black component and $\mA$ matches $P\position{p}$ with
        $Q^{\infty}\position{z}$.
        Since $\mX$ is an alignment of the whole pattern $P$, it induces an edge of $\bG_S$ incident to $P\position{p}$.
        As $P\position{p}$ belongs to a black component, this edge cannot be red, and
        hence $\mX$ matches $P\position{p}$ with some character $T\position{t}$.
        Because $\mX$ is obtained by composing $\mA$ and $\mZ$ on the shifted copy of
        $Q^{\infty}$, the alignment $\mZ$ matches $Q^{\infty}\position{z+\Delta_\mX}$ with
        $T\position{t}$.
        Hence, $T_Q\position{z+\Delta_\mX}=\hash^c$.

        The converse implication is symmetric; we provide it for completeness.
        Suppose that $T_Q\position{z+\Delta_\mX}=\hash^c$ for some black component $c\in
        \fragmentco{0}{\bc(\bG_S)}$.
        Then, for $t\coloneqq f_\mZ(z+\Delta_\mX)$, the character $T\position{t}$ belongs
        to the $c$th black component and $\mZ$ matches $Q^{\infty}\position{z+\Delta_\mX}$
        with $T\position{t}$.
        Since $\mX$ is the composition of $\mA$ and $\mZ$ on
        $Q^{\infty}\fragmentco{\Delta_\mX}{q_P+\Delta_\mX}$, the character $T\position{t}$
        lies in the image of $\mX$, and therefore $\mX$ induces an edge of $\bG_S$
        incident to $T\position{t}$.
        As $T\position{t}$ belongs to a black component, this edge cannot be red, and
        hence $\mX$ matches $T\position{t}$ with some character $P\position{p}$.
        By the definition of the composition, $\mA$ matches $P\position{p}$ with
        $Q^{\infty}\position{z}$, and therefore $P_Q\position{z}=\hash^c$.
    \end{claimproof}

    \begin{claim}\label{clm:Sper:match:2}
        There is $q \in \fragmentco{0}{q_P}$ such that $f^{-1}_{\mA}(\pi_0^0) = q$ and
        $f^{-1}_{\mZ}(\tau_0^0) = q$.
        Moreover, if $q_T > q_P$, then this $q$ satisfies $q \in \fragmentco{0}{|Q|}$ and
        also $f^{-1}_{\mA}(\pi_j^0) = j|Q| + q$ and $f^{-1}_{\mZ}(\tau_i^0) = i|Q| + q$
        for all $j \in \fragmentco{0}{m_0}$ and $i \in \fragmentco{0}{n_0}$.
    \end{claim}
    \begin{claimproof}
        Now, if $q_T = q_P$, then $S$ is degenerate, meaning that the $0$th component only
        contains $\pi_0^0$ and $\tau_0^0$. By \cref{clm:Sper:match}, we have $P_Q = T_Q$
        and thus we must have $f^{-1}_{\mA}(\pi_0^0) = f^{-1}_{\mZ}(\tau_0^0)$ because
        these are the only positions in $P_Q$ and $T_Q$ carrying $\#^0$.

        Otherwise, when $q_T > q_P$, we wish to apply \cref{fct:periodicity} to $P_Q$ and
        $T_Q$. To do so, we must first show that $q_T \leq 2q_P + 1$. Since the costs of
        $\mA$ and $\mZ$ are at most $2k$ and $12k$ respectively, it follows that $m \leq
        q_P + 2k$ and $q_T \leq n + 12k$. Consequently, we have $q_T \leq n + 12k \leq
        (3/2 \cdot m - 28k) + 12k \leq (3/2 \cdot q_P + 3k - 28k) + 12k \leq 2q_P + 1$.

        This allows us to apply \cref{fct:periodicity} to conclude that $\Delta_{\mXmid} =
        |Q|$ is a period of $P_Q$ and $T_Q$, since $\{\Delta_{\mXpref}, \Delta_{\mXmid},
        \Delta_{\mXsuf}\} \subseteq \Occ(P_Q,T_Q)$ by \cref{clm:Sper:match}.
        Since $\#^0$ appears in both $P_Q$ and $T_Q$ only at positions belonging to a
        single residue class modulo $|Q|$, and $|Q|$ is a period of both strings, the
        occurrences of $\#^0$ in $P_Q$ and $T_Q$ form arithmetic progressions with
        difference $|Q|$.
        Moreover, the maps $f^{-1}_{\mA}$ and $f^{-1}_{\mZ}$ are well-defined and strictly
        increasing on the characters of the $0$th black component, because these
        characters are matched by $\mA$ and $\mZ$, respectively, and both alignments are
        non-crossing.
        Since $\pi_0^0,\ldots,\pi_{m_0-1}^0$ and $\tau_0^0,\ldots,\tau_{n_0-1}^0$
        enumerate the characters of the $0$th black component in order, there is a unique
        $q \in \fragmentco{0}{|Q|}$ such that
        $f^{-1}_{\mA}(\pi_j^0) = j|Q| + q$ and $f^{-1}_{\mZ}(\tau_i^0) = i|Q| + q$ hold
        for all $j \in \fragmentco{0}{m_0}$ and $i \in \fragmentco{0}{n_0}$.
    \end{claimproof}

    Finally, we prove that $S$ captures all $k$-error occurrences.

    \begin{claim}
        Let $T\fragmentco{\ell}{r}$ be a $k$-error occurrence of $P$ in $T$. Then $S$
        captures $T\fragmentco{\ell}{r}$.
    \end{claim}
    \begin{claimproof}
        Since $\bc(\bG_S) > 0$, we must prove that there exists an index $i$ such that
        $|\tau_i^0 - \ell - \pi_0^0| \leq w + 3 \cdot 14k \leq 84k$.
        (Because $\cost(S) \leq 3 \cdot 14k$, applying the weight function from
        \cref{lem:exfuncover} yields $w = 42k$.)

        Let $q$ be as in \cref{clm:Sper:match:2}.
        Since $(0,0), (\pi_0^0,q)\in \mA$ and $\cost(\mA) \leq 2k$, \cref{obs:drift}
        yields $|q-\pi_0^0| \leq 2k$.
        Similarly, since $(0,0), (q,\tau_0^0)\in \mZ$ and
        $\cost(\mZ) \leq 6k$, \cref{obs:drift} yields $|q-\tau_0^0| \leq 6k$. If further
        $q_T > q_P$ holds, then by the periodicity of $P_Q$ and $T_Q$ we have that
        $|i|Q| + q - \tau_i^0| \leq 6k$ for every $i$ such that $i|Q| + q \in
        \fragmentco{0}{q_T}$.

        We begin by observing that since $T\fragmentco{\ell}{r}$ is a $k$-error occurrence
        of $P$ and $\cost(\mZ) \leq 6k$, we have $|m - (r-\ell)| \leq k$ and $|n - q_T|
        \leq 6k$. Consequently,
        \begin{align}
            \ell \leq r + k - m \leq n + k - m \leq q_T + 7k - m.
            \label{eq:constructSper:1}
        \end{align}

        Next, we perform a case distinction.
        First, suppose $q_T = q_P$.
        Via \cref{eq:constructSper:1} we get
        \[
            |\pi_0^0 - \ell - \tau_0^0| \leq |q - \pi_0^0| + |q - \tau_0^0| + \ell \leq 2k
            + 6k + q_T + 7k - m \leq 15k + |q_T - m| \leq 17k,
        \]
        where we use $|q_P - m| = |q_T - m| \leq 2k$ because $\cost(\mA) \leq 2k$. Hence,
        $S$ captures $T\fragmentco{\ell}{r}$.

        Otherwise, we have $q_T > q_P$.
        Here, we use \cite [Theorem 5.2]{CKW20}%
        \footnote{We use \cite[Theorem 5.2]{CKW20}
            with $d = 2$. It shows that for every $p \in \Occ_k^{E}(P,T)$, we have $p \bmod
            |Q| \leq 6k$ or $p \bmod |Q| \geq |Q|-6k$.}
        to get that $|\ell - i|Q|| \leq 6k$ for
        either $i = \floor{\ell/|Q|}$ or $i = \ceil{\ell/|Q|}$.
        In any case, via \cref{eq:constructSper:1} we have
        \[
            i|Q| + q \leq \ell + 2|Q| \leq q_T + 7k - m + 2|Q| \leq q_T + |Q|(2 + 7k) - m < q_T + 128k|Q| - m \leq q_T,
        \]
        where in the last step we use $|Q| \leq m/128k$.
        This means we indeed have $|i|Q|+q - \tau_i^0| \leq 6k$.
        Putting everything together we obtain
        \begin{align*}
            |\pi_0^0 - \ell - \tau_i^0|
             &\leq |(\tau_i^0 - i|Q|-q) - (\ell - i|Q|) - (\pi_0^0 - q)|\\
             &\leq |\tau_i^0 - i|Q|-q| + |\ell - i|Q|| + |\pi_0^0 - q| \leq 6k+6k+2k \leq 14k,
        \end{align*}
        and we conclude that $S$ captures $T\fragmentco{\ell}{r}$.
    \end{claimproof}
    Therefore, $S$ captures all $k$-error occurrences of $P$ in $T$ which concludes this
    proof.
\end{proof}

This allows us to give a more complete proof of \cref{thm:ccompl}.

\setcounter{mtheorem}{0}
\ccompl

\begin{proof}
    We describe a protocol allowing Alice to encode, using $\Oh(n/m\cdot (k
    \log(m|\Sigma|/k)))$ bits, the edit information for all alignments in
    \[S' \coloneqq \{\mX : P \onto T\fragmentco{t}{t'} \mid t,t' \in \fragment{0}{n},
    \edal{\mX}(P, T\fragmentco{t}{t'}) = \ed(P, T\fragmentco{t}{t'}) \leq k\}.\]

    Since $n \geq m \ge k$, we may assume $k \leq m/200$; otherwise, Alice may simply send
    $P$ and $T$ to Bob. Indeed, this uses $\Oh((m+n)\log |\Sigma|)=\Oh(n\log |\Sigma|)$
    bits, whereas $m/200 < k \le m$ implies $n/m \cdot k \log(m|\Sigma|/k) = \Omega(n \log
    |\Sigma|)$.
    We want to argue that we may further restrict ourselves to the case where $S'$
    encloses $T$.

    \begin{claim}\label{clm:split}
        Suppose $n \leq 4/3 \cdot m$ and there is a $k$-error occurrences of $P$ appearing
        a prefix and suffix of $T$.
        If a protocol using $\Oh(k\log(m|\Sigma|/k))$ bits exists that allows Alice to
        encode the edit information for all alignments in $S'$, then such a protocol also
        exists for the general case (if $n \geq 4/3 \cdot m$ or $P$ does not have a
        $k$-error occurrence as a prefix or suffix~of~$T$).
    \end{claim}

    \begin{claimproof}
        Divide $T$ into $\Oh(n/m)$ contiguous blocks of length $m/3-k = \Omega(m)$ (with
        the last block potentially being shorter).
        For the $i$-th block starting at position $b_i \in \fragmentco{0}{|T|}$, define
        $B_i = \fragment{b_i}{b_i + 4/3 \cdot m} \cap \fragment{0}{|T|}$ and $S_i = \{\mX
        : P \onto T\fragmentco{t}{t'} \mid t,t' \in B_i, \edal{\mX}(P,
    T\fragmentco{t}{t'}) = \ed(P, T\fragmentco{t}{t'}) \leq k\}$.

        We show that there exists such a protocol using $\Oh(k\log(m|\Sigma|/k))$ bits for
        the case where instead of $T$ we consider an arbitrary $B_i$.

        Therefore, let $i$ be arbitrary.
        If $S_i = \emptyset$, then Alice does not need to send anything.
        On the other hand, if $|S_i| \geq 1$, Alice uses the protocol from the assumption
        of the claim on $S_i$ and $T\fragmentco{\ell_i}{r_i}$, where $\ell_i = \min \{t
        \mid \mX : P \onto T\fragmentco{t}{t'} \in S_i\}$ and $r_i = \max \{t' \mid \mX :
        P \onto T\fragmentco{t}{t'} \in S_i\}$.
        This is indeed possible, because $|B_i| \leq 4/3 \cdot m$ and there is a $k$-error
        occurrence appearing as a prefix and suffix.

        Notice that this division scheme does not lose any alignment $\mX \in P \onto
        T\fragmentco{t}{t'}$, or in other words, $S' = \bigcup_{i} S_i$. Indeed, letting
        $b_i$ be such that $t$ is contained in the $i$-th block, we have that $t' \leq
        t+m+k \leq b_i+m+k+m/3-k \leq b_i + 4/3 \cdot m$ and $\fragmentco{t}{t'} \subseteq
        B_i$.
    \end{claimproof}

    Henceforth, $n \leq 4/3 \cdot m$ and there are $k$-error occurrences of $P$ appearing
    as both a prefix and a suffix of $T$.
    Since $k \leq m/200$, we have $28k \leq 7m/50 < m/6$, and therefore $n \leq 4/3 \cdot
    m \leq 3/2 \cdot m - 28k$.
    Alice now constructs a set $S$.
    We distinguish two cases:
    \begin{description}
        \item[If $|\floor{\OccE_k(P,T)/k}| \leq 963068 \cdot k$ or $k \geq 2\log m$] Alice
            constructs $S$ using \cref{alg:construction_S}.
            For this $S$, the set $\{\sE_{P,T}(\mX) : \mX \in S\}$ together with the
            starting/ending points of all $\mX \in S$ can be encoded in
            $\Oh(k\log(m|\Sigma|/k) + |S|\log m)$ bits.
            We proceed to argue that in this case $|S|\log m = \Oh(k \cdot \log(m/k))$.
            To this end, we further distinguish two cases.
            \begin{itemize}
                \item $k \geq 2 \log m$. Then,
                    \[
                        k \log (m/k) \geq 2\log m \cdot \log(m/(2\log m)) = 2\log^2 m -
                        2\log m \cdot \log (2\log m) \geq \Omega(\log^2 m),
                    \]
                    where we use the monotonicity of $x\log(m/x)$ for $0 < x < m/e$.
                    Using $|S| = \Oh(\log m)$ from \cref{alg:construction_S}, we get $|S|
                    \log m = \Oh(\log^2 m) \le \Oh(k \log (m/k))$.
                \item $k < 2 \log m$ and $|\floor{\OccE_k(P,T)/k}| \leq 963068 \cdot k$.
                    From \cref{alg:construction_S} we get $|S| =
                    \Oh(|\floor{\OccE_k(P,T)/k}|) \le \Oh(k)$. This means that $|S| \log m
                    = \Oh(k \log m)$.
                    Further note that $k < 2 \log m$ implies that $\log(m/k) \ge \log m -
                    \log (2\log m)= \Omega(\log m)$, so $|S| \log m = \Oh(k \log m) \le
                    \Oh(k \log (m/k))$.
            \end{itemize}
        \item[If $|\floor{\OccE_k(P,T)/k}| > 963068 \cdot k$] then the second case of
            \cite[Theorem 5.1]{CKW20}%
            \footnote{\cite[Theorem 5.1]{CKW20} states that given
                a pattern $P$ of length $m$, a text $T$ of length $n$, and a positive integer
                $k \leq m$, then at least one of the two following holds: either the $k$-error
                occurrences of $P$ in $T$ satisfy $|\floor{\OccE_k(P,T)/k}| \leq 642045 \cdot n/m
                \cdot k$; or there is a primitive string $Q$ of length $|Q| \leq m/128k$ that satisfies
                $\edp{P}{Q} < 2k$.}
            must hold.
            Thus, there is a primitive string $Q$ of length $|Q| \leq m/128k$ that
            satisfies $\edp{P}{Q} < 2k$.
            This allows Alice to construct a set $S$ using \cref{lem:constructSper}.
            By \cref{lem:constructSper}, the set $\{\sE_{P,T}(\mX) : \mX \in S\}$ together
            with the starting/ending points of all $\mX \in S$ can again be encoded in
            $\Oh(k\log(m|\Sigma|/k))$ bits.
    \end{description}
    Overall, we obtain a set $S$ of $K$-edit alignments, where $K=k$ in the first case and
    $K=14k$ in the second case, which succinctly encloses $T$ and satisfies $C_S =
    \fragmentco{0}{\bc(\bG_S)}$ or captures all $k$-error occurrences of $P$ in $T$, where
    the notion of \emph{capturing} is interpreted with threshold~$K$.

    Now, if $\bc(\bG_{S}) = 0$, then the information that Alice sends to Bob is
    $\{\sE_{P,T}(\mX) : \mX \in S\}$ together with the starting/ending points of all $\mX
    \in S$.
    As already argued before, this allows Bob to fully reconstruct $P$ and $T$.

    Otherwise, if $\bc(\bG_{S}) > 0$, then Alice constructs the function $\w_S$ from
    \cref{alg:construction_w}, applied with threshold $K$, of total weight $w =
    \Oh(\cost(S)) \le \Oh(k)$, and additionally sends the encoding of $\{(c,
    T\position{\tau_0^c}):c\in C_S\}$ to Bob, where $C_S$ is the black cover from
    \cref{prp:encode_simple_funcover}, also applied with threshold $K$.
    This allows Bob to construct the strings $P^\#$ and $T^\#$ as in \cref{prp:subhash}.
    If $C_S = \fragmentco{0}{\bc(\bG_S)}$, then $P^\#=P$ and $T^\# = T$ and we are set.

    Otherwise, by \cref{prp:subhash}\eqref{it:ed_subhash:i}, we have $\ed(P,
    T\fragmentco{t}{t'}) \leq \ed(P^\#, T^\#\fragmentco{t}{t'})$ for all $t,t' \in
    \fragment{0}{n}$. Moreover, every $k$-error occurrence is captured by $S$, so
    \cref{prp:subhash}\eqref{it:ed_subhash:ii}, again applied with threshold $K$, implies
    that the edit distance and edit information of all optimal alignments of cost at most
    $k$ are preserved. Therefore, computing the answer over $P^\#$ and $T^\#$ instead of
    $P$ and $T$ still yields the required information.

    Note that, by \cref{prp:encode_simple_funcover}, sending $\{(c,
    T\position{\tau_0^c}):c\in C_S\}$ requires only $\Oh(k \log (1+m|\Sigma|/k)) \le \Oh(k
    \log (m|\Sigma|/k))$ additional bits, where the last equality follows from $k \leq
    m/200$. This concludes the proof.
\end{proof}

\end{document}